%% file: 0.Main.tex
\documentclass[11pt]{article}
\usepackage{style}
\usepackage{comment}
\begin{document}
	
	\title{Planar Graph Homomorphisms: A Dichotomy and a Barrier from Quantum Groups}
	\author{
		Jin-Yi Cai\thanks{Department of Computer Sciences,
			University of Wisconsin-Madison}\\
		\texttt{jyc@cs.wics.edu}
		\and
		Ashwin Maran\footnotemark[1]\\
		\texttt{amaran@wisc.edu}
		\and
		Ben Young\footnotemark[1]\\
		\texttt{benyoung@cs.wisc.edu}
	}
	\date{}
	\maketitle
	
	\begin{abstract}
		We study the complexity of counting (weighted)
		planar graph homomorphism  problem $\PlGH(M)$
		parametrized by an arbitrary symmetric non-negative real valued matrix $M$.
		For matrices with pairwise distinct diagonal values,
		we prove a complete dichotomy theorem:
		$\PlGH(M)$ is 
		either polynomial-time tractable, or $\#$P-hard, according to a simple criterion.
		More generally, we obtain
		a dichotomy whenever every vertex pair of the graph represented by $M$ 
		can be separated using some planar edge gadget.
		
		A key 
		question in proving complexity dichotomies in the 
		planar setting is the expressive power
		of planar edge gadgets.
		We build on the framework of 
		Man\v{c}inska and Roberson~\cite{manvcinska2020quantum}
		to establish links between
		\emph{planar} edge gadgets and the
		theory of the
		\emph{quantum automorphism group} $\qut(M)$.
		We show that planar edge gadgets that can
		separate vertex pairs of $M$ exist 
		precisely when $\qut(M)$ is \emph{trivial}, 
		and prove that the problem of whether $\qut(M)$ is trivial 
		is undecidable.
		These results delineate the 
		frontier 
		for planar homomorphism counting problems 
		and uncover intrinsic barriers to 
		extending nonplanar reduction techniques 
		to the planar setting.
	\end{abstract}
	\thispagestyle{empty}
	
	\newpage
	
	\tableofcontents

	\thispagestyle{empty}
	\clearpage
	\pagenumbering{arabic}
	
	\input{1.Introduction}
	\input{2.Preliminaries}

	\input{3.Magnitude_Distinct_Hardness}
	\input{4.Diagonal_Distinct_Hardness}
	\input{5.Diagonal_Distinct_Dichotomy}
	\input{6.Quantum}
	\input{7.Future_Work}
	
	\newpage
	\appendix
	\input{A2.0.Model_of_Computation}
	\newpage
	\input{A2.Preliminaries}

	\newpage
	\input{A4.Diagonal_Distinct_Hardness}
	\newpage
	\input{A5.Diagonal_Distinct_Dichotomy}
	\newpage
	\input{A6.Quantum}
	
	\printbibliography
	
\end{document}

%% file: 1.Introduction.tex
\section{Introduction}\label{sec: introduction}

Given graphs $G$ and $H$, a mapping from $V(G)$ to $V(H)$ is called
a \textit{homomorphism} if the edges of $G$ are mapped to the edges of $H$.
More generally, 
let  $M = (M_{ij})$ be a $q \times q$ symmetric matrix.
In this paper, 
we consider arbitrary
$M_{ij} \in \mathbb{R}_{\geq 0}$, where $M_{ij}>0$ is the weight of edge $(i,j)$ in $H = H_M$.
Given $M$,
the partition function  $Z_{M}(G)$ for any input undirected multi-graph $G = (V, E)$
is 
$$Z_{M}(G) = \sum_{\sigma: V \rightarrow [q]} \prod_{(u, v) \in E} M_{\sigma(u) \sigma(v)}.$$
Isomorphic graphs $G \cong G'$ have the same value $Z_{M}(G) = Z_{M}(G')$,  thus every $M$
defines a graph property $Z_{M}(\cdot)$. For a 0-1 matrix $M$,
$Z_{M}(G)$ counts the number of homomorphisms from $G$ to $H$.
Graph homomorphism ($\GH$) encompasses a great deal of graph properties~\cite{lovasz1967operations, lovasz2012large, hell2004graphs}.

Each $M$ defines a
computational problem denoted by $\GH(M)$: $G \mapsto Z_{M}(G)$.
The  complexity of 
$\GH(M)$ has been a major focus of research.
A number of increasingly general complexity dichotomy theorems have been achieved~\cite{dyer2000complexity,bulatov2005complexity,goldberg2010complexity,cai2013graph}.
Graph homomorphism is also
a special case of
counting CSP, where 
a series of results established a complexity dichotomy for any set of constraint functions
${\cal F}$~\cite{bulatov2013complexity,dyer2010complexity,dyer2011decidability,dyer2013complexity,bulatov2012complexity,cai2016nonnegative,cai2017complexity}.

Parallel to this development, Valiant~\cite{valiant2008holographic} introduced
\emph{holographic algorithms}. It is well known that counting the number of perfect
matchings (\#PM) is \#P-complete~\cite{valiant1979complexity}. 
On the other hand,
since the 60's,
there has been a 
famous FKT algorithm~\cite{kasteleyn1961statistics,temperley1961dimer,kasteleyn1963dimer,kasteleyn1967graph} 
that 
can
compute \#PM on planar graphs in P-time. Valiant's holographic algorithms
greatly extended its reach, in fact so much so that a most intriguing question arises: Is this   a \emph{universal} algorithm that 
\emph{every} counting problem expressible as a sum-of-products
that \emph{can be solved} in P-time on planar graphs
(but \#P-hard in general) \emph{is solved} by this method alone?

After a series of work~\cite{cai2009holant,cai2016complete, backens2017new, backens2018complete, yang2022local, fu2019blockwise, fu2014holographic, cai2019holographic}
it was established that for every set of complex valued constraint functions ${\cal F}$
on the Boolean domain (i.e., domain size $q=2$) there is a 3-way exact
classification for  
\#CSP(${\cal F}$): 
 (1) P-time solvable, 
 (2) P-time solvable over planar graphs but \#P-hard over general graphs, 
 (3) \#P-hard over planar graphs.
Moreover,
category (2) consists of precisely  those problems that can be solved by
Valiant's holographic algorithm using FKT.
Extending this even for domain sizes $q= 3$ and $4$ has been difficult~\cite{cai2023complexity,cai2024polynomial},
requiring enormous effort with individually tailored  proof ideas.

This paper makes progress toward this universality in two directions. 
Let $\PlGH(M)$ denote the problem $\GH(M)$ when the
input graphs $G$ are restricted to planar graphs.
In one direction,  we 
prove dichotomies for $\PlGH(M)$, for all domain sizes $q$,
by proving the $\#$P-hardness of several open $\PlGH(M)$.
One such class of matrices for which we prove
$\#$P-hardness is the class of tensor products
of the form $M = \left(\begin{smallmatrix}
    a_{1} & b_{1}\\
    b_{1} & c_{1}\\
\end{smallmatrix} \right) \otimes \dots \otimes
\left(\begin{smallmatrix}
    a_{n} & b_{n}\\
    b_{n} & c_{n}\\
\end{smallmatrix} \right)$,
that are closely tied to tensor products of matchgates,
that belong in category (2) of problems that
are solvable over planar graphs by the holographic algorithm.
At the same time we prove some limitations of current proof techniques
by connecting to the theory of \emph{quantum groups}~\cite{atserias2019quantum,lupini2020nonlocal, manvcinska2020quantum}.

In the setting of general $\GH(M)$,
the classic tool used for establishing
$\#$P-hardness is \emph{gadget construction}.
Given a graph $G$, a new graph $G'$
is constructed from $G$ using some gadget.
By clever choices of $G'$, it can be shown that
computing $Z_{M}(G')$ corresponds to 
computing some $\#$P-hard problem on the instance $G$, 
thus proving that $\GH(M)$ is
$\#$P-hard.

In the $\PlGH$ setting, the main difficulty with
gadget construction is that several
gadgets used to prove
$\#$P-hardness of $\GH(M)$ do not preserve planarity
(i.e., $G$ is planar does not imply $G'$ is planar).
One class of gadgets that does preserve
planarity is the class of \emph{planar edge gadgets}.
Given a planar graph $\mathbf{K}$ with 2
labeled vertices on the external face, 
and any planar graph $G = (V, E)$,
every edge of $G$ can be replaced with the gadget $\mathbf{K}$,
thus preserving planarity of the constructed graph
$\mathbf{K}G$.
The problem of computing $Z_{M}(\mathbf{K}G)$
can then be framed as the problem of computing
$Z_{\mathbf{K}(M)}(G)$ for some other matrix
$\mathbf{K}(M)$,
and if $\PlGH(\mathbf{K}(M))$ can be shown to be
$\#$P-hard, that would imply $\PlGH(M)$ is $\#$P-hard as well.
Cai and Maran~\cite{cai2023complexity,cai2024polynomial}
use this strategy to
prove that for $q \in \{3, 4\}$, $\PlGH(M)$
is $\#$P-hard for $M$ not expressible  (as direct sums and tensor products) from 
categories (1) or (2), thus
establishing the 3-way classification for $q \leq 4$.

For general (non-planar) $\GH(M)$, the  expressive power of edge gadgets is inversely related
to the size of the automorphism group $\aut(M)$ of $M$, following work by Lov{\'a}sz~\cite{lovasz2006rank}.
A key structural observation 
(\cref{proposition: autOrbits}) shows that edge gadgets
cannot distinguish between two vertex labels $i, j$ 
if and only if they belong to the same orbit of 
$\aut(M)$.
This reduces 
questions about general edge gadget expressivity to 
decidable questions
about $\aut(M)$.
Building on the work of Man{\v{c}}inska and Roberson
(\cite{manvcinska2020quantum}) we establish links between
\emph{planar} edge gadgets and the
theory of the
\emph{quantum automorphism group} $\qut(M)$ \cite{lupini2020nonlocal, manvcinska2020quantum}.

Our main technical contribution is a 
dichotomy theorem for $\PlGH(M)$ 
(\cref{theorem: nonNegativeDichotomy}):
For any symmetric non-negative matrix $M$ 
whose diagonal entries are pairwise distinct, 
$\PlGH(M)$ is either polynomial-time computable or 
$\#P$-hard, 
with a simple and easily checkable criterion.
In \cref{sec: allDistinctHardness}, we
consider matrices where all entries 
are pairwise distinct (except $M_{ij} = M_{ji}$).
We
demonstrate how planar edge gadgets and 
gadget interpolation
techniques can be used to amplify
distinctions between different vertex types
and find reductions from known $\#$P-hard
Boolean $2 \times 2$ cases
(see \cref{theorem: magnitudeDistinctHardness}).
In \cref{sec: diagonalDistinctHardness,sec: diagonalDistinctDichotomy},
we sketch proofs of how similar techniques can be used
to establish the complexity dichotomy
when only the diagonal entries are pairwise distinct
(see \cref{theorem: nonNegativeDichotomy}).
In \cref{sec: quantumBody},
we show (\cref{theorem: fullGadgetSeparationNonNegative}) 
that this dichotomy extends to any matrix for which 
planar gadgets can separate its diagonal entries. 
Finally, we relate this combinatorial criterion 
to quantum automorphisms by proving that $\qut(M)$ is trivial 
iff
such a separation is possible
(\cref{corollary: distinct}), and prove
that given an arbitrary $M$ and indices $i,j$, it is \emph{undecidable}
whether there exists some planar edge gadget
$\mathbf{K}$
that can separate the diagonal entries
at those indices
(see \cref{theorem: undec}).
This work establishes the first connection between
dichotomy theorems for counting problems and the theory of quantum groups, and strongly suggests that further progress towards the
universality conjecture requires genuinely new ideas
beyond classical gadget constructions.

%% file: 2.Preliminaries.tex
\section{Preliminaries}\label{sec: preliminaries}

For any integer $q \geq 1$,
let $\Mat{q}{}(X)$ denote the set of
$q \times q$ matrices with entries from
$X \subseteq \mathbb{R}$.
For example, we can have $X = \mathbb{R}$, $\mathbb{R}_{\geq 0}$ or
$\mathbb{R}_{\neq 0}$.
We then let $\Sym{q}{}(X) \subset \Mat{q}{}(X)$
denote the subset of symmetric matrices, 
$\Sym{q}{F}(X) \subset \Sym{q}{}(X)$ 
and $\Sym{q}{pd}(X) \subset \Sym{q}{F}(X)$ denote, respectively, the
subsets of full rank and positive definite symmetric matrices.
We justify this decision to consider arbitrary $\mathbb{R}$-valued
matrices (rather than just algebraic real valued matrices)
in \autoref{appendix: modelComputation}.

\input{2.1.Edge_Gadgets}
\input{2.2.Thickening_Gadgets}
\input{2.3.Stretching_Gadgets}
\input{2.4.Bridge_Gadgets}
\input{2.5.Loop_Gadgets}

%% file: 2.1.Edge_Gadgets.tex
An \emph{edge gadget} is a 2-labeled graph $\mathbf{K} = 
(V(\mathbf{K}), E(\mathbf{K}))$
with 
distinguished vertices $\ell_{1} \neq \ell_{2}
\in V(\mathbf{K})$.
Given any $M \in \Sym{q}{}(\mathbb{R})$,
we define the \emph{signature} $\mathbf{K}(M) 
\in \Mat{q}{}(\mathbb{R})$ as:

\begin{equation}\label{equation: signature}
    \mathbf{K}(M)_{ij} = \sum_{\substack{\tau: 
    V(\mathbf{K}) \rightarrow [q]\\
    \tau(\ell_{1}) = i, \tau(\ell_{2}) = j}}
    \prod_{(u, v) \in E(\mathbf{K})}M_{\tau(u)\tau(v)}.
\end{equation}

In this work, we consider planar edge gadgets whose signatures are symmetric matrices.
\begin{definition}[$\Edge(M)$, $\PlEdge(M)$]
    For $M \in \Sym{q}{}(\mathbb{R})$,
    let $\Edge(M) := \{\mathbf{K}(M):
    \text{edge gadget } \mathbf{K}\} \cap \Sym{q}{}(\mathbb{R})$. Define
    $\PlEdge(M) \subset \Edge(M)$ to the set of \emph{planar} edge
    gadget signatures: those for which there is a planar embedding of $\vk$ in
    which the labeled vertices lie on the outer face.
\end{definition}

Given any planar graph $G = (V, E)$ and $\mathbf{K}$ with $\vk(M) \in \PlEdge(M)$, we construct the planar graph $\mathbf{K}G = (V(\mathbf{K}G), E(\mathbf{K}G))$
by replacing every
edge $(u, v) \in E$ with a copy of
$\mathbf{K}$, and identifying $(u, v)$ with
$\ell_{1}, \ell_{2}$ respectively.
Technically, $\vk G$ is
not yet well-defined as a graph, as its definition depends on
how we name each edge $\{u, v\}$ as $(u,v)$ or $(v,u)$;
however, by the symmetry of $\vk(M)$, the value $Z_M(\vk G)$ will not depend on 
these edge orientations. We have

\begin{equation} \label{eq:signature_matrix}
    Z_M(\vk G)
    = \sum_{\sigma: V \to [q]} \prod_{(u,v) \in E} 
    ~\sum_{\substack{\tau: V(\vk) \to [q] \\ \tau(\ell_1) = \sigma(u),  \tau(\ell_2) = \sigma(v)}}~
    \prod_{(u',v') \in E(\vk)} M_{\tau(u')\tau(v')}
    = Z_{\vk(M)}(G).
\end{equation}

Therefore $\PlGH(\vk(M)) \leq \PlGH(M)$ for all
$\vk(M) \in \PlEdge(M)$.
Some examples follow:

%% file: 2.2.Thickening_Gadgets.tex
\subsection{Thickening Gadgets}\label{sec: thickeningGadgets}
For an integer $n \geq 1$,
the \emph{thickening gadget} $\mathbf{T}_{n}$ consists
of two vertices $\ell_{1}, \ell_{2}$,
that are connected by $n$ parallel edges.
From \cref{equation: signature}, we have that for
all $i, j \in [q]$,
$$\mathbf{T_{n}}(M)_{ij} = 
\sum_{\substack{\tau: V(\mathbf{T_{n}}) \rightarrow [q]\\
\tau(\ell_{1}) = i, \tau(\ell_{2}) = j}}
\prod_{(u, v) \in E(\mathbf{T_{n}})}M_{\tau(u) \tau(v)} 
= \prod_{(u, v) \in E(\mathbf{T_{n}})}M_{ij}
= (M_{ij})^{n}.$$


\begin{figure}
	\centering
	\scalebox{0.8}{\input{images/thickening}}
	\caption{A graph $G$, the thickening gadget $\mathbf{T_{4}}$, and $\mathbf{T_{4}}G$.}
	\label{fig:thickening}
\end{figure}
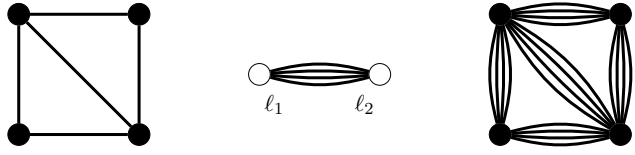

We see that $\mathbf{T_{n}}(M) \in \PlEdge(M)$ for all
$M \in \Sym{q}{}(\mathbb{R})$, for all $n \geq 1$.
By a technique called \emph{gadget interpolation}
we can reduce from more matrices than $\mathbf{T_{n}}(M)$.
The idea 
is that with
oracle access to $\PlGH(M)$ we can compute
$Z_{M}(\mathbf{T_{n}}G) = Z_{\mathbf{T_{n}}(M)}(G)$ 
for 
polynomially many $n$, and then form a Vandermonde system
of linear equations. 
By solving this system, we can compute 
$Z_{N}(G)$ for much more general matrices $N$ (see \autoref{appendix: thickeningAppendix}).
A consequence is the following:

\begin{definition}\label{definition: mathcalT}
    Let $M \in \Sym{q}{}(\mathbb{R}_{> 0})$.
    We may define
    $\mathcal{T}_{M}: \mathbb{R} \rightarrow
    \Sym{q}{}(\mathbb{R}_{> 0})$ such that
    $$\mathcal{T}_{M}(\theta)_{ij} = (M_{ij})^{\theta}, \qquad
    \text{for all } i, j \in [q].$$
\end{definition}

\begin{restatable}{lemma}{simpleThickening}
\label{lemma: simpleThickening}
    Let $M \in \Sym{q}{}(\mathbb{R}_{> 0})$. Then,
    $\PlGH(\mathcal{T}_{M}(\theta)) \leq 
    \PlGH(M)$ for all $\theta \in \mathbb{R}$.
\end{restatable}

The full power of the thickening gadget
allows us to interpolate even more matrices.

\begin{restatable}{definition}{generatingSet}
\label{definition: generatingSet}
    Let $\mathcal{A} \subseteq \mathbb{R}_{\neq 0}$
    be a set of non-zero real numbers.
    A finite set 
    $\{g_{t}\}_{t \in [d]} \in (\mathbb{R}_{> 1})^{d}$,
    for some integer $d \geq 0$, 
    is called a generating set of $\mathcal{A}$ if
    for every $a \in \mathcal{A}$, there exists a
    \emph{unique} $(e_0, e_{1}, \dots, e_{d}) \in \{0, 1\} \times {\mathbb{Z}}^{d}$ 
    such that $a = (-1)^{e_0} {g_{1}^{e_{1}} \cdots g_{d}^{e_{d}}}$.
\end{restatable}

\begin{remark*}
    Note that given any finite $\mathcal{A} \subset 
    \mathbb{Z}_{\neq 0}$,
    the set of prime factors of 
    $\prod_{a \in \mathcal{A}}|a|$
    forms a generating set of $\mathcal{A}$.
    In general, every finite 
    $\mathcal{A} \subset \mathbb{R}_{\neq 0}$ 
    has a generating set (see \cref{lemma: generatingSet} in
    \autoref{appendix: thickeningAppendix}).
\end{remark*}

\begin{definition}\label{definition: mathcalTStar}
    Let $M \in \Sym{q}{}(\mathbb{R}_{\neq 0})$
    with its entries generated by some $\{g_t\}_{t \in [d]}$,
    such that $M_{ij} = (-1)^{e_{ij0}} \cdot
    g_{1}^{e_{ij1}} \cdots g_{d}^{e_{ijd}}$.
    Let $e_{t}^{*} = \min_{i, j \in [q]}e_{ijt}$
    for all $t \in [d]$.
    Define $\mathcal{T}^{*}_{M}: \mathbb{R}^{d} \rightarrow
    \Sym{q}{}(\mathbb{R})$
    such that
    $\mathcal{T}^{*}_{M}(\mathbf{z})_{ij} = 
    (-1)^{e_{ij0}} \cdot z_{1}^{e_{ij1} - e_{1}^{*}} \cdots
    z_{d}^{e_{ijd} - e_{d}^{*}}$
    is a signed monomial in $\mathbf{z} = (z_{1}, \dots, z_{d})$
    for all $i, j \in [q]$.
\end{definition}

\begin{restatable}{lemma}{thickeningLemma}\label{lemma: thickeningLemma}
    Let $M \in \Sym{q}{}(\mathbb{R}_{\neq 0})$
    with its entries generated by some $\{g_t\}_{t \in [d]}$.
    Then, $\PlGH(\mathcal{T}^{*}_{M}(\mathbf{z})) \leq \PlGH(M)$
    for all $\mathbf{z} \in \mathbb{R}^{d}$.
\end{restatable}


\begin{restatable}{lemma}{thickeningSingleVariable}
\label{lemma: thickeningSingleVariable}
    Let $M \in \Sym{q}{}(\mathbb{R}_{> 0})$ with
    its entries generated by some $\{g_t\}_{t \in [d]}$.
    Then, we can find 
    $f_{i}: \mathbb{R} \rightarrow \mathbb{R}$
    for $i \in [d]$, and construct
    $\widehat{\mathcal{T}}_{M}: z \mapsto \mathcal{T}^{*}_{M}(f_{1}(z), \dots,
    f_{d}(z))$, such that
    $\widehat{\mathcal{T}}_{M}(z)_{ij} = z^{x_{ij}}$ for 
    $x_{ij} \in \mathbb{Z}_{\geq 0}$,
    such that for all $(i, j)$, $M_{ij} < M_{i'j'} 
    \iff x_{ij} < x_{i'j'}$, and $M_{ij} = M_{i'j'} 
    \iff x_{ij} = x_{i'j'}$.
    Moreover, if $M \in \Sym{q}{F}(\mathbb{R}_{> 0})$,
    then there exists some $\widehat{z} \in \mathbb{R}$, such that
    $\widehat{\mathcal{T}}_{M}(\widehat{z}) \in 
    \Sym{q}{F}(\mathbb{R}_{> 0})$.
\end{restatable}

A proof as well as a detailed discussion of \cref{lemma: thickeningLemma}
and \cref{lemma: thickeningSingleVariable}
can be found in \autoref{appendix: thickeningAppendix}.

%% file: images/thickening.tex
\begin{tikzpicture}[line join=miter, draw opacity=1]

    \node[circle, draw=black, fill=black] (A) at (-1, -1){};
    \node[circle, draw=black, fill=black] (B) at (-1, 1){};
    \node[circle, draw=black, fill=black] (C) at (1, 1){};
    \node[circle, draw=black, fill=black] (D) at (1, -1){};
    
    \draw[line width=0.5mm, black] (A) -- (B);
    
    \draw[line width=0.5mm, black] (B) -- (C);
    
    \draw[line width=0.5mm, black] (C) -- (D);
    
    \draw[line width=0.5mm, black] (D) -- (A);
    
    \draw[line width=0.5mm, black] (B) -- (D);

\begin{scope}[xshift = 4cm]

    \node[circle, draw=black, fill=white] (A) at (-1, 0){};
    \node[circle, draw=black, fill=white] (B) at (1, 0){};

    \node[draw=none] at (-0.75, -0.5) {$\ell_{1}$};
    \node[draw=none] at (0.75, -0.5) {$\ell_{2}$};
    
    \foreach \angle in {5, 15} {
      \draw[line width=0.5mm] (A) to[bend left=\angle] (B);
      \draw[line width=0.5mm] (A) to[bend right=\angle] (B);
    }
    
\end{scope}


    
    
    
    
    


\begin{scope}[xshift = 8cm]

    \node[circle, draw=black, fill=black] (A) at (-1, -1){};
    \node[circle, draw=black, fill=black] (B) at (-1, 1){};
    \node[circle, draw=black, fill=black] (C) at (1, 1){};
    \node[circle, draw=black, fill=black] (D) at (1, -1){};
    
    \foreach \angle in {5, 15} {
      \draw[line width=0.5mm] (A) to[bend left=\angle] (B);
      \draw[line width=0.5mm] (A) to[bend right=\angle] (B);
    }
    
    \foreach \angle in {5, 15} {
      \draw[line width=0.5mm] (B) to[bend left=\angle] (C);
      \draw[line width=0.5mm] (B) to[bend right=\angle] (C);
    }
    
    \foreach \angle in {5, 15} {
      \draw[line width=0.5mm] (C) to[bend left=\angle] (D);
      \draw[line width=0.5mm] (C) to[bend right=\angle] (D);
    }
    
    \foreach \angle in {5, 15} {
      \draw[line width=0.5mm] (D) to[bend left=\angle] (A);
      \draw[line width=0.5mm] (D) to[bend right=\angle] (A);
    }
    
    \foreach \angle in {5, 15} {
      \draw[line width=0.5mm] (B) to[bend left=\angle] (D);
      \draw[line width=0.5mm] (B) to[bend right=\angle] (D);
    }

\end{scope}

\end{tikzpicture}

%% file: 2.3.Stretching_Gadgets.tex
\subsection{Stretching Gadgets}\label{sec: stretchingGadgets}
For an integer $n \geq 1$, the \emph{stretching gadget}
$\mathbf{S_{n}}$ consists of two vertices
$\ell_{1}, \ell_{2}$, that are connected
by a path of length $n$.
From \cref{equation: signature}, we have that for all
$i, j \in [q]$,
$$\mathbf{S_{n}}(M)_{ij} = 
\sum_{\substack{\tau: V(\mathbf{S_{n}}) \rightarrow [q]\\
\tau(\ell_{1}) = i, \tau(\ell_{2}) = j}}
\prod_{(u, v) \in E(\mathbf{S_{n}})}M_{ij}
= \sum_{k_{1}, \dots, k_{n-1} \in [q]}
M_{ik_{1}}M_{k_{1}k_{2}} \cdots M_{k_{n-1}j}
= (M^{n})_{ij}.$$


\begin{figure}
	\centering
	\scalebox{0.8}{\input{images/stretching}}
	\caption{A graph $G$, the stretching gadget $\mathbf{S_{4}}$, and $\mathbf{S_{4}}G$.}
	\label{fig:stretching}
\end{figure}
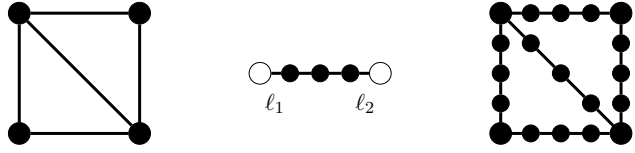

We see that $\mathbf{S_{n}}(M) \in \PlEdge(M)$ for all
$M \in \Sym{q}{}(\mathbb{R})$, for all $n \geq 1$.
Let $M \in \Sym{q}{pd}(\mathbb{R})$
with eigenvalues $\lambda_{1}, \dots, \lambda_{q} > 0$.
There exists an orthonormal matrix 
of unit eigenvectors $H$
such that $M = HDH^{\tt{T}}$, where
$D = \diag(\lambda_{1}, \dots, \lambda_{q})$
is a diagonal matrix.
Note that $M^{n} = HD^{n}H^{\tt{T}}$ for
all $n \geq 1$.
So, 
with oracle access to $\PlGH(M)$, we can compute
$Z_{M}(\mathbf{S_{n}}G) = Z_{\mathbf{S_{n}}(M)}(G)$
for polynomially many $n$.
From this, we obtain a Vandermonde system
and 
this lets us compute $Z_{N}(G)$ for other
matrices $N$ (see \autoref{appendix: stretchingAppendix}).
A consequence is the following:

\begin{definition}\label{definition: mathcalS}
    Let $M = HDH^{\tt{T}} \in \Sym{q}{pd}(\mathbb{R})$.
    We may  define $\mathcal{S}_{M}: \mathbb{R} \rightarrow 
    \Sym{q}{pd}(\mathbb{R})$ such that
    $$\mathcal{S}_{M}(\theta) = HD^{\theta}H^{\tt{T}},~~~~~~\mbox{for all real $\theta \in \mathbb{R}$}.$$
\end{definition}

\begin{restatable}{lemma}{stretchingLemma}\label{lemma: stretchingLemma}
    Let $M = HDH^{\tt{T}} \in \Sym{q}{pd}(\mathbb{R})$.
    Then, $\PlGH(\mathcal{S}_{M}(\theta)) \leq \PlGH(M)$
    for all $\theta \in \mathbb{R}_{\neq 0}$.
\end{restatable}

%% file: images/stretching.tex
\begin{tikzpicture}[line join=miter, draw opacity=1]

    \node[circle, draw=black, fill=black] (A) at (-1, -1){};
    \node[circle, draw=black, fill=black] (B) at (-1, 1){};
    \node[circle, draw=black, fill=black] (C) at (1, 1){};
    \node[circle, draw=black, fill=black] (D) at (1, -1){};
    
    \draw[line width=0.5mm, black] (A) -- (B);
    
    \draw[line width=0.5mm, black] (B) -- (C);
    
    \draw[line width=0.5mm, black] (C) -- (D);
    
    \draw[line width=0.5mm, black] (D) -- (A);
    
    \draw[line width=0.5mm, black] (B) -- (D);

\begin{scope}[xshift = 4cm]

    \node[circle, draw=black, fill=white] (A) at (-1, 0){};
    \node[circle, draw=black, fill=white] (B) at (1, 0){};

    \node[draw=none] at (-0.75, -0.5) {$\ell_{1}$};
    \node[draw=none] at (0.75, -0.5) {$\ell_{2}$};
    
    \node[circle, draw=black, fill=black, inner sep=0pt, minimum size=8] (AB1) at (-0.5, 0){};
    \node[circle, draw=black, fill=black, inner sep=0pt, minimum size=8] (AB2) at (0, 0){};
    \node[circle, draw=black, fill=black, inner sep=0pt, minimum size=8] (AB3) at (0.5, 0){};
    
    \draw[line width=0.5mm, black] (A) -- (AB1);
    \draw[line width=0.5mm, black] (AB1) -- (AB2);
    \draw[line width=0.5mm, black] (AB2) -- (AB3);
    \draw[line width=0.5mm, black] (AB3) -- (B);
    
\end{scope}


    
    
    
    
    
    
    
    
    
    


\begin{scope}[xshift = 8cm]

    \node[circle, draw=black, fill=black] (A) at (-1, -1){};
    \node[circle, draw=black, fill=black] (B) at (-1, 1){};
    \node[circle, draw=black, fill=black] (C) at (1, 1){};
    \node[circle, draw=black, fill=black] (D) at (1, -1){};
    
    \node[circle, draw=black, fill=black, inner sep=0pt, minimum size=8] (AB1) at (-1, -0.5){};
    \node[circle, draw=black, fill=black, inner sep=0pt, minimum size=8] (AB2) at (-1, 0){};
    \node[circle, draw=black, fill=black, inner sep=0pt, minimum size=8] (AB3) at (-1, 0.5){};
    
    \node[circle, draw=black, fill=black, inner sep=0pt, minimum size=8] (BC1) at (-0.5, 1){};
    \node[circle, draw=black, fill=black, inner sep=0pt, minimum size=8] (BC2) at (0, 1){};
    \node[circle, draw=black, fill=black, inner sep=0pt, minimum size=8] (BC3) at (0.5, 1){};
    
    \node[circle, draw=black, fill=black, inner sep=0pt, minimum size=8] (CD1) at (1, 0.5){};
    \node[circle, draw=black, fill=black, inner sep=0pt, minimum size=8] (CD2) at (1, 0){};
    \node[circle, draw=black, fill=black, inner sep=0pt, minimum size=8] (CD3) at (1, -0.5){};
    
    \node[circle, draw=black, fill=black, inner sep=0pt, minimum size=8] (DA1) at (0.5, -1){};
    \node[circle, draw=black, fill=black, inner sep=0pt, minimum size=8] (DA2) at (0, -1){};
    \node[circle, draw=black, fill=black, inner sep=0pt, minimum size=8] (DA3) at (-0.5, -1){};
    
    \node[circle, draw=black, fill=black, inner sep=0pt, minimum size=8] (BD1) at (-0.5, 0.5){};
    \node[circle, draw=black, fill=black, inner sep=0pt, minimum size=8] (BD2) at (0, 0){};
    \node[circle, draw=black, fill=black, inner sep=0pt, minimum size=8] (BD3) at (0.5, -0.5){};

    \draw[line width=0.5mm, black] (A) -- (AB1);
    \draw[line width=0.5mm, black] (AB1) -- (AB2);
    \draw[line width=0.5mm, black] (AB2) -- (AB3);
    \draw[line width=0.5mm, black] (AB3) -- (B);
    
    \draw[line width=0.5mm, black] (B) -- (BC1);
    \draw[line width=0.5mm, black] (BC1) -- (BC2);
    \draw[line width=0.5mm, black] (BC2) -- (BC3);
    \draw[line width=0.5mm, black] (BC3) -- (C);
    
    \draw[line width=0.5mm, black] (C) -- (CD1);
    \draw[line width=0.5mm, black] (CD1) -- (CD2);
    \draw[line width=0.5mm, black] (CD2) -- (CD3);
    \draw[line width=0.5mm, black] (CD3) -- (D);
    
    \draw[line width=0.5mm, black] (D) -- (DA1);
    \draw[line width=0.5mm, black] (DA1) -- (DA2);
    \draw[line width=0.5mm, black] (DA2) -- (DA3);
    \draw[line width=0.5mm, black] (DA3) -- (A);
    
    \draw[line width=0.5mm, black] (B) -- (BD1);
    \draw[line width=0.5mm, black] (BD1) -- (BD2);
    \draw[line width=0.5mm, black] (BD2) -- (BD3);
    \draw[line width=0.5mm, black] (BD3) -- (D);

\end{scope}

\end{tikzpicture}

%% file: 2.4.Bridge_Gadgets.tex
\subsection{Bridge Gadgets}\label{sec: bridgingGadgets}

For an integer $n \geq 1$, the \emph{bridge gadget}
$\mathbf{B_{n}}$ consists of two vertices
$\ell_{1}, \ell_{2}$, that are connected
by a path of length $3$, with the middle
edge replaced with $n$ parallel edges.
From \cref{equation: signature}, we have that for all
$i, j \in [q]$,
$$\mathbf{B_{n}}(M)_{ij} = 
\sum_{\substack{\tau: V(\mathbf{B_{n}}) \rightarrow [q]\\
\tau(\ell_{1}) = i, \tau(\ell_{2}) = j}}
\prod_{(u, v) \in E(\mathbf{B_{n}})}M_{ij}
= \sum_{k_{1}, k_{2} \in [q]}
M_{ik_{1}}(M_{k_{1}k_{2}})^{n}M_{k_{2}j}
= (M \cdot \mathbf{T_{n}}(M) \cdot M)_{ij}.$$
We see that $\mathbf{B_{n}}(M) \in \PlEdge(M)$ for all
$M \in \Sym{q}{}(\mathbb{R})$, for all $n \geq 1$.
With gadget interpolation
(details in \autoref{appendix: bridgingAppendix}), we also have

\begin{figure}
	\centering
	\scalebox{0.8}{\input{images/bridging}}
	\caption{A graph $G$, the bridge gadget $\mathbf{B_{4}}$, and $\mathbf{B_{4}}G$.}
	\label{fig:bridging}
\end{figure}
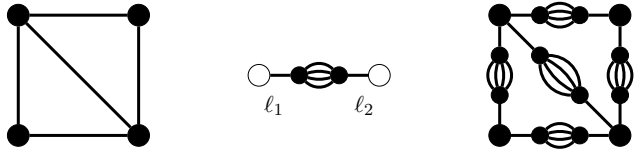


\begin{definition}\label{definition: mathcalB}
    Let $M \in \Sym{q}{}(\mathbb{R}_{\neq 0})$.
    We  define $\mathcal{B}_{M}: \mathbb{R} \rightarrow
    \Sym{q}{}(\mathbb{R})$, such that
    $$\mathcal{B}_{M}(\theta)_{ij} = \sum_{x, y \in [q]}
    M_{ix}(M_{xy})^{\theta}M_{jy},~~~~\mbox{for all $i, j \in [q]$.}$$
\end{definition}

\begin{restatable}{lemma}{bridgingLemma}\label{lemma: bridgingLemma}
    Let $M \in \Sym{q}{}(\mathbb{R}_{\neq 0})$.
    Then, $\PlGH(\mathcal{B}_{M}(\theta)) \leq \PlGH(M)$
    for all $\theta \in \mathbb{R}$.
\end{restatable}

%% file: images/bridging.tex
\begin{tikzpicture}[line join=miter, draw opacity=1]

    \node[circle, draw=black, fill=black] (A) at (-1, -1){};
    \node[circle, draw=black, fill=black] (B) at (-1, 1){};
    \node[circle, draw=black, fill=black] (C) at (1, 1){};
    \node[circle, draw=black, fill=black] (D) at (1, -1){};
    
    \draw[line width=0.5mm, black] (A) -- (B);
    
    \draw[line width=0.5mm, black] (B) -- (C);
    
    \draw[line width=0.5mm, black] (C) -- (D);
    
    \draw[line width=0.5mm, black] (D) -- (A);
    
    \draw[line width=0.5mm, black] (B) -- (D);

\begin{scope}[xshift = 4cm]

    \node[circle, draw=black, fill=white] (A) at (-1, 0){};
    \node[circle, draw=black, fill=white] (B) at (1, 0){};

    \node[draw=none] at (-0.75, -0.5) {$\ell_{1}$};
    \node[draw=none] at (0.75, -0.5) {$\ell_{2}$};
    
    \node[circle, draw=black, fill=black, inner sep=0pt, minimum size=8] (AB1) at (-0.33, 0){};
    \node[circle, draw=black, fill=black, inner sep=0pt, minimum size=8] (AB2) at (0.33, 0){};
    
    \draw[line width=0.5mm, black] (A) -- (AB1);
    \foreach \angle in {15, 45} {
      \draw[line width=0.5mm] (AB1) to[bend left=\angle] (AB2);
      \draw[line width=0.5mm] (AB1) to[bend right=\angle] (AB2);
    }
    \draw[line width=0.5mm, black] (AB2) -- (B);
    
\end{scope}

\begin{scope}[xshift = 8cm]

    \node[circle, draw=black, fill=black] (A) at (-1, -1){};
    \node[circle, draw=black, fill=black] (B) at (-1, 1){};
    \node[circle, draw=black, fill=black] (C) at (1, 1){};
    \node[circle, draw=black, fill=black] (D) at (1, -1){};
    
    \node[circle, draw=black, fill=black, inner sep=0pt, minimum size=8] (AB1) at (-1, -0.33){};
    \node[circle, draw=black, fill=black, inner sep=0pt, minimum size=8] (AB2) at (-1, 0.33){};
    
    \node[circle, draw=black, fill=black, inner sep=0pt, minimum size=8] (BC1) at (-0.33, 1){};
    \node[circle, draw=black, fill=black, inner sep=0pt, minimum size=8] (BC2) at (0.33, 1){};
    
    \node[circle, draw=black, fill=black, inner sep=0pt, minimum size=8] (CD1) at (1, 0.33){};
    \node[circle, draw=black, fill=black, inner sep=0pt, minimum size=8] (CD2) at (1, -0.33){};
    
    \node[circle, draw=black, fill=black, inner sep=0pt, minimum size=8] (DA1) at (0.33, -1){};
    \node[circle, draw=black, fill=black, inner sep=0pt, minimum size=8] (DA2) at (-0.33, -1){};
    
    \node[circle, draw=black, fill=black, inner sep=0pt, minimum size=8] (BD1) at (-0.33, 0.33){};
    \node[circle, draw=black, fill=black, inner sep=0pt, minimum size=8] (BD2) at (0.33, -0.33){};

    \draw[line width=0.5mm, black] (A) -- (AB1);
    \foreach \angle in {15, 45} {
      \draw[line width=0.5mm] (AB1) to[bend left=\angle] (AB2);
      \draw[line width=0.5mm] (AB1) to[bend right=\angle] (AB2);
    }
    \draw[line width=0.5mm, black] (AB2) -- (B);
    
    \draw[line width=0.5mm, black] (B) -- (BC1);
    \foreach \angle in {15, 45} {
      \draw[line width=0.5mm] (BC1) to[bend left=\angle] (BC2);
      \draw[line width=0.5mm] (BC1) to[bend right=\angle] (BC2);
    }
    \draw[line width=0.5mm, black] (BC2) -- (C);
    
    \draw[line width=0.5mm, black] (C) -- (CD1);
    \foreach \angle in {15, 45} {
      \draw[line width=0.5mm] (CD1) to[bend left=\angle] (CD2);
      \draw[line width=0.5mm] (CD1) to[bend right=\angle] (CD2);
    }
    \draw[line width=0.5mm, black] (CD2) -- (D);
    
    \draw[line width=0.5mm, black] (D) -- (DA1);
    \foreach \angle in {15, 45} {
      \draw[line width=0.5mm] (DA1) to[bend left=\angle] (DA2);
      \draw[line width=0.5mm] (DA1) to[bend right=\angle] (DA2);
    }
    \draw[line width=0.5mm, black] (DA2) -- (A);
    
    \draw[line width=0.5mm, black] (B) -- (BD1);
    \foreach \angle in {15, 45} {
      \draw[line width=0.5mm] (BD1) to[bend left=\angle] (BD2);
      \draw[line width=0.5mm] (BD1) to[bend right=\angle] (BD2);
    }
    \draw[line width=0.5mm, black] (BD2) -- (D);

\end{scope}

\end{tikzpicture}

%% file: 2.5.Loop_Gadgets.tex
\subsection{Loop Gadgets}\label{sec: loopGadgets}

For an integer $n \geq 1$, the \emph{loop gadget}
$\mathbf{L_{n}}$ consists of two vertices
$\ell_{1}, \ell_{2}$, that are connected
by a single edge, with $n$ self loops on both
$\ell_{1}$ and $\ell_{2}$.
By \cref{equation: signature}, for all
$i, j \in [q]$,
$$\mathbf{L_{n}}(M)_{ij} = 
\sum_{\substack{\tau: V(\mathbf{L_{n}}) \rightarrow [q]\\
\tau(\ell_{1}) = i, \tau(\ell_{2}) = j}}
\prod_{(u, v) \in E(\mathbf{L_{n}})}M_{ij}
= (M_{ii})^{n} \cdot M_{ij} \cdot (M_{jj})^{n}.$$
We see that $\mathbf{L_{n}}(M) \in \PlEdge(M)$ for all
$M \in \Sym{q}{}(\mathbb{R})$, for all $n \geq 0$.
Gadget interpolation now lets us prove the following
(details in \autoref{appendix: loopAppendix}).

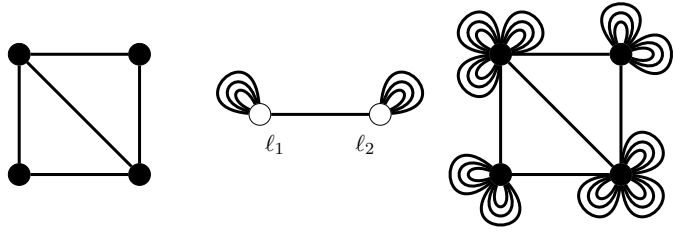
\begin{figure}
	\centering
	\scalebox{0.8}{\input{images/loop}}
	\caption{A graph $G$, the loop gadget $\mathbf{L_{3}}$, and 
    $\mathbf{L_{3}}G$.}
	\label{fig:loop}
\end{figure}

\begin{definition}\label{definition: mathcalL}
    Let $M \in \Sym{q}{}(\mathbb{R}_{> 0})$.
    We define $\mathcal{L}_{M}: \mathbb{R} \rightarrow 
    \Sym{q}{}(\mathbb{R})$:
    $$\mathcal{L}_{M}(\theta) = \mathcal{D}_{M}(\theta)
    \cdot M \cdot \mathcal{D}_{M}(\theta),$$
    where $\mathcal{D}_{M}(\theta)$ is a diagonal matrix with
    $\mathcal{D}_{M}(\theta)_{ii} = (M_{ii})^{\theta}.$
\end{definition}

\begin{restatable}{lemma}{loopLemmaSimple}\label{lemma: loopLemmaSimple}
    Let $M \in \Sym{q}{}(\mathbb{R}_{> 0})$.
    Then $\PlGH(\mathcal{L}_{M}(\theta)) \leq \PlGH(M)$
    for all $\theta \in \mathbb{R}$.
\end{restatable}

As with the thickening gadgets,
generating sets let us interpolate a larger class
of matrices here.

\begin{definition}\label{definition: mathcalLStar}
    Let $M \in \Sym{q}{}(\mathbb{R}_{> 0})$.
    Let $\{g_t\}_{t \in [d]}$ be a generating set for
    $\{M_{ii}: 1 \le i \le q\}$ with
    $M_{ii} = g_{1}^{e_{ii1}} \cdots g_{d}^{e_{iid}}$,
    and  $e_{t}^{*} = \min_{i \in [q]}e_{iit}$
    for all $t \in [d]$.
    We define $\mathcal{L}^{*}_{M}: \mathbb{R}^{d} \rightarrow 
    \Sym{q}{}(\mathbb{R})$:
    $$\mathcal{L}^{*}_{M}(\mathbf{z}) = \mathcal{D}^{*}_{M}(\mathbf{z})
    \cdot M \cdot \mathcal{D}^{*}_{M}(\mathbf{z}),$$
    where $\mathcal{D}^{*}_{M}(\mathbf{z})$ is a diagonal matrix with
    $\mathcal{D}^{*}_{M}(\mathbf{z})_{ii} = z_{1}^{e_{ii1} - e_{1}^{*}} \cdots z_{d}^{e_{iid} - e_{d}^{*}}$.
\end{definition}

\begin{restatable}{lemma}{loopLemma}\label{lemma: loopLemma}
    Let $M \in \Sym{q}{}(\mathbb{R}_{> 0})$ with
    its diagonal entries generated by some $\{g_t\}_{t \in [d]}$.
    Then $\PlGH(\mathcal{L}^{*}_{M}(\mathbf{z})) \leq \PlGH(M)$
    for all $\mathbf{z} \in \mathbb{R}^{d}$.
\end{restatable}

%% file: images/loop.tex
\begin{tikzpicture}[line join=miter, draw opacity=1]

    \node[circle, draw=black, fill=black] (A) at (-1, -1){};
    \node[circle, draw=black, fill=black] (B) at (-1, 1){};
    \node[circle, draw=black, fill=black] (C) at (1, 1){};
    \node[circle, draw=black, fill=black] (D) at (1, -1){};
    
    \draw[line width=0.5mm, black] (A) -- (B);
    
    \draw[line width=0.5mm, black] (B) -- (C);
    
    \draw[line width=0.5mm, black] (C) -- (D);
    
    \draw[line width=0.5mm, black] (D) -- (A);
    
    \draw[line width=0.5mm, black] (B) -- (D);

\begin{scope}[xshift = 4cm]

    \node[circle, draw=black, fill=white] (A) at (-1, 0){};
    \node[circle, draw=black, fill=white] (B) at (1, 0){};

    \node[draw=none] at (-0.75, -0.5) {$\ell_{1}$};
    \node[draw=none] at (0.75, -0.5) {$\ell_{2}$};

    \draw[line width=0.5mm, black] (A) -- (B);

    \draw[line width=0.5mm] (A) .. controls +(115:7mm) and +(155:7mm) .. (A);
    \draw[line width=0.5mm] (A) .. controls +(105:10mm) and +(165:10mm) .. (A);
    \draw[line width=0.5mm] (A) .. controls +(95:14mm) and +(175:14mm) .. (A);

    \draw[line width=0.5mm] (B) .. controls +(25:7mm) and +(65:7mm) .. (B);
    \draw[line width=0.5mm] (B) .. controls +(15:10mm) and +(75:10mm) .. (B);
    \draw[line width=0.5mm] (B) .. controls +(5:14mm) and +(85:14mm) .. (B);

\end{scope}

\begin{scope}[xshift = 8cm]

    \node[circle, draw=black, fill=black] (A) at (-1, -1){};
    \node[circle, draw=black, fill=black] (B) at (-1, 1){};
    \node[circle, draw=black, fill=black] (C) at (1, 1){};
    \node[circle, draw=black, fill=black] (D) at (1, -1){};
    
    \draw[line width=0.5mm, black] (A) -- (B);
    
    \draw[line width=0.5mm, black] (B) -- (C);
    
    \draw[line width=0.5mm, black] (C) -- (D);
    
    \draw[line width=0.5mm, black] (D) -- (A);
    
    \draw[line width=0.5mm, black] (B) -- (D);

    \draw[line width=0.5mm] (A) .. controls +(160:7mm) and +(200:7mm) .. (A);
    \draw[line width=0.5mm] (A) .. controls +(150:10mm) and +(210:10mm) .. (A);
    \draw[line width=0.5mm] (A) .. controls +(140:14mm) and +(220:14mm) .. (A);
    
    \draw[line width=0.5mm] (A) .. controls +(250:7mm) and +(290:7mm) .. (A);
    \draw[line width=0.5mm] (A) .. controls +(240:10mm) and +(300:10mm) .. (A);
    \draw[line width=0.5mm] (A) .. controls +(230:14mm) and +(310:14mm) .. (A);
    
    \draw[line width=0.5mm] (B) .. controls +(205:7mm) and +(245:7mm) .. (B);
    \draw[line width=0.5mm] (B) .. controls +(195:10mm) and +(255:10mm) .. (B);
    \draw[line width=0.5mm] (B) .. controls +(185:14mm) and +(265:14mm) .. (B);
    
    \draw[line width=0.5mm] (B) .. controls +(115:7mm) and +(155:7mm) .. (B);
    \draw[line width=0.5mm] (B) .. controls +(105:10mm) and +(165:10mm) .. (B);
    \draw[line width=0.5mm] (B) .. controls +(95:14mm) and +(175:14mm) .. (B);
    
    \draw[line width=0.5mm] (B) .. controls +(25:7mm) and +(65:7mm) .. (B);
    \draw[line width=0.5mm] (B) .. controls +(15:10mm) and +(75:10mm) .. (B);
    \draw[line width=0.5mm] (B) .. controls +(5:14mm) and +(85:14mm) .. (B);
    
    \draw[line width=0.5mm] (C) .. controls +(70:7mm) and +(110:7mm) .. (C);
    \draw[line width=0.5mm] (C) .. controls +(60:10mm) and +(120:10mm) .. (C);
    \draw[line width=0.5mm] (C) .. controls +(50:14mm) and +(130:14mm) .. (C);
    
    \draw[line width=0.5mm] (C) .. controls +(-20:7mm) and +(20:7mm) .. (C);
    \draw[line width=0.5mm] (C) .. controls +(-30:10mm) and +(30:10mm) .. (C);
    \draw[line width=0.5mm] (C) .. controls +(-40:14mm) and +(40:14mm) .. (C);
    
    \draw[line width=0.5mm] (D) .. controls +(205:7mm) and +(245:7mm) .. (D);
    \draw[line width=0.5mm] (D) .. controls +(195:10mm) and +(255:10mm) .. (D);
    \draw[line width=0.5mm] (D) .. controls +(185:14mm) and +(265:14mm) .. (D);
    
    \draw[line width=0.5mm] (D) .. controls +(295:7mm) and +(335:7mm) .. (D);
    \draw[line width=0.5mm] (D) .. controls +(285:10mm) and +(345:10mm) .. (D);
    \draw[line width=0.5mm] (D) .. controls +(275:14mm) and +(355:14mm) .. (D);
    
    \draw[line width=0.5mm] (D) .. controls +(25:7mm) and +(65:7mm) .. (D);
    \draw[line width=0.5mm] (D) .. controls +(15:10mm) and +(75:10mm) .. (D);
    \draw[line width=0.5mm] (D) .. controls +(5:14mm) and +(85:14mm) .. (D);

\end{scope}

\end{tikzpicture}

%% file: 3.Magnitude_Distinct_Hardness.tex
\section{Hardness of magnitude-distinct matrices}\label{sec: allDistinctHardness}

In this section, we  demonstrate the power of
the gadget interpolation techniques
to prove that for any \emph{magnitude-distinct} matrix 
$M \in \Sym{q}{pd}(\mathbb{R}_{> 0})$
(\cref{definition: magnitudeDistinctMatrix}),
$\PlGH(M)$ is $\#$P-hard.

\begin{definition}\label{definition: magnitudeDistinctMatrix}
	We call $M \in \Sym{q}{}(\mathbb{R})$ \emph{magnitude-distinct}
	if $|M_{ij}| = |M_{i'j'}|$ for $i \leq j$ and $i' \leq j'$
	implies that $(i, j) = (i', j')$.
	Equivalently, $M$ is magnitude-distinct 
	if $\phi_{\rm{mag}}(M) \neq 0$, where
	\begin{equation}
		\label{defn:phi-function-for-magnitude-distinct}
		\phi_{\rm{mag}}(M) = \prod_{i \leq j}
		\prod_{\substack{i' \leq j':\\ (i', j') \neq (i, j)}}
		\left((M_{ij})^{2} - (M_{i'j'})^{2}\right)
	\end{equation}
\end{definition}

\begin{definition}\label{definition: UAnalytic}
	Given an open set 
	$U \subseteq \Sym{q}{}(\mathbb{R})$,
	a function $\rho: U
	\rightarrow \mathbb{R}$ is called
	a $U$-analytic function if $\rho$
	is a real analytic function in the
	entries of the matrix.
	A function $\rho: \Sym{q}{}(\mathbb{R}) \rightarrow \mathbb{R}$
	is called a $\Sym{q}{}(\mathbb{R})$-polynomial,
	if $\rho$ is a polynomial on the entries of the matrix.
\end{definition}

\begin{remark*}
	Trivially, a $\Sym{q}{}(\mathbb{R})$-polynomial
	is also a $\Sym{q}{}(\mathbb{R})$-analytic function.
\end{remark*}

\begin{lemma}\label{lemma: analyticGadgets}
	Let $U \subseteq \Sym{q}{}(\mathbb{R})$ be
	an open set.
	Let $\mathcal{K}: Z \rightarrow U$,
	for some open, connected $Z \subseteq \mathbb{R}^{d}$,
	such that
	$\mathcal{K}(\mathbf{z})_{ij}$ is an analytic
	function of $\mathbf{z}$ for all $i, j \in [q]$.
	Let $\mathcal{F}$ be a countable set of
	$U$-analytic functions, such that
	for each $\rho \in \mathcal{F}$, there exists some
	$\mathbf{z}_{\rho} \in Z$ such that 
	$\rho(\mathcal{K}(\mathbf{z}_{\rho})) \neq 0$.
	Then, given any open subset
	$\mathcal{O} \subseteq Z$,
	there exists some $\mathbf{z}^{*} \in \mathcal{O}$,
	such that $\rho(\mathcal{K}(\mathbf{z}^{*})) \neq 0$ 
	for all $\rho \in \mathcal{F}$.
\end{lemma}
\begin{proof}
	For each $\rho \in \mathcal{F}$,
	$\rho(A)$ is analytic in the entries of $A$
	as $A$ varies over $U$.
	The entries of $\mathcal{K}(\mathbf{z})$ are
	analytic in $\mathbf{z} \in Z$, 
	and $\mathcal{K}(\mathbf{z})$ takes values in $U$.
	It follows that each
	$\rho(\mathcal{K}(\mathbf{z}))$ is analytic
	as a function of $\mathbf{z} \in Z$.
	Moreover, 
	from our assumption,
	$\rho(\mathcal{K}(\mathbf{z}))$ is not the zero function
	for any $\rho \in \mathcal{F}$.
	It follows that the zero set $\emptyset_{\rho} = 
	\{\mathbf{z}: \rho(\mathcal{K}(\mathbf{z})) = 0\}$
	has measure 0 (see \cite{mityagin2015zero}).
	Since a countable union of measure 0 sets has
	measure 0, and the measure of the open
	set $\mathcal{O}$ is positive, there exists some
	$\mathbf{z}^{*} \in (\mathcal{O} \setminus 
	\cup_{\rho}\emptyset_{\rho})$,
	which satisfies our requirements.
\end{proof}

\begin{remark*}
	$\mathcal{T}_{M}$,
	$\mathcal{T}^{*}_{M}$, $\widehat{\mathcal{T}}_{M}$, 
	$\mathcal{S}_{M}$, $\mathcal{B}_{M}$, 
	$\mathcal{L}_{M}$,
	and $\mathcal{L}^{*}_{M}$ all
	satisfy the conditions of $\mathcal{K}$ in
	\cref{lemma: analyticGadgets}.
\end{remark*}

We use \cref{lemma: analyticGadgets} as follows.
Assume  we can \emph{encode} a matrix
property as $\rho(M) \neq 0$ 
for some $\Sym{q}{}(\mathbb{R})$-analytic function $\rho$, 
and suppose
we already have  some $M \in \Sym{q}{}(\mathbb{R})$
with $\rho(M) \neq 0$.
Then, when we use  gadget interpolation functions
to find an 
$N \in \Sym{q}{}(\mathbb{R})$ 
to achieve some goals, we can \emph{inherit} this property
$\rho(N) \neq 0$.
This lets us satisfy \emph{countably}
many such matrix properties.

\begin{remark*}
	For example, one such matrix property is that
	$A \in \Sym{q}{}(\mathbb{R}_{\neq 0})$, which we can
	encode with the function $\rho: A \mapsto 
	\prod_{i \leq j} A_{ij}$.
	Another property is $A \in \Sym{q}{F}(\mathbb{R})$,
	which we can encode with the determinant function $\det$.
\end{remark*}

\begin{definition}\label{definition: IClose}
	We call $M \in \Sym{q}{}(\mathbb{R})$ \emph{I-close}
	if $|M_{ij} - I_{ij}| < \nicefrac{1}{3}$ 
	for all $i, j \in [q]$, where $I \in 
	\Sym{q}{}(\mathbb{R})$ is the
	identity matrix.
\end{definition}

Being I-close is \emph{not} a property
we can encode with $\Sym{q}{}(\mathbb{R})$-analytic
functions.
But, we can use the stretching gadget function
$\mathcal{S}_{M}$ to obtain $N$ \emph{close} to $I$,
that inherits relevant matrix properties. 

\begin{lemma}\label{lemma: I-closeStretching}
	Let $M \in \Sym{q}{pd}(\mathbb{R}_{> 0})$.
	Let $\mathcal{F}$ be some countable family of
	$\Sym{q}{}(\mathbb{R})$-polynomials, such that
	$\rho(M) \neq 0$ for all $\rho \in \mathcal{F}$.
	Then, there exists some
	I-close $N \in \Sym{q}{pd}(\mathbb{R}_{> 0})$
	such that $\rho(N) \neq 0$ for all
	$\rho \in \mathcal{F}$, and
	$\PlGH(N) \leq \PlGH(M)$.
\end{lemma}
\begin{proof}
	We let $\mathcal{F'} = 
	\{\rho': A \mapsto \rho(\mathbf{T_{2}}(A))\ |\ 
	\rho \in \mathcal{F} \}$.
	Recall $\mathcal{T}_{M}$ from \cref{definition: mathcalT}.
	By hypothesis, $\rho(M) \neq 0$
	for all $\rho \in \mathcal{F}$.
	So, $\rho'(\mathcal{T}_{M}(\nicefrac{1}{2}))
	\neq 0$ for all $\rho' \in \mathcal{F'}$.
	
	We are given $M \in
	\Sym{q}{pd}(\mathbb{R}_{> 0})$.
	Since the eigenvalues of a matrix are continuous as
	a function of its entries
	and $\mathcal{T}_{M}(1) = M
	\in \Sym{q}{pd}(\mathbb{R}_{> 0})$,
	there exists
	some open 
	interval $\mathcal{O} \subseteq \mathbb{R}$ around
	$1$, s.t.,
	$\mathcal{T}_{M}(\theta)
	\in \Sym{q}{pd}(\mathbb{R}_{> 0})$ for all 
	$\theta \in \mathcal{O}$.
	\cref{lemma: analyticGadgets} now lets us find some
	$\theta^{*} \in \mathcal{O}$ such that
	$M' = \mathcal{T}_{M}(\theta^{*})$
	satisfies:
	$\rho'(M') \neq 0$ for all $\rho' \in \mathcal{F'}$, and
	$\PlGH(M') \leq \PlGH(M)$.
	
	As $M' \in \Sym{q}{pd}(\mathbb{R}_{> 0})$,
	we can  consider $\mathcal{S}_{M'}(\theta)$
	(recall \cref{definition: mathcalS}).
	Note that $\mathcal{S}_{M'}(0) = I$.
	Since each entry
	of $\mathcal{S}_{M'}(\theta)$
	is continuous in $\theta$,
	there exists an open interval
	$\mathcal{O}' \subseteq \mathbb{R}$ around 
	$0$ such that for all
	$\theta \in \mathcal{O}'$, 
	$|\mathcal{S}_{M'}(\theta)_{ij} - I_{ij}| < \nicefrac{1}{10}$,
	and
	$$(\mathcal{S}_{M'}(\theta)_{ii})^{2}
	> \sum_{j \neq i \in [q]}(\mathcal{S}_{M'}(\theta)_{ij})^{2},
	\qquad \text{ for all } i \in [q].$$
	We use \cref{lemma: analyticGadgets} again
	to find $\theta^{**} \in \mathcal{O}'$ such that
	$M'' = \mathcal{S}_{M'}(\theta^{**})
	\in \Sym{q}{pd}(\mathbb{R})$ satisfies:
	$\rho'(M'') \neq 0$ for all $\rho' \in \mathcal{F'}$,
	$\prod_{i, j}(M'')_{ij} \neq 0$,
	and $\PlGH(M'') \leq \PlGH(M') \leq \PlGH(M)$.
	
	We finally let $N = T_{2}(M'')$.
	From the construction of $\mathcal{F'}$, we see
	that $\rho'(M'') \neq 0 \implies \rho(N) \neq 0$ for all
	$\rho \in \mathcal{F}$.
	Moreover, by the choice of $\theta^{*} \in \mathcal{O}'$,
	we see that $N_{ii} \in ((\nicefrac{9}{10})^{2}, 
	(\nicefrac{11}{10})^{2}) \subset 
	(\nicefrac{2}{3}, \nicefrac{4}{3})$,
	for all $i \in [q]$,
	and $N_{ij} \in (0, \nicefrac{1}{3})$ for all
	$i \neq j \in [q]$.
	Since $(M'')_{ij} \in \mathbb{R}_{\neq 0}$
	for all $i, j \in [q]$, we also have that
	$N \in \Sym{q}{}(\mathbb{R}_{> 0})$.
	Moreover,
	$N_{ii} > \sum_{j \neq i}|N_{ij}|$.
	So, by the Gershgorin Circle Theorem,
	(see ~\cite{horn2012matrix} Theorem 6.1.1)
	$N \in \Sym{q}{pd}(\mathbb{R}_{ > 0})$
	is the required matrix.
\end{proof}


\begin{definition}\label{definition: DiagonalDominant}
	We call $M \in \Sym{q}{}(\mathbb{R})$
	\emph{diagonal dominant}, if
	$M_{ii} > |M_{jk}|$ for all
	$i\in [q]$, $j \neq k \in [q]$.
\end{definition}
Clearly, if $M \in \Sym{q}{}(\mathbb{R})$
is I-close then it is also diagonal dominant.
We are now ready to prove the
main theorem of this section,
but we shall need to use the following
theorem:

\begin{theorem}[\cite{guo2020complexity}]
	\label{theorem: domain2Hardness}
	The problem $\PlGH(M)$ is $\#$P-hard, for
	$M = \left( \begin{smallmatrix}
		x & y\\
		y & z\\
	\end{smallmatrix} \right) \in 
	\Sym{2}{F}(\mathbb{R}_{> 0})$,
	unless $x = z$, in which case,
	$\PlGH(M)$ is polynomial-time computable.
\end{theorem}

\begin{theorem}\label{theorem: magnitudeDistinctHardness}
	Let $M \in \Sym{q}{pd}(\mathbb{R}_{> 0})$
	be magnitude-distinct ($q \geq 3$).
	Then, $\PlGH(M)$ is $\#$P-hard.
\end{theorem}
\begin{proof}
	We can use \cref{lemma: I-closeStretching},
	with  $\mathcal{F} \ni \phi_{\rm{mag}}$ 
	(where $\phi_{\rm{mag}}$ is defined in
	\cref{defn:phi-function-for-magnitude-distinct}),
	to obtain an I-close (and therefore diagonal dominant),
	and magnitude-distinct
	$M' \in \Sym{q}{pd}(\mathbb{R}_{> 0})$,
	such that $\PlGH(M') \leq \PlGH(M)$.
	We will now consider $\widehat{\mathcal{T}}_{M'}(z)$
	as defined in \cref{lemma: thickeningSingleVariable}.
	We know that there exist $x_{ij} \in \mathbb{Z}_{\geq 0}$
	such that $\widehat{\mathcal{T}}_{M'}(z)_{ij} = z^{x_{ij}}$, and
	there exists some $\widehat{z} \in \mathbb{R}$, such that 
	$\det(\widehat{\mathcal{T}}_{M'}(\widehat{z})) \neq 0$.
	
	Let $x^{*} = \min_{i, j}x_{ij}$.
	From the construction of $\widehat{\mathcal{T}}_{M'}$ 
	and the fact that $M'$ is magnitude-distinct
	and diagonal dominant,
	it follows that $x^{*} = x_{i^{*}j^{*}}$ for some $i^{*} < j^{*}$,
	such that $x^{*} < x_{ij}$ for all other $i \leq j
	\in [q]$.
	It can also be shown (see \cref{lemma: MequivalentCM})
	that $\PlGH(\widehat{\mathcal{T}}_{M'}(z)) \equiv 
	\PlGH(z^{-x^{*}} \cdot \widehat{\mathcal{T}}_{M'}(z))$
	for all $z \in \mathbb{R}_{\geq 0}$.
	Therefore, we may assume without loss of generality,
	that $x^{*} = x_{i^{*}j^{*}} = 0$.
	
	As a polynomial in $z$, $\det(\widehat{\mathcal{T}}_{M'}(z))$ 
	is not identically zero
	by the existence of
	$\widehat{z}$.
	So, there exists some $z^{*}$, s.t.,
	$\det(\widehat{\mathcal{T}}_{M'}(z)) \neq 0$
	for all $z > z^{*}$.
	Since $M'$ is magnitude-distinct, 
	it also follows that for all
	$k \neq i^{*}, j^{*} \in [q]$,
	$M'_{i^{*}k} \notin \{M'_{j^{*}j}: j \in [q]\}$.
	Moreover, since $q \geq 3$, such a $k \in [q]$ exists.
	So, the  difference
	$ S_{2}(\widehat{\mathcal{T}}_{M'}(z))_{i^{*}i^{*}} -  
	S_{2}(\widehat{\mathcal{T}}_{M'}(z))_{j^{*}j^{*}}$ is
	$f(z) = \sum_{k \in [q]}z^{2x_{i^{*}k}} - 
	\sum_{k \in [q]}z^{2x_{ j^{*}k}}$, which is not the zero
	polynomial. 
	So, there  exists some $z^{**}
	\in \mathbb{R}$ such that for all $z > z^{**}$,
	$f(z) \neq 0$.
	
	We can now pick some prime integer $p > \max(z^{*}, z^{**})$, and let
	$M'' = \mathbf{S_{2}}(\widehat{\mathcal{T}}_{M'}(p))
	= (\widehat{\mathcal{T}}_{M'}(p))^{2}$.
	From our choice of $p > z^{*}$,
	$\det(\widehat{\mathcal{T}}_{M'}(p)) \neq 0$.
	Being the square of a full rank symmetric matrix, $M''$ is
	positive definite, so $ M
	\in \Sym{q}{pd}(\mathbb{Z}_{> 0})$.
	We note that for all $i \in [q]$,
	$M''_{ii} = 
	\sum_{j \in [q]}p^{2x_{ij}}$.
	Therefore, 
	$$M''_{i^{*}i^{*}}
	= p^{2x^{*}} + \sum_{j \neq j^{*}}p^{2x_{i^{*}j}}
	= 1 + \sum_{j \neq j^{*}}p^{2x_{i^{*}j}} \equiv
	1 \mod p.$$
	Similarly, $M''_{j^{*}j^{*}} \equiv 1 \mod p$.
	Moreover, for $k \neq i^{*}, j^{*}$,
	we see that $M''_{kk} \equiv 0 \mod p$,
	since $x_{kj} > 0$ for all $j \in [q]$.
	Furthermore, from our choice of $p > z^{**}$,
	we can also see that
	$M''_{i^{*}i^{*}} \neq M''_{j^{*}j^{*}}$
	
	Since $M'' \in \Sym{q}{pd}(\mathbb{Z}_{> 0})$,
	the set of primes that divide
	$\prod_{i}M''_{ii}$ form a generating set
	for the diagonal entries of $M''$.
	We may assume that
	this generating set is $\{g_{1}, \dots, g_{d}\}$.
	Since we know that $p \mid M''_{ii}$ for $i \notin
	\{i^{*}, j^{*}\}$, and $q \ge 3$, we know that
	$p \in \{g_{1},  \dots, g_{d} \}$.
	We may name $g_{1} = p$.
	Now, we can let $M''' = \mathcal{L}^{*}_{M''}(0, 1, \dots, 1)$.
	From \cref{definition: mathcalLStar}, we note that
	$\mathcal{D}^{*}_{M''}(0, 1, \dots, 1)_{ii} = 0$ for
	all $i \notin \{i^{*}, j^{*}\}$, since $g_{1} \mid M''_{ii}$.
	Moreover, $\mathcal{D}^{*}_{M''}(0, 1, \dots, 1)_{i^{*}i^{*}}
	= \mathcal{D}^{*}_{M''}(0, 1, \dots, 1)_{j^{*}j^{*}} = 1$.
	Thus, with $\mathbf{z} = (0, 1, \dots, 1)$,
	$\mathcal{D}^{*}_{M''}(\mathbf{z})$
	selects exactly two rows and columns
	$\{i^{*}, j^{*}\}$    of $M''$, in
	$M''' = \mathcal{D}^{*}_{M''}(\mathbf{z})
	\cdot M'' \cdot \mathcal{D}^{*}_{M''}(\mathbf{z})$.
	So,  $M'''_{ij} = M''_{ij}$
	iff $\{i,  j\} \subseteq \{i^{*}, j^{*}\}$,
	and $M'''_{ij} = 0$ otherwise.
	
	We can now let $N \in \Sym{2}{}(\mathbb{Z}_{> 0})$ be the 2 by 2
	principal sub-matrix of $M''$ formed by the rows and columns
	$i^{*},  j^{*} \in [q]$. Then $\PlGH(N) \equiv \PlGH(M''') $.
	Since $M'' \in \Sym{q}{pd}(\mathbb{Z}_{> 0})$, it follows
	that all principal sub-matrices must also be positive definite
	(see ~\cite{horn2012matrix} Observation 7.1.2).
	So, $N \in \Sym{2}{pd}(\mathbb{Z}_{> 0})$.
	Since $M''_{i^{*}i^{*}} \neq M''_{j^{*}j^{*}}$,
	\cref{theorem: domain2Hardness} tells us that
	$\PlGH(N)$ is $\#$P-hard, and $\PlGH(N) \equiv \PlGH(M''') \leq \PlGH(M)$.
\end{proof}

%% file: 4.Diagonal_Distinct_Hardness.tex
\section{Hardness of diagonal-distinct matrices}\label{sec: diagonalDistinctHardness}

\begin{definition}\label{definition: diagonalDistinct}
	We call $M \in \Sym{q}{}(\mathbb{R})$
	\emph{diagonal-distinct} (resp., \emph{strongly diagonal-distinct})
	if the $q$ elements in $\{M_{11}, \ldots, M_{qq}\}$ are distinct 
	(resp., for any $i, j, k \in [q]$,  $M_{ii} = M_{jk}$ 
	$\Longrightarrow$  $i = j= k$).
	Equivalently, $M$ satisfies 
	$\phi_{\rm{diag}}(M) \neq 0$ (resp., $\psi_{\rm{diag}}(M) \neq 0$), where
	\begin{equation}
		\label{defn:phi-psi-function-for-diagonal-distinct}
		\phi_{\rm{diag}}(M) = \prod_{i < j}
		\left(M_{ii} - M_{jj}\right)
		~~~\mbox{and}~~~
		\psi_{\rm{diag}}(M) = \prod_{i \in [q]} \prod_{(j,k) \neq (i, i)}
		\left(M_{ii} - M_{jk}\right)
	\end{equation}
\end{definition}

In this section, we will prove that diagonal-distinctness
is sufficient to prove $\#$P-hardness for
$M \in \Sym{q}{pd}(\mathbb{R}_{> 0})$
(all missing proofs are in 
\autoref{appendix: diagonalDistinctHardness}).
We prove this by induction. 
Assume the following is true for all $2 \leq t < q$.

\begin{induction*}[Induction Hypothesis ($t$)]\label{induction: inductionHyp}
	If $M \in \Sym{t}{pd}(\mathbb{R}_{> 0})$ is 
	diagonal-distinct
	then $\PlGH(M)$ is $\#$P-hard.
\end{induction*}

The base case $t = 2$ is proved by \cref{theorem: domain2Hardness}.

\begin{restatable}{theorem}{inductionStep}\label{theorem: inductionStep}
	Let $M \in \Sym{q}{F}(\mathbb{R}_{> 0})$
	be strongly diagonal-distinct.
	Let $m^{*} = \min_{i, j}M_{ij}$.
	If there exists some $i^{*} \in [q]$, such that
	$M_{i^{*}j} > m^{*}$ for all $j \in [q]$,
	and $i_{1}^{\dagger} \neq i_{2}^{\dagger} \in [q]$,
	$j_{1}^{\dagger}, j_{2}^{\dagger} \in [q]$, such that
	$M_{i_{1}^{\dagger}j_{1}^{\dagger}} = 
	M_{i_{2}^{\dagger}j_{2}^{\dagger}} = m^{*}$,
	then $\PlGH(M)$ is $\#$P-hard.
\end{restatable}
The proof is similar to \cref{theorem: magnitudeDistinctHardness},
and uses \nameref{induction: inductionHyp}.

\begin{remark*}
	\cref{theorem: inductionStep} places
	the additional constraint that
	$M$ be strongly diagonal-distinct,
	but we shall see later, using
	\cref{lemma: analyticGadgets}, that
	if $M \in \Sym{q}{pd}(\mathbb{R}_{> 0})$ 
	is diagonal-distinct, we can find
	some $\PlGH(N) \leq \PlGH(M)$ that is
	also \emph{strongly diagonal-distinct}.
\end{remark*}

The next step of our proof is to show that given
a diagonal distinct
$M \in \Sym{q}{pd}(\mathbb{R}_{> 0})$,
we can 
find some $N \in \Sym{q}{pd}(\mathbb{R}_{> 0})$
with additional desirable properties,
such that $\PlGH(N) \leq \PlGH(M)$.

\begin{definition}\label{definition: intraRowDistinct}
	We call $M \in \Sym{q}{}(\mathbb{R})$ \emph{intra-row-distinct}
	if for all $i \in [q]$, $M_{ij} = M_{ij'} \implies j = j'$.
	Equivalently, $M$ is intra-row-distinct if
	$\phi_{\rm{row}}(M) \neq 0$, where
	\begin{equation}
		\label{defn:phi-function-for-intra-row-distinct}
		\phi_{\rm{row}}(M) = \prod_{i \in [q]} \prod_{j < j' \in [q]}
		\left(M_{ij} - M_{ij'}\right)
	\end{equation} 
\end{definition}

\begin{restatable}{lemma}{diagDistinctImpliesMatrixProperties}
	\label{lemma: diagDistinctImpliesMatrixProperties}
	Let $M \in \Sym{q}{pd}(\mathbb{R}_{> 0})$ be
	diagonal distinct.
	Then, either $\PlGH(M)$ is $\#$P-hard,
	or there exists an I-close (thus diagonal dominant),
	strongly diagonal distinct,
	and intra-row-distinct
	$N \in \Sym{q}{pd}(\mathbb{R}_{> 0})$, s.t.,
	$\PlGH(N) \leq \PlGH(M)$.
\end{restatable}


Given a matrix $N$
that satisfies the conditions of
\cref{lemma: diagDistinctImpliesMatrixProperties},
we can use \cref{theorem: inductionStep} to
prove that $\PlGH(N)$ is $\#$P-hard,
so long as the minimum entry of $N$
does not appear on \emph{every} row of $N$.
Given that we already know $N$ to be strongly
diagonal distinct and intra-row-distinct,
it seems reasonable to hope that we can also
ensure that no entry appears on every row of $N$.
In that case, it would trivially follow that the minimum
entry of the matrix also does not appear in every row.
This however, turns out to not be possible.
For example, if $M = \left( \begin{smallmatrix}
	x & y\\
	y & z\\
\end{smallmatrix} \right) \otimes \left( \begin{smallmatrix}
	a & b\\
	b & c\\
\end{smallmatrix} \right)$,
we see that $M_{14} = M_{23} = (b \cdot y)$.
More worryingly, it can be verified that
$\mathbf{K}(M) = \mathbf{K}\left(\left( \begin{smallmatrix}
	x & y\\
	y & z\\
\end{smallmatrix} \right)\right) \otimes 
\mathbf{K}\left(\left( \begin{smallmatrix}
	a & b\\
	b & c\\
\end{smallmatrix} \right)\right),$
for all $\mathbf{K}(M) \in \PlEdge(M)$.
In other words, using edge gadget reductions,
we cannot possibly find a reduction
$\PlGH(N) \leq \PlGH(M)$ which does not satisfy:
$N_{14} = N_{23}$.

So, instead of trying to prove that we can find
$\PlGH(N) \leq \PlGH(M)$ where no entry appears
on every row of $N$, we instead try to ensure
that even if some entry appears on every row of $N$,
it is not the minimum entry of $N$.
We accomplish this by using the various gadget functions
we have introduced in \cref{sec: preliminaries}
to prove a series
of lemmas, that either place successively more 
restrictions on the form of the matrix $N$,
or prove that $\PlGH(N)$ is $\#$P-hard.
The exact statements of all these lemmas along
with a detailed discussion of their proofs can be
found in \autoref{appendix: diagonalDistinctHardness}.
For now, we will summarize the results here as one
theorem with just a few of the most relevant
restrictions.

\begin{theorem}\label{theorem: diagDistinctAllProperties}
	Let $N \in \Sym{q}{pd}(\mathbb{R}_{ > 0})$ 
	be diagonal dominant,
	strongly diagonal distinct, and intra-row-distinct.
	Then, either $\PlGH(N)$ is $\#$P-hard, or
	after a renaming of domain elements $[q]$
	(below we write $i^{\rm R} = q+1-i$, for all $i \in [q]$),    
	\begin{specialenumerate}
		\item
		$N_{i,i^{\rm R}} = \min_{i,j \in [q]} N_{i,j}$
		for all $i \in [q]$,
		\setcounter{specialenumeratei}{3}
		\item
		$N_{11} > N_{22} > \dots > N_{qq}$,
		\setcounter{specialenumeratei}{9}
		\item
		$\aleph(N) = 0$, where
		$\aleph$ is the
		$\Sym{q}{}(\mathbb{R}_{> 0})$-analytic function defined as:
	\end{specialenumerate}
	{\small$$\aleph(A) = \Big(\ln(A_{1q}) - \ln(A_{qq}) \Big)
		\Big(\ln(A_{(q-1)(q-1)}) - \ln(A_{qq})\Big) - 
		\Big(\ln(A_{11}) - \ln(A_{qq})\Big)
		\Big(\ln(A_{(q-1)q}) - \ln(A_{qq}) \Big).$$}
\end{theorem}

The exact statements of all the other restrictions that can be placed
on the form of $N$ 
can be found in
\autoref{appendix: diagonalDistinctHardness}
(see \cref{lemma: diagDistinctAntiDiagonal,lemma: diagDistinctRestrictions,lemma: diagDistinctLoopRestrictions,lemma: diagDistinctSecondLastRow,lemma: diagDistinctAleph}).
Finally, we prove the following lemma.

\begin{restatable}{lemma}{diagDistinctHardness}
	\label{lemma: diagDistinctHardness}
	Let $N \in \Sym{q}{pd}(\mathbb{R}_{> 0})$
	be I-close, strongly diagonal distinct, and
	intra-row-distinct, that satisfies
	all the restrictions in 
	\cref{theorem: diagDistinctAllProperties}.
	Then there exists some
	$N' \in \Sym{q}{pd}(\mathbb{R}_{> 0})$,
	such that $\PlGH(N') \leq \PlGH(N)$,
	and is diagonal dominant, strongly diagonal distinct,
	and intra-row-distinct, but
	\emph{does not} simultaneously
	satisfy all the restrictions in
	\cref{theorem: diagDistinctAllProperties}.
\end{restatable}



\begin{theorem}\label{theorem: diagonalDistinctHardness}
	Let $M \in \Sym{q}{pd}(\mathbb{R}_{> 0})$ be
	diagonal distinct (for $q \geq 2$).
	Then, $\PlGH(M)$ is $\#$P-hard.
\end{theorem}
\begin{proof}
	From \cref{lemma: diagDistinctImpliesMatrixProperties},
	we see that either $\PlGH(M)$ is $\#$P-hard,
	or there exists an I-close (thus diagonal dominant),
	strongly diagonal distinct, intra-row-distinct
	$N \in \Sym{q}{pd}(\mathbb{R}_{> 0})$, s.t.,
	$\PlGH(N) \leq \PlGH(M)$.
	$\PlGH(N)$ is $\#$P-hard, unless $N$ satisfies all  
	properties in \cref{theorem: diagDistinctAllProperties}.
	Now, suppose that $N$ satisfies all these properties.
	\cref{lemma: diagDistinctHardness} implies that
	$\PlGH(N') \leq \PlGH(N)$ does not satisfy
	all the properties of \cref{theorem: diagDistinctAllProperties}.
	So, $\PlGH(N') \leq \PlGH(N) \leq \PlGH(M)$ is 
	$\#$P-hard.
\end{proof}

%% file: 5.Diagonal_Distinct_Dichotomy.tex
\section{Dichotomy for diagonal-distinct matrices}\label{sec: diagonalDistinctDichotomy}

In this section, we
use previously known techniques to
establish
a complexity dichotomy for all
diagonal distinct
$M \in \Sym{q}{}(\mathbb{R}_{\geq 0})$
(without assuming all
positive  and being positive definite).
The first step is to extend
\cref{theorem: diagonalDistinctHardness}
from positive definite matrices to
full rank matrices.
This is fairly straightforward, since
$M \in \Sym{q}{F}(\mathbb{R}_{> 0})$
implies that $\mathbf{S_{2}}(M) = M^{2} \in 
\Sym{q}{pd}(\mathbb{R}_{>0})$.
With clever use of gadgets, we can also
ensure diagonal distinctness, to prove the
following:

\begin{restatable}{lemma}{fullRankHardness}
\label{lemma: fullRankHardness}
    Let $M \in \Sym{q}{F}(\mathbb{R}_{> 0})$
    be diagonal distinct (for $q \geq 2$).
    Then, $\PlGH(M)$ is $\#$P-hard.
\end{restatable}

A full proof can be found in
\autoref{appendix: diagonalDistinctDichotomy}.
Using techniques and results adapted from
~\cite{govorov2020dichotomy} and
~\cite{cai2013graph}, we
can also prove the following 
(more details in \autoref{appendix: diagonalDistinctDichotomy}):

\begin{restatable}{theorem}{positiveDichotomy}
\label{theorem: positiveDichotomy}
    Let $M \in \Sym{q}{}(\mathbb{R}_{> 0})$
    be diagonal distinct.
    Then $\PlGH(M)$ is polynomial time tractable
    if $M$ has rank $1$.
    In all other cases,
    $\PlGH(M)$ is $\#$P-hard.
\end{restatable}


\begin{restatable}{theorem}{nonNegativeDichotomy}
\label{theorem: nonNegativeDichotomy}
    Let $M \in \Sym{q}{}(\mathbb{R}_{\geq 0})$
    be diagonal distinct.
    Then $\PlGH(M)$ is polynomial time tractable
    if (after a renaming of $[q]$), 
    $M = M_{1} \oplus \dots \oplus M_{r}$
    for rank $0$ or rank $1$ matrices $M_{i}$.
    In all other cases, $\PlGH(M)$ is $\#$P-hard.
\end{restatable}

%% file: 6.Quantum.tex
\section{Edge gadgets and quantum automorphisms}\label{sec: quantumBody}

The dichotomy of 
\cref{theorem: nonNegativeDichotomy} can be extended
to prove the following:
\begin{restatable}{theorem}{fullGadgetSeparationNonNegative}
    \label{theorem: fullGadgetSeparationNonNegative}
    Let $M \in \Sym{q}{}(\mathbb{R}_{\geq 0})$.
    Assume that for all $x < y \in [q]$, there exists
    some $\mathbf{K_{xy}}(M) \in \PlEdge(M)$ such that
    $\mathbf{K_{xy}}(M)_{xx} \neq \mathbf{K_{xy}}(M)_{yy}$.
    Then, $\PlGH(M)$ is polynomial-time tractable
    if (after a renaming of $[q]$),
    $M = M_{1} \oplus \dots \oplus M_{r}$,
    for rank $0$ or rank $1$ or bipartite rank $2$ matrices
    $M_{i}$.
    In all other cases, $\PlGH(M)$ is $\#$P-hard.   
\end{restatable}
A proof of \cref{theorem: fullGadgetSeparationNonNegative}
can be found in \autoref{appendix: quantumAppendix}.
The following question now arises:
Given $M \in \Sym{q}{}(\mathbb{R}_{> 0})$ and $i < j \in [q]$,
is  there  a
$\mathbf{K}(M) \in \PlEdge(M)$ such that 
$\mathbf{K}(M)_{ii} \neq \mathbf{K}(M)_{jj}$? 
When we allow all, possibly \emph{nonplanar},
$\mathbf{K}(M) \in \Edge(M)$, this question has a 
definitive answer in terms of graph isomorphism.
The \emph{automorphism group} $\aut(M)$ of $M$ is the group of
permutations $\sigma \in S_q$ such that 
$M_{ij} = M_{\sigma(i) \sigma(j)}$ for all
$i,j \in [q]$. 
The \emph{orbits} of $\aut(M)$ are the equivalence classes of the relation
$\sim$ on $[q]$ defined by $i \sim j \iff 
\exists\ \sigma \in \aut(M): \sigma(i) = j$. 
If $i \sim j$,
then $\mathbf{K}(M)_{ii} = \mathbf{K}(M)_{jj}$ for all $\mathbf{K}(M)
\in \Edge(M)$. 
For $M \in \Sym{q}{}(\mathbb{R}_{> 0})$, we show, using results of \cite{lovasz2006rank,young2022equality}, that the converse also holds, even restricted to symmetric gadget signatures.

\begin{restatable}{proposition}{autOrbits} 
    \label{proposition: autOrbits}
    Let $M \in \Sym{q}{}(\mathbb{R}_{\geq 0})$. Then
    $\mathbf{K}(M)_{ii} = \mathbf{K}(M)_{jj}$ 
    for every $\mathbf{K}(M) \in \Edge(M)$ 
    if and only if $i$ and $j$ are in the same
    orbit of $\aut(M)$.
\end{restatable}

We refer the reader to \autoref{appendix: quantumAppendix} for proofs of the
results in this section.
It is decidable whether $i \sim j$, so, by~\cref{proposition: autOrbits}, it is decidable whether there is a
$\mathbf{K}(M) \in \Edge(M)$ separating $i$ and $j$.
For a planar version of \autoref{proposition: autOrbits}, we turn to the quantum versions
of $\aut(M)$ and its orbits.
The \emph{quantum automorphism group} $\qut(M)$ of $M$ is a
highly abstract 
object
designed to mimic and generalize $\aut(M)$. In particular, $\qut(M)$ also admits well-behaved orbits on $[q]$ that are often
identical to the orbits of $\aut(M)$, but can be strictly coarser. To prove the equivalence
between quantum isomorphism and planar graph homomorphism indistinguishability, 
Man\v{c}inska and Roberson \cite{manvcinska2020quantum} developed connections between the orbits of $\qut(M)$ and planar homomorphism gadgets (generalized to weighted
graphs in \cite{cai2023planar}), from which we 
can prove the following.
\begin{restatable}{theorem}{qutOrbits} 
    \label{theorem: qutOrbits}
    Let $M \in \Sym{q}{}(\mathbb{R}_{\geq 0})$. Then
    $\mathbf{K}(M)_{ii} = \mathbf{K}(M)_{jj}$ 
    for every $\mathbf{K}(M) \in \PlEdge(M)$ 
    if and only if $i$ and $j$ are in the same
    orbit of $\qut(M)$.
\end{restatable}
We say $\qut(M)$ is \emph{trivial} if its orbits all have size 1. In this case, any
$i \neq j \in [q]$ are distinguishable by some $K \in \PlEdge(M)$.
We call $M$ \emph{connected} if the underlying graph
$\Gamma_{M} = ([q], E(\Gamma_{M}))$
is connected, where $(i,  j) \in
E(\Gamma_{M})$ iff $M_{ij} > 0$.
For such $M$,
we can find a single edge gadget distinguishing all pairs $i,j \in [q]$ simultaneously.

\begin{restatable}{corollary}{distinct}
    \label{corollary: distinct}
    Let $M \in \Sym{q}{}(\mathbb{R}_{\geq 0})$ be connected. Then
    $\qut(M)$ is trivial if and only if there exists a diagonal distinct $\vk(M) \in \PlEdge(M)$.
\end{restatable}
On the other extreme, we could have $\qut(M) = S_q^+$, the \emph{quantum symmetric group}.
Let $I,J \in \Sym{q}{}(\mathbb{R}_{\geq 0})$ be the 
identity and the all-ones matrices, respectively.
\begin{restatable}{proposition}{quantumSymmetric}
\label{proposition: quantumSymmetric}
Let $M \in \Sym{q}{}(\mathbb{R}_{\geq 0})$. Then
    $\qut(M) = S_q^+$ if and only if $M \in {\rm span}_{\mathbb{R}}(I,J)$.
\end{restatable}
Observe that every matrix in $M \in \text{span}_{\mathbb{R}}(I,J)$ is either a scalar
multiple of $I$, a scalar multiple of $J$, or of the form $\{aI + bJ \mid ab \neq 0\}$.
In the first two cases, $\PlGH(M)$ is trivially tractable, and in the third case (for $q \geq 3$),
\#P-hard 
from the Potts model \cite{vertigan2005computational}.

Mirroring the classical case, connected graphs $H$ and $H'$ with adjacency matrices $M$ and $M'$ are quantum isomorphic if and only if some $v \in V(H)$ and $v' \in V(H')$ are in the same orbit of $\qut(M \oplus M')$ \cite[Theorem 4.5]{lupini2020nonlocal}. It follows that
quantum isomorphism of graphs reduces to the following problem:
Given $q \geq 1$, $M \in \Sym{q}{}(\mathbb{R}_{> 0})$ and $i, j \in [q]$, determine whether $i$ and $j$ are in the same orbit of $\qut(M)$.
Quantum isomorphism is undecidable via the theory of quantum strategies for
nonlocal games \cite{atserias2019quantum,slofstra2019set}, so, by \autoref{theorem: qutOrbits}, we conclude that the problem of finding a planar
edge gadget with specific distinct diagonal entries is undecidable.
\begin{restatable}{theorem}{undec}
    \label{theorem: undec}
    The following problem is undecidable:
    given $q \geq 1$, $M \in \Sym{q}{}(\mathbb{R}_{> 0})$ 
    and $i, j \in [q]$, determine whether there 
    exists some $\mathbf{K}(M) \in \PlEdge(M)$ such that 
    $\mathbf{K}(M)_{ii} \neq \mathbf{K}(M)_{jj}$.
\end{restatable}
\begin{remark}\label{remark: separate_nonplanar}
Let $M$ and $M'$ be the adjacency matrices of graphs which are quantum
isomorphic but not isomorphic. Infinitely many such pairs can
be constructed from e.g. nonlocal games \cite{atserias2019quantum} or
Hadamard matrices \cite{chan2024quantum}. We may
assume the graphs are connected (see \autoref{proposition: separate_nonplanar} in 
\autoref{appendix: quantumAppendix}), 
so there exist $i \in V(M)$ and $j \in V(M')$ 
in the same orbit of $\qut(M \oplus M')$ but
different orbits of $\aut(M \oplus M')$. 
By \autoref{proposition: autOrbits} and \autoref{theorem: qutOrbits}, 
these $i$ and $j$ can be separated, 
but \emph{only} by a nonplanar gadget: 
There is a 
$\mathbf{K}'(M \oplus M') \in \mathfrak{E}(M \oplus M')$ such that
$\mathbf{K}'(M \oplus M')_{ii} \neq \mathbf{K}'(M \oplus M')_{jj}$,
but for every $\mathbf{K}(M \oplus M')
\in \PlEdge(M \oplus M')$, we have $\mathbf{K}(M \oplus M')_{ii} 
= \mathbf{K}(M \oplus M')_{jj}$.
\end{remark}

%% file: 7.Future_Work.tex
\section{Implications and Future Work}\label{sec: futureWork}

\cref{corollary: distinct} hints at some of the challenges
we face when trying to prove
$\#$P-hardness of $\PlGH(M)$ using
planar edge gadgets, when $\qut(M)$ is
non-trivial.
For example, given a graph $G = (V, E)$
Dyer and Greenhill (\cite{dyer2000complexity})
construct a graph $G'$ by starting with
the disjoint union of $G$ and 
a star graph with $|V|$ leaf nodes, and then 
identifying each of these leaf nodes with
one of the vertices of $V$.
This highly non-planar construction lets them prove
that
$\GH\left(\left(\begin{smallmatrix}
	1 & 1\\
	1 & 0\\
\end{smallmatrix} \right)\right) \leq \GH(M)$,
even for matrices $M$
with non-trivial $\qut(M)$.
We see from \cref{corollary: distinct} that
such a reduction from a diagonal distinct matrix
cannot be accomplished for 
all $M \in \Sym{q}{}(\mathbb{R}_{\geq 0})$
with non-trivial $\qut(M)$
using any planar edge gadgets.

This does not imply that planar edge gadgets
\emph{cannot} be used to establish the complexity dichotomy
in the $\PlGH$ setting.
For example, Cai and Maran (\cite{cai2024polynomial})
exclusively use planar edge gadgets to establish
the complexity dichotomy for 
$M \in \Sym{4}{}(\mathbb{R}_{\geq 0})$.
But even then, in their proof of 100 pages,
nearly 50 pages are devoted to dealing
with specific sub-cases of non-trivial
$\qut(M)$.
They use the very specific forms of the matrices
for the small $q = 4$ case,
and through a tortured series of lemmas, prove
$\#$P-hardness for each non-trivial $\qut(M)$.
From \cref{theorem: undec}, we see that for
arbitrary $q$, it is impossible
to characterize the form 
that a matrix $M$ with non-trivial
$\qut(M)$ must take.
So, this strategy \emph{cannot} work for arbitrary $q$.



By \cref{theorem: qutOrbits}, the dichotomy in \cref{theorem: fullGadgetSeparationNonNegative} is equivalent to the following: For any $M \in \Sym{q}{}(\mathbb{R}_{\geq 0})$, if $\qut(M)$ is trivial, then $\PlGH(M)$ is either tractable or $\#$P-hard. On the other extreme, \cref{proposition: quantumSymmetric} 
and the surrounding discussion
gives a (relatively trivial) dichotomy of the same form for $M$ such that $\qut(M) = S_q^+$. For $M$ with intermediate $\qut(M)$ (some quantum symmetry, but not full
quantum symmetry), no dichotomy for $\PlGH(M)$ is yet known.

To improve upon this current work,
we will need to establish a dichotomy in the case where
some pairs of vertices cannot be separated using
planar edge gadgets.
Any inductive proof would necessarily have to deal
with the base case of matrices where all the
diagonal values are identical.
So, a natural next step would be to try to establish
the complexity dichotomy in this case, which
corresponds to $\qut(M)$ having a single orbit.

Towards this goal, our paper hints at several
challenges we face:
There exist $M$ for which $\qut(M)$ has a single orbit ($M$ is
\emph{quantum vertex transitive}), but $\aut(M)$ has multiple orbits ($M$ is not
vertex transitive) \cite{lupini2020nonlocal}. It follows as in \cref{remark: separate_nonplanar} that there is a nonplanar edge gadget that does
not have a uniform diagonal, but there is no such planar edge gadget.
This makes it hard to even characterize
these problems.
This case also poses another additional challenge:
Unlike the case with trivial $\qut(M)$,
there exist some new polynomial-time tractable
$\PlGH(M)$ (tensor products of matchgates).
So, proving a dichotomy would require
\emph{carving out} these polynomial-time tractable problems,
before proving the $\#$P-hardness of the rest.

These challenges might lead us to consider looking
outside of $\PlEdge(M)$ and using
other previously unexplored techniques to find reductions
$\Pi \leq \PlGH(M)$ from problems $\Pi$ that cannot be
represented as $\PlGH(N)$ problems.
Alternatively, we may have to develop
much stronger analytic techniques
(without relying on the specific forms of $M$)
to prove $\#$P-hardness, within the constraints
of having to use planar edge gadgets.

%% file: A2.0.Model_of_Computation.tex
\section{Model of Computation}\label{appendix: modelComputation}

The Turing machine model is naturally suited to the study of computation over discrete structures such as integers or graphs. 
When  $M \in \Sym{q}{}(\mathbb{R})$, for  $\PlGH(M)$ 
one usually restricts $M$ to be a matrix
with only algebraic numbers. This is strictly for the consideration
of the model of computation, even though allowing all
real-valued matrices would be more natural. 
 
 There is a formal (albeit nonconstructive) method to treat  $\PlGH(M)$ 
 for arbitrary real-valued matrices $M$ and yet stay strictly
 within the Turing machine model in terms of  bit-complexity.
In this paper, because our proof depends heavily on analytic
argument with continuous functions on $\mathbb{R}^{d}$, this
logical formal view becomes necessary.

To begin with, we recall a theorem from field theory:
Every extension field ${\bf F}$ over  ${\mathbb Q}$
by a finite set of real numbers is  a finite algebraic extension ${\bf E}'$
of a certain  purely transcendental   extension field ${\bf E}$
over ${\mathbb Q}$,
which has the form ${\bf E} = {\mathbb Q}(X_1,
\ldots, X_m)$ where $m \ge 0$ and $X_1,
\ldots, X_m$ are algebraically independent~\cite{jacobson1985basic} (Theorem 8.35, p.~512).
${\bf F}$ is said to have
a finite transcendence degree $m$ over ${\mathbb Q}$.
It is known that $m$ is uniquely defined for ${\bf F}$.
Since 
$\rm{char}~{\mathbb Q} =0$,
the finite algebraic extension ${\bf E}'$ over ${\bf E}$
is actually simple, ${\bf E}' = {\bf E}(\beta)$ for some $\beta$,
and it is specified by a
minimal polynomial in ${\bf E}[Y]$, the polynomial ring over ${\bf E}$ with one
indeterminant $Y$.
Now given a real matrix $M$, let ${\bf F} = {\mathbb Q}(M)$ 
be the extension field by adjoining the entries of $M$.
We  consider $M$ is fixed for the problem $\PlGH(M)$,
and thus we may assume (nonconstructively) that the form
${\bf F} = {\bf E}(\beta)$
 and ${\bf E} = {\mathbb Q}(X_1,
\ldots, X_m)$ are given. (This means, among other things,
that the minimal polynomial of $\beta$ over ${\bf E}$ is given,
and all arithmetic operations can be performed on ${\bf F}$.)

Now, the computational problem $\PlGH(M)$ is the following:
Given a planar 
$G$, compute $Z_{M}(G)$ as an element in ${\bf F}$
(which is expressed as a polynomial in $\beta$ with coefficients in ${\bf E}$).
More concretely, we can show that this is equivalent to the following problem $\COUNT(M)$:
The input  is a pair $(G,x)$,
  where $G=(V,E)$ is a planar graph and $x\in {\bf F}$.
The output is\vspace{-0.1cm}
$$
\text{\#}_{M}(G,x)= \Big|\big\{\sigma:V\rightarrow
  [q]\hspace{0.08cm}: \hspace{0.08cm} \prod_{(u, v) \in E} m_{\sigma(u), \sigma(v)}=x\big\}\Big|,\label{full_COUNTM}
  $$
a non-negative integer. Note that, in this definition,
we are basically combining terms with the same 
product value in the definition of $Z_{M}(G)$.

Let $n=|E|$.
Define  $X$ to be the  set of all possible product values
appearing in $Z_{M}(G)$:
\begin{equation}\label{full_definitionreuse}
X=\left\{\prod_{i,j\in [q]}m_{ij}^{k_{ij}}\hspace{0.08cm}\Big|\hspace{0.1cm}
  \text{integers $k_{ij}\ge 0$ and $\sum_{i,j\in [q]}k_{ij}=n$}
\right\}.
\end{equation}
There are $\binom{n+q^2-1}{q^2-1} = n^{O(1)}$ many integer sequences
$(k_{i,j})$ such that  $k_{i,j}\ge 0$ and $\sum_{i,j\in [q]}k_{i,j}=n$.
$X$
  is defined as a set, not a multi-set.
After removing repeated
elements the cardinality $|X|$ is also polynomial in $n$.
 For fixed and given ${\bf F}$
 the elements in $X$ can be enumerated in polynomial time in $n$.
 (It is important that ${\bf F}$ and $q$ are all treated as fixed constants.)
It then follows from the definition that
  $\text{\#}_{M}(G,x)= 0$ for any $x\notin X$.
This gives us the following relation:
\[Z_M(G)= \sum_{x\in X} x \cdot \text{\#}_{M}(G,x),\ \ \ \text{for any 
  graph $G$,}\]
and thus, $\PlGH(M)\le \COUNT(M).$

For the other direction,
we construct, for any $p\in [|X|]$ (recall that $|X|$ is polynomial in $n$), the graph $\mathbf{T_{p}}G$.
Then,
$$
Z_M(\mathbf{T_p}G) =
Z_{\mathbf{T_{p}}(M)}(G) = \sum_{x\in X} x^p \cdot \text{\#}_{M}(G,x),\ \ \ \text{for any 
  graph $G$.}
$$
This is a Vandermonde system; it has full rank since
elements in $X$ are distinct by definition. So by
querying $\PlGH(M)$ for the
  values of $Z_M(\mathbf{T_p}G)$,
  we can solve it in polynomial time
  and get $\text{\#}_{M}(G,x)$ for\vspace{0.0015cm} every non-zero $x\in X$.
To obtain $\text{\#}_{M}(G,0)$ (if $0\in X$), we note that
$$
\sum_{x\in X} \text{\#}_{M}(G,x) =q^{|V|}.
$$
This gives us a polynomial-time reduction and thus, $\COUNT(M)\le \PlGH(M)$.
We have proved
\begin{lemma}\label{lemma: full_count}
For any fixed  $M \in \Sym{q}{}(\mathbb{R})$,  
   $\PlGH(M)\equiv \COUNT(M)$.
\end{lemma}
Thus, $\PlGH(M)$ can be identified with the 
problem of producing those polynomially many integer coefficients
in the canonical expression for $Z_M(G)$ as a sum
of (distinct) terms from $X$.

This  formalistic view has 
the advantage that we can treat the complexity 
of  $\PlGH(M)$ for  general $M$, and not restricted to algebraic numbers. 
 Thus, numbers such as $e$ or $\pi$
need not be excluded.
More importantly, in this paper this generality is essential, due to the proof techniques that we employ
(see \cref{lemma: analyticGadgets}).
In short, in this paper, treating the complexity of  $\PlGH(M)$ for  general real $M$
is not a \emph{bug} but a \emph{feature}.

However, we note that this treatment 
has the following subtlety.  
For the computational problem $\PlGH(M)$
the formalistic view demands that
${\bf F}$ be specified in the form ${\bf F} = {\bf E}(\beta)$.
Such a form exists, and its specification is of
constant size when  measured
in terms of the size  of the input graph $G$. 
However, 
in reality many basic questions for transcendental numbers
are unknown.
For example, it is still unknown whether $e + \pi$ or $e \pi$ are
rational, algebraic irrational or transcendental,
and it is open whether ${\mathbb Q}(e, \pi)$ has  transcendence degree
2 (or 1) over  ${\mathbb Q}$, i.e., whether $e$ and $\pi$ are algebraically
independent.
The formalistic view here non-constructively
assumes this information is given for ${\bf F}$.
A  polynomial time reduction $\Pi_1  \le \Pi_2$  from one problem
to another  in this setting merely
implies that the \emph{existence} of a polynomial time algorithm
for $\Pi_2$ logically implies the 
\emph{existence} of  a  polynomial time algorithm
for $\Pi_1$. We do not actually obtain such
an algorithm constructively. 

This logical detour not withstanding, if a reader
is  interested only in the complexity of $\PlGH(M)$ 
for $M \in \Sym{q}{}(\mathbb{Z}_{\geq 0})$, then the
complexity dichotomy proved in this paper holds
according to the standard definition of $\PlGH(M)$ 
for integral $M$ in terms of
the model of computation; the fact that this is proved 
in a broader setting for all real matrices $M$ 
is irrelevant. This  is akin to the situation
in analytic number theory, where one might be 
interested in a question strictly about the ordinary
integers, but the theorems are proved
in a broader setting of analysis.

%% file: A2.Preliminaries.tex

\input{A2.2.Thickening_Gadgets}
\newpage
\input{A2.3.Stretching_Gadgets}
\newpage
\input{A2.4.Bridge_Gadgets}
\newpage
\input{A2.5.Loop_Gadgets}

%% file: A2.2.Thickening_Gadgets.tex
\section{Thickening Gadgets}\label{appendix: thickeningAppendix}

In this section, we shall explore the thickening gadgets,
and gadget interpolation more rigorously.
For $m, n \ge 1$, let
$$\mathcal{P}_{m}(n) = \left\{\mathbf{x} = (x_{i})_{i \in [m]}
\in (\mathbb{Z}_{\geq 0})^{m} \hspace{0.08cm}\Big|\hspace{0.1cm}
\sum_{i \in [m]}x_{i} = n\right\}.$$

We note that
given any graph $G = (V, E)$,
\begin{equation}\label{equation: thickeningEqn}
    Z_{M}(\mathbf{T_{n}}G) = Z_{\mathbf{T_{n}}(M)}(G) = 
    \sum_{x \in X(G)}x^{n} \cdot \text{\#}_{M}(G, x)
\end{equation}
where
\begin{equation}\label{equation: thickeningX(G)}
    X(G) = \left\{\prod_{i, j\in [q]}
    M_{ij}^{k_{ij}}\hspace{0.08cm}\Big|\hspace{0.1cm}
    \mathbf{k} = (k_{ij})_{i, j \in [q]} \in
    \mathcal{P}_{q^{2}}(|E|)\right\},
\end{equation}
and
\begin{equation}\label{equation: thickeningNumMappings}
    \text{\#}_{M}(G,x)= \Big|\big\{\sigma:V\rightarrow [q]\hspace{0.08cm}:
    \hspace{0.08cm} \prod_{(u, v) \in E} M_{\sigma(u)\sigma(v)}=x\big\}\Big|.
\end{equation}

Note that given any $x \in X(G)$,
$\text{\#}_{M}(G, x)$ does not depend on $n$,
but depends only on the entries of the matrix $M$.
Generating Sets, as defined in \cref{definition: generatingSet}
let us deal with this dependence.
First, we shall prove that generating sets always exist.

\generatingSet*

\begin{lemma}\label{lemma: generatingSet}
    Every finite set $\mathcal{A} \subset \mathbb{R}_{\neq 0}$
    of non-zero real numbers has
    a generating set.
\end{lemma}
\begin{proof}
    Consider the multiplicative group $\mathcal{G}$ 
    generated by 
    the positive real numbers 
    $\{|a| : a \in \mathcal{A}\}$.
    It is  a subgroup of the multiplicative group
    $(\mathbb{R}_{> 0}, \cdot)$.
    Since  $\mathcal{A}$ is finite, and 
    $(\mathbb{R}_{> 0}, \cdot)$ is torsion-free,
    the group $\mathcal{G}$ is  a finitely generated free Abelian group, 
    and thus isomorphic to 
    $\mathbb{Z}^d$ for some $d \ge 0$.
    Let $f$ be this isomorphism from $\mathbb{Z}^{d}$ to
    the multiplicative group.
    By flipping $\pm 1$ in $\mathbb{Z}$ we may assume that this isomorphism maps the basis
    elements of $\mathbb{Z}^{d}$ to some elements
    $\{g_{t}\}_{t \in [d]}$ such that $g_{t} > 1$
    for all $t \in [d]$.
    The set $\{g_{t}\}_{t \in [d]}$ is
    a generating set.
\end{proof}

We now use \cref{lemma: generatingSet} to find a generating set
for the entries $(M_{ij})_{i, j \in [q]}$ of any 
matrix $M \in \Sym{q}{}(\mathbb{R}_{\neq 0})$.
Note that this generating set need not be unique. 
However, with respect to a fixed generating set, 
for any $M_{ij}$, there are unique integers 
$e_{ij0} \in \{0, 1\}$, and
$e_{ij1}, \dots, e_{ijd} \in \mathbb{Z}$,
such that
\begin{equation}\label{equation: generatingM}
    M_{ij} = (-1)^{e_{ij0}} \cdot g_{1}^{e_{ij1}} \cdots g_{d}^{e_{ijd}}.
\end{equation}

\begin{remark*}
    It should be noted that since $M$ is symmetric, $M_{ij} = M_{ji}$
    for all $i, j \in [q]$.
    The uniqueness of the integers $e_{ijt}$
    in \cref{equation: generatingM} then implies that
    for all $i, j \in [q]$,
    $e_{ijt} = e_{jit}$ for all $t \in [d]$.
\end{remark*}

We can now prove the following lemma.

\begin{lemma}\label{lemma: thickeningGeneralLemma}
    Let $M \in \Sym{q}{}(\mathbb{R}_{\neq 0})$
    with its entries generated by some $\{g_t\}_{t \in [d]}$,
    such that
    $$M_{ij} = (-1)^{e_{ij0}} \cdot
    g_{1}^{e_{ij1}} \cdots g_{d}^{e_{ijd}}.$$
    Furthermore, assume that $e_{ijt} \geq 0$ for all
    $i, j \in [q]$, and $t \in [d]$.
    Then, for any
    $N \in \Sym{q}{}(\mathbb{R}_{\neq 0})$
    of the form
    $$N_{ij} = (-1)^{e_{ij0}} \cdot z_{1}^{e_{ij1}}
    \cdots z_{d}^{e_{ijd}},$$
    with parameters
    $\mathbf{z} = (z_{1}, \dots, z_{d}) \in \mathbb{R}^{d}$,
    we have $\PlGH(N) \leq \PlGH(M)$.
\end{lemma}
\begin{proof}
    For any $n \geq 1$ and graph $G = (V, E)$, recall from
    \cref{equation: thickeningEqn} that
    \begin{equation}\label{Z_M-expression}
        Z_{M}(\mathbf{T_{n}}G)
        = \sum_{x \in X(G)} x^{n} \cdot \text{\#}_{M}(G, x),
    \end{equation}
    where $X(G)$ and $\text{\#}_{M}(G, x)$ are as in
    \cref{equation: thickeningX(G),equation: thickeningNumMappings}.
    Since $X(G)$ is a finite set with $|X(G)| \le |E|^{O(1)}$,
    querying $\PlGH(M)$ for $n \in [|X(G)|]$
    gives us a full-rank Vandermonde system in the variables
    $\{\text{\#}_{M}(G, x)\}_{x \in X(G)}$.
    Hence all $\text{\#}_{M}(G, x)$ can be computed in polynomial time.

    Next, each $x \in X(G)$ can be written as
    $x = \prod_{i, j \in [q]} M_{ij}^{k_{ij}}$
    for some $\mathbf{k} = (k_{ij}) \in \mathcal{P}_{q^{2}}(|E|)$.
    Since the entries of $M$ are generated by $\{g_t\}_{t \in [d]}$,
    every such $x$ admits a unique representation
    $$x = (-1)^{e^{x}_{0}} \cdot 
    g_{1}^{e^{x}_{1}} \cdots g_{d}^{e^{x}_{d}},
    \qquad e^{x}_{t} = \sum_{i, j \in [q]} k_{ij} e_{ijt},$$
    with $e^{x}_{t} \in \mathbb{Z}_{\ge 0}$ since
    all $e_{ijt}, k_{ij} \ge 0$.

    Define a map $\widehat{y}: X(G) \to \mathbb{R}$ by
    $$\widehat{y}(x)
    = (-1)^{e^{x}_{0}} \cdot 
    z_{1}^{e^{x}_{1}} \cdots z_{d}^{e^{x}_{d}},$$
    where
    $\mathbf{z} = (z_{1}, \dots, z_{d})$ are the parameters
    of $N \in \Sym{q}{}(\mathbb{R})$.
    For any $\mathbf{k} \in \mathcal{P}_{q^{2}}(|E|)$,
    this definition ensures that
    $$\widehat{y} \left(\prod_{i, j \in [q]} M_{ij}^{k_{ij}}\right)
    = \prod_{i, j \in [q]} N_{ij}^{k_{ij}}.$$
    Let
    $$Y(G) = \left\{\prod_{i, j \in [q]} N_{ij}^{k_{ij}}
    \ \middle|\ \mathbf{k} \in \mathcal{P}_{q^{2}}(|E|) \right\}.$$

    For any labeling $\sigma : V \to [q]$, define
    $k_{ij} = |\{(u,v) \in E : \sigma(u)=i, \sigma(v)=j\}|$.
    Then
    $$\widehat{y}\left( \prod_{(u,v)\in E}
    M_{\sigma(u)\sigma(v)} \right)
    = \widehat{y}\left(\prod_{i, j \in [q]}M_{ij}^{k_{ij}}\right)
    = \prod_{i, j \in [q]} N_{ij}^{k_{ij}}
    = \prod_{(u,v)\in E} N_{\sigma(u)\sigma(v)}.$$
    Hence, for each $y \in Y(G)$,
    $$\left\{ \sigma : V \to [q] \ \middle|\ 
    \prod_{(u,v)\in E} N_{\sigma(u)\sigma(v)} = y \right\}
    = \bigsqcup_{\substack{x \in X(G):\\ \widehat{y}(x)=y}}
    \left\{ \sigma : V \to [q] \ \middle|\ 
    \prod_{(u,v)\in E} M_{\sigma(u)\sigma(v)} = x \right\}.$$
    
    Therefore, we obtain
    $$\text{\#}_{N}(G, y)
    = \sum_{\substack{x \in X(G):\\ \widehat{y}(x) = y}}
    \text{\#}_{M}(G, x).$$

    Using the values $\text{\#}_{M}(G, x)$ already determined,
    we can compute
    \begin{align*}
        Z_{N}(G)
        &= \sum_{y \in Y(G)} y \cdot \text{\#}_{N}(G, y)
        = \sum_{x \in X(G)} \widehat{y}(x) \cdot \text{\#}_{M}(G, x).
    \end{align*}
    Therefore, $Z_{N}(G)$ can be computed in polynomial time
    given oracle access to $\PlGH(M)$, implying
    $\PlGH(N) \leq \PlGH(M)$.
\end{proof}



\cref{lemma: thickeningGeneralLemma} lets us prove
\cref{lemma: thickeningLemma},
but we shall first need the following
elementary lemma.

\begin{lemma}\label{lemma: MequivalentCM}
    Let $M \in \Sym{q}{}(\mathbb{R}_{\neq 0})$, and
    let $c \in \mathbb{R}_{\neq 0}$.
    Then $\PlGH(M) \equiv \PlGH(cM)$.
\end{lemma}
\begin{proof}
    Note that given any $G = (V, E)$
    $$Z_{cM}(G) = \sum_{\sigma: V \rightarrow [q]}
    \prod_{(u, v) \in E}(cM)_{\sigma(u)\sigma(v)}
    = c^{|E|} \sum_{\sigma: V \rightarrow [q]}
    \prod_{(u, v) \in E}cM_{\sigma(u)\sigma(v)}
    = c^{|E|}Z_{M}(G).$$
    Therefore, oracle access to one of $\PlGH(M)$ or $\PlGH(cM)$
    gives the other.
\end{proof}

We are now ready to prove
\cref{lemma: thickeningLemma}.

\thickeningLemma*
\begin{proof}
    Since $M \in \Sym{q}{}(\mathbb{R}_{\neq 0})$
    is generated by $\{g_t\}_{t \in [d]}$, we know that
    there exist integers $e_{ijt} \in \mathbb{Z}$ s.t.,
    $M_{ij} = (-1)^{e_{ij0}} \cdot
    g_{1}^{e_{ij1}} \cdots g_{d}^{e_{ijd}}$
    for all $i, j \in [q]$.
    Now, we let $e_{t}^{*} = \min_{i, j \in [q]}e_{ijt}$
    for all $t \in [d]$, and let
    $c = \prod_{t \in [d]}g_{t}^{-e_{t}^{*}}$.
    From \cref{lemma: MequivalentCM}, we see that
    $\PlGH(M) \equiv \PlGH(cM)$.
    Moreover, we also see that the entries of the matrix
    $cM$ are also generated by $\{g_t\}_{t \in [d]}$
    such that $(cM)_{ij} = (-1)^{e_{ij0}} \cdot
    g_{1}^{e_{ij1} - e_{1}^{*}} \cdots g_{d}^{e_{ijd} - e_{d}^{*}}$
    for all $i, j \in [q]$.
    From our choice of $e_{t}^{*}$, it follows that
    $e_{ijt} - e_{t}^{*} \geq 0$ for all $i, j \in [q]$, $t \in [d]$.
    Therefore, \cref{lemma: thickeningGeneralLemma}
    implies that $\PlGH(\mathcal{T}^{*}_{M}(\mathbf{z})) \leq \PlGH(cM)
    \equiv \PlGH(M)$ for all
    $\mathbf{z} = (z_{1}, \dots, z_{d}) \in \mathbb{R}^{d}$.
\end{proof}

\cref{lemma: simpleThickening,lemma: thickeningSingleVariable}
now follow as corollaries:
\simpleThickening*
\begin{proof}
    Let the entries of $M \in \Sym{q}{}(\mathbb{R}_{> 0})$
    be generated by some $\{g_t\}_{t \in [d]}$,
    such that $M_{ij} = g_{1}^{e_{ij1}} \cdots g_{d}^{e_{ijd}}$
    for all $i, j \in [q]$.
    Let $e_{t}^{*} = \min_{i, j \in [q]}e_{ijt}$ for $t \in [d]$.
    Then, we note from \cref{definition: mathcalTStar}, that
    $$\mathcal{T}^{*}_{M}(g_{1}^{\theta}, \dots, g_{d}^{\theta})_{ij}
    = (g_{1}^{e_{ij1}} \cdots g_{d}^{e_{ijt}})^{\theta} \cdot
    (g_{1}^{-e_{1}^{*}}\cdots g_{d}^{-e_{d}^{*}})^{\theta}
    = (M_{ij})^{\theta} \cdot c^{\theta},$$
    for $c = (g_{1}^{-e_{1}^{*}}\cdots g_{d}^{-e_{d}^{*}})$.
    Importantly, $c$ is independent of all $i, j \in [q]$.
    So, $\mathcal{T}^{*}_{M}(g_{1}^{\theta}, \dots, g_{d}^{\theta})
    = c^{\theta} \cdot \mathcal{T}_{M}(\theta)$ for all
    $\theta \in \mathbb{R}$.
    Therefore, \cref{lemma: thickeningLemma,lemma: MequivalentCM}
    imply that
    $$\PlGH(\mathcal{T}_{M}(\theta)) \equiv 
    \PlGH(\mathcal{T}^{*}_{M}(g_{1}^{\theta}, \dots, g_{d}^{\theta}))
    \leq \PlGH(M),$$
    for all $\theta \in \mathbb{R}$.
\end{proof}

\thickeningSingleVariable*
\begin{proof}
    Let $\{g_t\}_{t \in [d]}$ be a
    generating set for the
    entries of $M$, with
    $M_{ij} = 
    g_{1}^{e_{ij1}} \cdots g_{d}^{e_{ijd}}$,
    $e_{t}^{*} = \min_{i, j \in [q]}e_{ijt}$, 
    ($i, j \in [q], t \in [d]$). 
    From \cref{lemma: thickeningLemma}, 
    $\PlGH(\mathcal{T}^{*}_{M}(\mathbf{z})) 
    \leq \PlGH(M)$ for all
    $\mathbf{z} \in \mathbb{R}^{d}$.

    Define
    $\mathcal{E}_{M}: \mathbb{R}^{d} \rightarrow
    \Sym{q}{}(\mathbb{R})$, such that
    $\mathcal{E}_{M}(y_{1}, \dots, y_{d}) =
    \mathcal{T}^{*}_{M}(e^{y_{1}}, \dots, e^{y_{d}})$.
    By the definition of generators, if $M_{ij} = M_{i'j'}$ then
    $(\mathcal{E}_{M}(\mathbf{y}))_{ij} =
    (\mathcal{E}_{M}(\mathbf{y}))_{i'j'}$.
    From \cref{definition: mathcalTStar},
    we see that $\mathcal{E}_{M}(\ln(g_{1}), \dots,
    \ln(g_{d})) = \mathcal{T}^{*}_{M}(g_{1}, \dots, g_{d})
    = c \cdot M$ for some $c \in \mathbb{R}_{> 0}$.
    Define
    $\delta_{M}: \Sym{q}{}(\mathbb{R}) \rightarrow
    \mathbb{R}$, 
    by
    $$\delta_{M}(A) = \min_{\substack{i \leq j, \, i' \leq j': \\ 
    M_{ij} \neq M_{i'j'}}} |A_{ij} - A_{i'j'}|,$$
    which measures the smallest absolute difference 
    between entries of $A$ corresponding to 
    distinct entries of $M$.
    Trivially, $\delta_{M}(c \cdot M) > 0$.
    Since $\delta_{M}(A)$ is a continuous function in the
    entries of $A$, there  exists
    an open ball $\mathcal{O}$ centered at $(\ln(g_{1}), \dots,
    \ln(g_{d}))$, such that
    $\delta_{M}(\mathcal{E}_{M}(\mathbf{y})) > 0$
    for all $\mathbf{y} \in \mathcal{O}$.
    Moreover, if $M \in \Sym{q}{F}(\mathbb{R}_{> 0})$,
    we note that since the eigenvalues of a matrix
    are continuous as a function of its entries,
    there  exists an open ball $\mathcal{O}'$ centered at 
    $(\ln(g_1), \dots, \ln(g_d))$, such that
    $\det(\mathcal{E}_{M}(\mathbf{y})) \neq 0$
    for all $\mathbf{y} \in \mathcal{O}'$.
    In that case, we replace $\mathcal{O}$ 
    with the smaller ball $\mathcal{O} \cap \mathcal{O}'$
    going forward.
    
    In particular, we can find some
    $\mathbf{y}^{*} = (y_{1}^{*}, \dots, y_{d}^{*}) \in 
    \mathbb{Q}_{> 0}^{d} \cap \mathcal{O}$.
    Now,  for any $(i, j)$ and
    $(i', j')$, if $M_{ij} < M_{i'j'}$, 
    we claim that  $(\mathcal{E}_{M}(\mathbf{y}^{*}))_{ij} <
    (\mathcal{E}_{M}(\mathbf{y}^{*}))_{i'j'}$.
    Otherwise,
     by the Intermediate Value
    Theorem, there  exists some 
    $t \in [0, 1]$ such that $\mathcal{E}_{M}(\mathbf{y})_{ij} = 
    \mathcal{E}_{M}(\mathbf{y})_{i'j'}$ where
    $\mathbf{y} = t \cdot \mathbf{y}^{*} + (1 - t) \cdot
    (\ln(g_{1}), \dots, \ln(g_{d})) \in \mathcal{O}$.
    This is a contradiction
    to $\delta_{M}(\mathcal{E}_{M}(\mathbf{y})) > 0$
    for all $\mathbf{y} \in \mathcal{O}$.
    Hence,
    $$M_{ij} < M_{i'j'} \iff 
    \mathcal{E}_{M}(\ln(g_{1}), \dots, \ln(g_{d}))_{ij} <
    \mathcal{E}_{M}(\ln(g_{1}), \dots, \ln(g_{d}))_{i'j'} \iff 
    \mathcal{E}_{M}(\mathbf{y}^{*})_{ij} < 
    \mathcal{E}_{M}(\mathbf{y}^{*})_{i'j'}.$$

    Since $\mathbf{y}^{*} \in \mathbb{Q}_{> 0}^{d}$,
    we know that there exists some $\kappa \in \mathbb{Z}_{> 0}$,
    such that $\kappa \cdot \mathbf{y}^{*} \in \mathbb{Z}_{> 0}^{d}$.
    Now, we define
    $$x_{ij} = \sum_{t \in [d]}\kappa y^{*}_{t}(e_{ijt} - e_{t}^{*}),
    \qquad \text{ for all } i \leq j \in [q].$$
    Since $\kappa y^{*}_{t} \in \mathbb{Z}_{> 0}$, and
    $(e_{ijt} - e_{t}^{*}) \in \mathbb{Z}_{\geq 0}$,
    we see that $x_{ij} \in \mathbb{Z}_{\geq 0}$ for
    all $i \leq j \in [q]$.
    It is also easy to see that $e^{x_{ij}} = 
    \mathcal{E}_{M}(\kappa \cdot \mathbf{y}^{*})_{ij}$
    for all $i, j \in [q]$.
    Therefore, $M_{ij} < M_{i'j'} \iff 
    x_{ij} < x_{i'j'}$.

    We can now define $\widehat{\mathcal{T}}_{M}: \mathbb{R}_{> 0} 
    \rightarrow \Sym{q}{}(\mathbb{R})$, such that
    $\widehat{\mathcal{T}}_{M}(z)_{ij} = z^{x_{ij}}$
    for all $i, j \in [q]$.
    We see that $\widehat{\mathcal{T}}_{M}(z) = 
    \mathcal{E}_{M}(\ln(z)(\kappa \cdot \mathbf{y}^{*}))
    \leq \PlGH(M)$.
    Moreover, if $M \in \Sym{q}{F}(\mathbb{R}_{> 0})$,
    then $\det(\widehat{\mathcal{T}}_{M}(e^{\nicefrac{1}{\kappa}}))
    = \det(\mathcal{E}_{M}(\mathbf{y}^{*})) \neq 0$.
\end{proof}

%% file: A2.3.Stretching_Gadgets.tex
\section{Stretching Gadgets}\label{appendix: stretchingAppendix}

In this section, we shall explore the stretching gadgets,
and gadget interpolation more rigorously.
For $m, n \ge 1$, let
$$\mathcal{P}_{m}(n) = \left\{\mathbf{x} = (x_{i})_{i \in [m]}
\in (\mathbb{Z}_{\geq 0})^{m} \hspace{0.08cm}\Big|\hspace{0.1cm}
\sum_{i \in [m]}x_{i} = n\right\}.$$

Given a matrix $M \in \Sym{q}{pd}(\mathbb{R})$,
we know that there exists an orthonormal matrix
(not necessarily unique) 
of unit eigenvectors $H$
with entries $H_{ij} \in \mathbb{R}$,
and a diagonal matrix of eigenvalues
$D = \diag(\lambda_{1}, \dots, \lambda_{q})$,
such that $M = HDH^{\tt{T}}$.
Therefore,
$$(\mathbf{S_{n}}(M))_{ij} = 
(M^{n})_{ij} = (H_{i1}H_{j1}) \lambda_{1}^{n} + \dots +
(H_{iq}H_{jq}) \lambda_{q}^{n},$$
for all $i, j \in [q]$, and all $n \geq 1$.

It then follows that
\begin{equation}\label{eqn: stretching}
    Z_{\mathbf{S_{n}}(M)}(G) = \sum_{\mathbf{k} \in \mathcal{P}_{q}(|E|)}c_{H}(G, \mathbf{k}) \cdot \left(\lambda_{1}^{k_{1}} \cdots \lambda_{q}^{k_{q}} \right)^{n},
\end{equation}
where
\begin{equation}\label{equation: stretchingCH(G)}
    c_{H}(G, \mathbf{k}) = 
    \sum_{\sigma: V \rightarrow [q]}
    \left( \sum_{\substack{E_{1} \sqcup \dots \sqcup E_{q}
    = E\\|E_{i}| = k_{i}}} \left( \prod_{i \in [q]}
    \prod_{(u, v) \in E_{i}}H_{\sigma(u)i}H_{\sigma(v)i}
    \right) \right)
\end{equation}
depends only on $G$ and the orthogonal matrix $H$, but not on $D$.

We can now prove \cref{lemma: stretchingLemma}.

\stretchingLemma*
\begin{proof}
    Recall from \cref{eqn: stretching} that for any graph $G = (V, E)$,
    $$Z_{M}(\mathbf{S_{n}}G) = 
    \sum_{\mathbf{k} \in \mathcal{P}_{q}(|E|)}
    c_{H}(G, \mathbf{k}) \cdot \left(
    \lambda_{1}^{k_{1}} \cdots \lambda_{q}^{k_{q}}\right)^{n},$$
    where $c_{H}(G, \mathbf{k})$ is defined in
    \cref{equation: stretchingCH(G)}.

    Define the set (i.e., without multiplicity)
    $$\Lambda_{D}(G) = \left\{ \prod_{i \in [q]} \lambda_{i}^{k_{i}}
    \ \Big|\  \mathbf{k} \in \mathcal{P}_{q}(|E|) \right\}.$$
    For each $\mu \in \Lambda_{D}(G)$, define
    $$X_{D}(G, \mu) = \left\{ \mathbf{k} \in \mathcal{P}_{q}(|E|)
    \ \Big|\ \prod_{i \in [q]}\lambda_{i}^{k_{i}} = \mu \right\},
    \qquad c_{H, D}(G, \mu)
    = \sum_{\mathbf{k} \in X_{D}(G, \mu)} c_{H}(G, \mathbf{k}).$$
    Substituting, we obtain
    $$Z_{M}(\mathbf{S_{n}}G) = \sum_{\mu \in \Lambda_{D}(G)}
    c_{H, D}(G, \mu) \cdot \mu^{n}.$$

    Since $\Lambda_{D}(G)$ is a set, its elements are distinct and 
    $|\Lambda_{D}(G)| \leq |E|^{O(1)}$.
    With oracle access to $\PlGH(M)$, we can compute 
    $Z_{M}(\mathbf{S_{n}}G)$ for
    $n \in [|\Lambda_{D}(G)|]$.
    This yields a full-rank Vandermonde system, which can be solved in polynomial time
    to obtain all $c_{H, D}(G, \mu)$.

    Now consider $Z_{N}(G)$ for an arbitrary
    $N = H\Delta H^{\tt T}$, where $\Delta = D^{\theta}$ for
    $\theta \in \mathbb{R}_{\neq 0}$.
    The eigenvalues of $N$ are $(\lambda_{1}^{\theta}, \dots, \lambda_{q}^{\theta})$.
    Define $\Lambda_{\Delta}(G)$ analogously to $\Lambda_{D}(G)$, replacing
    $D_{ii} = \lambda_{i}$ by $\Delta_{ii} = \lambda_{i}^{\theta}$.
    Similarly, define $X_{\Delta}(G, \nu)$ for each $\nu \in \Lambda_{\Delta}(G)$, and
    $$c_{H, \Delta}(G, \nu)
    = \sum_{\mathbf{k} \in X_{\Delta}(G, \nu)}
    c_{H}(G, \mathbf{k}).$$
    Then
    $$Z_{N}(G) = \sum_{\nu \in \Lambda_{\Delta}(G)}
    c_{H, \Delta}(G, \nu) \cdot \nu.$$

    For any $\nu \in \Lambda_{\Delta}(G)$, there exists
    $\mathbf{k} \in \mathcal{P}_{q}(|E|)$ such that
    $$\nu = \prod_{i \in [q]}(\lambda_{i}^{\theta})^{k_{i}}
    = \left(\prod_{i \in [q]}\lambda_{i}^{k_{i}}\right)^{\theta}.$$
    Since each $\lambda_{i} > 0$ and $\theta \neq 0$,
    we have a unique $\mu = \nu^{\nicefrac{1}{\theta}} \in \Lambda_{D}(G)$.
    Consequently,
    $$X_{\Delta}(G, \nu) = X_{D}(G, \nu^{\nicefrac{1}{\theta}}),
    \qquad c_{H, \Delta}(G, \nu) = c_{H, D}(G, \nu^{\nicefrac{1}{\theta}}).$$
    Therefore, using the previously computed $c_{H, D}(G, \mu)$ values,
    we can evaluate
    $$Z_{N}(G) = \sum_{\nu \in \Lambda_{\Delta}(G)}
    c_{H, D}(G, \nu^{\nicefrac{1}{\theta}}) \cdot \nu.$$
    Hence, $\PlGH(\mathcal{S}_{M}(\theta)) \leq \PlGH(M)$ for all $\theta \in \mathbb{R}_{\neq 0}$.
\end{proof}

In~\cref{lemma: stretchingLemma}, the reduction is also valid for $\theta=0$ (trivially).
But we will not need that fact.

%% file: A2.4.Bridge_Gadgets.tex
\section{Bridge Gadgets}\label{appendix: bridgingAppendix}

In this section, we shall explore the bridging gadgets,
and gadget interpolation more rigorously.
For $m, n \ge 1$, let
$$\mathcal{P}_{m}(n) = \left\{\mathbf{x} = (x_{i})_{i \in [m]}
\in (\mathbb{Z}_{\geq 0})^{m} \hspace{0.08cm}\Big|\hspace{0.1cm}
\sum_{i \in [m]}x_{i} = n\right\}.$$
We shall now prove \cref{lemma: bridgingLemma}.

\bridgingLemma*
\begin{proof}
    We know that for any $n \geq 1$, and given any
    planar graph $G = (V, E)$,
    $$Z_{M}(\mathbf{B_{n}}G) = 
    \sum_{\sigma: V \rightarrow [q]} \prod_{(u, v) \in E}
    \sum_{x, y \in [q]}M_{\sigma(u)x}(M_{xy})^{n}M_{\sigma(v)y}.$$
    
    As in \cref{appendix: thickeningAppendix},
    we  define the \emph{set}
    $$X(G) = \left\{ \prod_{i, j \in [q]}
    M_{ij}^{k_{ij}} \Big|\hspace{0.1cm}
    \mathbf{k} \in \mathcal{P}_{q^{2}}(|E|) \right\}.$$
    For each $\mathbf{k} \in \mathcal{P}_{q^{2}}(|E|)$,
    we also define
    $$Y(G, \mathbf{k}) = 
    \left\{(\mathbf{x}, \mathbf{y}) \in [q]^{|E|} \times [q]^{|E|}
    \hspace{0.08cm}\Big|\hspace{0.1cm} 
    (\forall\ 
    i, j \in [q])~\left[\ k_{ij} = |\{t: x_{t} = i, y_{t} = j\}| \ \right]\  \right\}.$$
    Finally, for each $x \in X(G)$, we define
    $$c(G, x) = \sum_{\sigma: V \rightarrow [q]}\left(
    \sum_{\substack{\mathbf{k} \in \mathcal{P}_{q^{2}}(|E|):\\
    \prod_{i, j \in [q]}M_{ij}^{k_{ij}} = x}}\left(
    \sum_{\mathbf{x}, \mathbf{y} \in Y(G, \mathbf{k})}\left(
    \prod_{(u, v) \in E} M_{\sigma(u)x_{u}}M_{\sigma(v)y_{v}}
    \right)\right)\right).$$
    
    Now, we note that
    $$Z_{M}(\mathbf{B_{n}}G) = 
    Z_{\mathbf{B_{n}}(M)}(G) = \sum_{x \in X(G)}x^{n} \cdot
    c(G, x).$$
    Since each $x \in X(G)$ is distinct, and
    since $|X(G)| \leq |E|^{O(1)}$, we can use
    oracle access to $\PlGH(M)$ to compute
    $Z_{M}(\mathbf{B_{n}}G)$ for $n \in [|X_{M}(G)|]$.
    We will then have a full rank Vandermonde system of
    linear equations, which can be solved in
    polynomial time to find
    $c(G, x)$ for all $x \in X(G)$.

    We will now note that each $x \in X(G)$ is the product
    of entries of the matrix $M$.
    Since $M \in \text{Sym}_{q}(\mathbb{R}_{> 0})$,
    this implies that every element of $X(G)$
    is positive.
    Moreover, we see that for every $\mathbf{k} \in
    \mathcal{P}_{q^{2}}(|E|)$, and for every
    $\theta \in \mathbb{R}$,
    $$\left(\prod_{i, j \in [q]}M_{ij}^{k_{ij}}\right)^{\theta}
    = \prod_{i, j \in [q]}(M_{ij}^{\theta})^{k_{ij}}$$
    is well defined.
    Therefore, for any $\theta \in \mathbb{R}$, 
    we can compute the quantity
    $$\sum_{x \in X(G)} x^{\theta} \cdot c(G, x)
    = \sum_{\sigma: V \rightarrow [q]} \prod_{(u, v) \in E}
    \left(\sum_{x, y \in [q]}
    M_{\sigma(u)x}(M_{xy})^{\theta}M_{\sigma(v)y}\right)
    = Z_{\mathcal{B}_{M}(\theta)}(G)$$
    in polynomial time. This proves that
    $\PlGH(\mathcal{B}_{M}(\theta)) \leq 
    \PlGH(M)$ for all $\theta \in \mathbb{R}$.
\end{proof}

%% file: A2.5.Loop_Gadgets.tex
\section{Loop Gadgets}\label{appendix: loopAppendix}

In this section, we shall explore the loop gadgets,
and gadget interpolation more rigorously.
For $m, n \ge 1$, let
$$\mathcal{P}_{m}(n) = \left\{\mathbf{x} = (x_{i})_{i \in [m]}
\in (\mathbb{Z}_{\geq 0})^{m} \hspace{0.08cm}\Big|\hspace{0.1cm}
\sum_{i \in [m]}x_{i} = n\right\}.$$

We recall that 
$$Z_{M}(\mathbf{L_{n}}G) = 
\sum_{\sigma: V \rightarrow [q]} \prod_{(u, v) \in E}
(M_{\sigma(u)\sigma(u)})^{n} M_{\sigma(u) \sigma(v)}
(M_{\sigma(v) \sigma(v)})^{n} = Z_{\mathbf{L_{n}}(M)}(G),
\quad \text{ for all } n \geq 1,$$
where $\mathbf{L_{n}}(M) = (\diag(M))^{n} \cdot M \cdot (\diag(M))^{n}$.

We can now prove
\cref{lemma: loopLemma}.
The proof is very similar to the proof of
\cref{lemma: thickeningLemma} in
\cref{appendix: thickeningAppendix}.
We first prove the following lemma.

\begin{lemma}\label{lemma: loopLemmaGeneral}
    Let $M \in \Sym{q}{}(\mathbb{R}_{> 0})$.
    Let $A \in \Sym{q}{}(\mathbb{R}_{\geq 0})$ be
    a diagonal matrix with positive diagonal entries,  
    that are generated by some $\{g_t\}_{t \in [d]}$,
    such that $A_{ii} = g_{1}^{e_{ii1}} \cdots g_{d}^{e_{iid}}$
    for all $i \in [q]$.
    Furthermore, assume that $e_{iit} \geq 0$ for all
    $i \in [q]$, $t \in [d]$, and that
    $\PlGH(A^{n}MA^{n}) \leq \PlGH(M)$ for all $n \geq 1$.
    Then, for any symmetric matrix
    $B \in \Sym{q}{}(\mathbb{R})$
    of the form
    $$B_{ii} = z_{1}^{e_{ii1}}
    \cdots z_{d}^{e_{iid}},$$
    with parameters
    $\mathbf{z} = (z_{1}, \dots, z_{d}) \in \mathbb{R}^{d}$,
    we have $\PlGH(BMB) \leq \PlGH(M)$.
\end{lemma}
\begin{proof}
    We note that
    $$Z_{A^{n}MA^{n}}(G) = 
    \sum_{\sigma: V \rightarrow [q]}
    \left(\prod_{(u, v) \in E} M_{\sigma(u) \sigma(v)} \right)
    \left(\prod_{v \in V}(A_{\sigma(v) \sigma(v)})
    ^{\deg(v) \cdot n} \right),$$
    for all $n \geq 1$.

    We now define the \emph{set}
    $$X(G) = \left\{\prod_{i\in [q]}
    A_{ii}^{k_{i}}\hspace{0.08cm}\Big|\hspace{0.1cm}
    \mathbf{k} = (k_{i})_{i \in [q]} \in
    \mathcal{P}_{q}(2|E|)\right\}, \text{ and }$$
    $$c_{A, M}(G, x) = \sum_{\substack{\sigma: V \rightarrow [q]:\\
    x = \prod_{v}(A_{\sigma(v) \sigma(v)})^{\deg(v)}}}
    \prod_{(u, v) \in E}M_{\sigma(u) \sigma(v)}.$$
    We see that
    $$Z_{A^{n}MA^{n}}(G) = 
    \sum_{x \in X(G)} x^{n} \cdot c_{A, M}(G, x).$$
    Crucially, we note that $c_{A, M}(G, x)$ is
    independent of $n$.
    Since $X(G)$ is a set,
    its elements are distinct and 
    $|X(G)| \leq |E|^{O(1)}$.
    From our assumption that
    $\PlGH(A^{n}MA^{n}) \leq \PlGH(M)$,
    oracle access to $\PlGH(M)$
    allows us to compute
    $Z_{A^{n}MA^{n}}(G)$ for
    $n \in [|X(G)|]$.
    This yields a full-rank Vandermonde system,
    which can be solved in polynomial time
    to obtain all $c_{A, M}(G, x)$.

    Next, each $x \in X(G)$ can be written as
    $x = \prod_{i \in [q]}A_{ii}^{k_{i}}$
    for some $\mathbf{k} = (k_{i})_{i \in [q]}
    \in \mathcal{P}_{q}(2|E|)$.
    Since $\{g_t\}_{t \in [d]}$ is a generating
    set of $\{A_{ii}\}_{i \in [q]}$, we note that
    each $x \in X(G)$ admits a unique representation
    $$x = g_{1}^{e_{1}^{x}} \cdots g_{d}^{e_{d}^{x}},
    \qquad e_{t}^{x} = \sum_{i \in [q]}k_{i}e_{iit},$$
    with $e_{t}^{x} \in \mathbb{Z}_{\geq 0}$,
    since all $e_{iit}$, $k_{i} \geq 0$.

    We now define the map
    $\widehat{y}: X(G) \rightarrow \mathbb{R}$ by
    $$\widehat{y}(x) = z_{1}^{e_{1}^{x}} \cdots z_{d}^{e_{d}^{x}},$$
    where $\mathbf{z} = (z_{1}, \dots, z_{d})$ are the parameters
    of $B \in \Sym{q}{}(\mathbb{R})$.
    Then, for any $\mathbf{k} \in \mathcal{P}_{q}(2|E|)$,
    this definition ensures that
    $$\widehat{y}\left(\prod_{i \in [q]}A_{ii}^{k_{i}}\right)
    = \prod_{i \in [q]}B_{ii}^{k_{i}}.$$
    Now, we define the \emph{set}
    $$Y(G) = \left\{\prod_{i\in [q]}
    B_{ii}^{k_{i}}\hspace{0.08cm}\Big|\hspace{0.1cm}
    \mathbf{k} = (k_{i})_{i \in [q]} \in
    \mathcal{P}_{q}(2|E|)\right\}.$$

    For any labeling $\sigma: V \rightarrow [q]$, define
    $\mathbf{k}$ such that $k_{i} = 
    \sum_{v \in V: \sigma(v) = i} \deg(v)$.
    Then,
    $$\widehat{y}\left(\prod_{v \in V}
    (A_{\sigma(v) \sigma(v)})^{\deg(v)} \right)
    = \widehat{y} \left(\prod_{i \in [q]}
    A_{ii}^{k_{i}}\right)
    = \prod_{i \in [q]} B_{ii}^{k_{i}}
    = \prod_{v \in V} (B_{\sigma (v) \sigma (v)})^{\deg(v)}.$$
    Hence, for each $y \in Y(G)$,
    $$\left\{ \sigma : V \to [q] \ \middle|\ 
    \prod_{v\in V} (B_{\sigma(v)\sigma(v)})^{\deg(v)} = y \right\}
    = \bigsqcup_{\substack{x \in X(G):\\ \widehat{y}(x)=y}}
    \left\{ \sigma : V \to [q] \ \middle|\ 
    \prod_{v \in V} (A_{\sigma(v)\sigma(v)})^{\deg(v)} = x \right\}.$$
    
    Therefore, we obtain
    $$c_{B, M}(G, y) = \sum_{\substack{x \in X(G):\\
    \widehat{y}(x) = y}} c_{A, M}(G, x).$$
    Using the values of $c_{A, M}(G, x)$ we have already determined,
    we can compute
    $$Z_{BMB}(G) 
    = \sum_{y \in Y(G)}y \cdot c_{B, M}(G, y)
    = \sum_{x \in X(G)} \widehat{y}(x) \cdot c_{A, M}(G, x).$$
    Therefore, $\PlGH(BMB) \leq 
    \PlGH(M)$.
\end{proof}

Just as was the case in
\cref{appendix: thickeningAppendix}, we can use
\cref{lemma: loopLemmaGeneral,lemma: MequivalentCM} to prove
\cref{lemma: loopLemma}.

\loopLemma*
\begin{proof}
    We let $A = \diag(M)$ be the diagonal
    matrix with the same diagonal entries as $M$.
    Since the diagonal entries of $A$
    are generated by $\{g_t\}_{t \in [d]}$, we know that
    there exist integers $e_{iit} \in \mathbb{Z}$ s.t.,
    $A_{ii} = g_{1}^{e_{ii1}} \cdots g_{d}^{e_{iid}}$
    for all $i \in [q]$.
    Now, we let $e_{t}^{*} = \min_{i \in [q]}e_{iit}$
    for all $t \in [d]$, and let
    $c = \prod_{t \in [d]}g_{t}^{-e_{t}^{*}}$.
    From \cref{lemma: MequivalentCM}, we see that
    $\PlGH((c A)^{n} M (c A)^{n}) 
    \equiv \PlGH(A^{n} M A^{n})$.
    Since $\mathbf{L_{n}}(M) = A^{n}MA^{n}$, we know that
    $\PlGH(A^{n}MA^{n}) \leq \PlGH(M)$ for all $n \geq 1$.
    So, it follows that $\PlGH((c A)^{n}M(c A)^{n}) \leq \PlGH(M)$
    for all $n \geq 1$ as well.
    
    Moreover, we also see that the diagonal entries of the matrix
    $(c A)$ are also generated by
    $\{g_t\}_{t \in [d]}$
    such that $(c A)_{ii} =
    g_{1}^{e_{ii1} - e_{1}^{*}} \cdots g_{d}^{e_{iid} - e_{d}^{*}}$
    for all $i \in [q]$.
    From our choice of $e_{t}^{*}$, it follows that
    $e_{iit} - e_{t}^{*} \geq 0$ for all $i \in [q]$, $t \in [d]$.
    Therefore, \cref{lemma: loopLemmaGeneral}
    implies that $\PlGH(\mathcal{L}^{*}_{M}
    (\mathbf{z})) \leq \PlGH(M)$
    for all $\mathbf{z} = (z_{1}, \dots, z_{d}) \in \mathbb{R}^{d}$.
\end{proof}

The proof of \cref{lemma: loopLemmaSimple} now follows
as an immediate corollary.

\loopLemmaSimple*
\begin{proof}
    Let the diagonal entries of 
    $M \in \Sym{q}{}(\mathbb{R}_{> 0})$
    be generated by some $\{g_t\}_{t \in [d]}$,
    such that $M_{ii} = g_{1}^{e_{ii1}} \cdots g_{d}^{e_{iid}}$
    for all $i, j \in [q]$.
    Let $e_{t}^{*} = \min_{i\in [q]}e_{iit}$ for $t \in [d]$.
    Then, we note from \cref{definition: mathcalLStar}, that
    $$\mathcal{D}^{*}_{M}(g_{1}^{\theta}, \dots, g_{d}^{\theta})_{ii}
    = (g_{1}^{e_{ii1}} \cdots g_{d}^{e_{iit}})^{\theta} \cdot
    (g_{1}^{-e_{1}^{*}}\cdots g_{d}^{-e_{d}^{*}})^{\theta}
    = (M_{ii})^{\theta} \cdot c^{\theta},$$
    for $c = (g_{1}^{-e_{1}^{*}}\cdots g_{d}^{-e_{d}^{*}})$.
    Importantly, $c$ is independent of $i \in [q]$,
    so, $\mathcal{D}^{*}_{M}(g_{1}^{\theta}, \dots, g_{d}^{\theta})
    = c^{\theta} \cdot \mathcal{D}_{M}(\theta)$ for all
    $\theta \in \mathbb{R}$.
    Therefore, \cref{lemma: loopLemma,lemma: MequivalentCM}
    imply that
    $$\PlGH(\mathcal{L}_{M}(\theta)) \equiv
    \PlGH(\mathcal{D}_{M}(\theta) 
    \cdot M \cdot \mathcal{D}_{M}(\theta)) 
    \equiv \PlGH(\mathcal{L}^{*}_{M}(g_{1}^{\theta}, 
    \dots, g_{d}^{\theta}))
    \leq \PlGH(M),$$
    for all $\theta \in \mathbb{R}$.
\end{proof}

%% file: A4.Diagonal_Distinct_Hardness.tex
\section{Hardness of diagonal-distinct matrices}\label{appendix: diagonalDistinctHardness}

In this section, we will prove
\cref{theorem: diagonalDistinctHardness} 
by completing the induction step.

We will first prove \cref{theorem: inductionStep}.

\inductionStep*
\begin{proof}
	We first use \cref{lemma: thickeningSingleVariable} to
	find integers $(x_{ij})$, s.t.,
	$M_{ij} < M_{i'j'} \iff x_{ij} < x_{i'j'}$.
	Let $x^{*} = \min_{i, j}x_{ij}$,
	then
	$x_{ij} = x^{*} \iff M_{ij} = m^{*}$.
	As in the proof of \cref{theorem: magnitudeDistinctHardness},
	we may assume that $x^{*} = 0$,
	since $\PlGH(z^{-x^{*}}\widehat{\mathcal{T}}_{M}(z)) 
	\equiv \PlGH(\widehat{\mathcal{T}}_{M}(z))$.
	We also know from \cref{lemma: thickeningSingleVariable}
	that $\det(\widehat{\mathcal{T}}_{M}(z))$ is not the zero
	polynomial.
	So, there exists some $z^{*} \in \mathbb{R}$, s.t., 
	for all $z > z^{*}$, $\det(\widehat{\mathcal{T}}_{M}(z)) \neq 0$.
	We also note that since $M$ is strongly diagonal-distinct,
	$x_{i_{1}i_{1}} \notin \{x_{i_{2}j}: j \in [q]\}$
	for all $i_{1} \neq i_{2} \in [q]$.
	So, $f_{i_{1}i_{2}}(z) = 
	\sum_{j}z^{2x_{i_{1}j}} - \sum_{j}z^{2x_{i_{2}j}}$
	is not the zero polynomial for all 
	$i_{1} \neq i_{2} \in [q]$.
	So, there must exist some $z^{**}$,
	s.t., for all $z > z^{**}$,
	$f_{i_{1}i_{2}}(z) \neq 0$
	for all $i_{1} \neq i_{2} \in [q]$.
	
	We now pick some prime integer $p >
	\max(q, z^{*}, z^{**})$.
	From our choice of $p > z^{*}$, we see that
	$M' = \mathbf{S_{2}}(\widehat{\mathcal{T}}_{M}(p))
	= (\widehat{\mathcal{T}}_{M}(p))^{2}
	\in\Sym{q}{pd}(\mathbb{Z}_{> 0})$.
	For any $i \in [q]$, let $c(i) = 
	|\{j \in [q]: M_{ij} = m^{*} \}|$.
	Then, we see that $M'_{ii} \equiv c(i) \mod p$.
	From our choice of $M$ and $p > q$, we see that
	$M'_{i^{*}i^{*}} \equiv 0 \mod p$,
	and $M'_{i_{1}^{\dagger}i_{1}^{\dagger}} \not\equiv 
	0 \mod p$, and
	$M'_{i_{2}^{\dagger}i_{2}^{\dagger}} \not\equiv 
	0 \mod p$.
	We let $\{g_t\}_{t \in [d]}$ be the set of
	primes that divide the diagonal entries of $M'$.
	We name
	$g_{1} = p$, and consider
	$M'' = \mathcal{L}^{*}_{M'}(0, 1, \dots, 1)$.
	
	We consider the principal sub-matrix $N$ of $M'$
	formed by the rows and columns in the  set
	$I(N) = \{i \in [q]: M'_{ii} \not\equiv 0 \mod p \}$.
	We note that $i^{*} \notin I(N)$, and
	$i_{1}^{\dagger}, i_{2}^{\dagger} \in I(N)$,
	so $2 \leq |I(N)| < q$.
	Since $M' \in \Sym{q}{pd}(\mathbb{R}_{> 0})$,
	$N \in \Sym{|I(N)|}{pd}(\mathbb{R}_{> 0})$, and
	from our choice of $p > z^{**}$, we also find that
	$M'$ (and consequently $N$) is diagonal-distinct.
	So, from \nameref{induction: inductionHyp},
	we see that $\PlGH(N) \equiv \PlGH(M'')
	\leq \PlGH(M)$ is $\#$P-hard.
\end{proof}

We will now prove \cref{lemma: diagDistinctImpliesMatrixProperties},
but first, we shall need the following lemma:

\begin{lemma}\label{lemma: intraRowDistinctLoop}
	Let $M \in \Sym{q}{pd}(\mathbb{R}_{> 0})$ be
	diagonal distinct.
	Let $\mathcal{F}$ be some countable family of
	$\Sym{q}{}(\mathbb{R})$-polynomials, such that
	$\rho(M) \neq 0$ for all $\rho \in \mathcal{F}$.
	Then, there exists some
	intra-row-distinct $N \in \Sym{q}{pd}(\mathbb{R}_{> 0})$
	such that $\rho(N) \neq 0$ for all $\rho \in \mathcal{F}$, and
	$\PlGH(N) \leq \PlGH(M)$.
\end{lemma}
\begin{proof}
	Since $M$ is diagonal distinct, 
	we may assume that
	$\phi_{\rm{diag}} \in \mathcal{F}$ (function $\phi_{\rm{diag}}$ from 
	\cref{defn:phi-psi-function-for-diagonal-distinct}).
	We will now consider $\mathbf{L_{n}}(M)$ for $n \geq 1$.
	We note that for all $i \in [q]$, 
	and $j_{1} \neq j_{2} \in [q]$,
	$$\frac{\mathbf{L_{n}}(M)_{ij_{1}}}
	{\mathbf{L_{n}}(M)_{ij_{2}}} = 
	\frac{M_{ij_{1}}(M_{ii} \cdot M_{j_{1}j_{1}})^{n}}
	{M_{ij_{2}}(M_{ii} \cdot M_{j_{2}j_{2}})^{n}} = 
	\frac{M_{ij_{1}}}{M_{ij_{2}}} \left(\frac
	{M_{j_{1}j_{1}}}{M_{j_{2}j_{2}}} \right)^{n}.$$
	Since $M$ is diagonal distinct, we note that
	$(\nicefrac{M_{j_{1}j_{1}}}{M_{j_{2}j_{2}}}) \neq 1$.
	So for large  $n$, 
	$\mathbf{L_{n}}(M)_{ij_{1}} \neq 
	\mathbf{L_{n}}(M)_{ij_{2}}$.
	
	This implies that for each $i \in [q]$,
	and $j_{1} \neq j_{2} \in [q]$, there exists some
	$\theta_{ij_{1}j_{2}} \in \mathbb{R}$ such that
	$(\mathcal{L}_{M}(\theta_{ij_{1}j_{2}}))_{ij_{1}} - 
	(\mathcal{L}_{M}(\theta_{ij_{1}j_{2}}))_{ij_{2}} \neq 0$.
	Let $\zeta_{i, j_{1}, j_{2}}: A \mapsto (A_{ij_{1}} - A_{ij_{2}})$ for all $i \in [q], j_{1} \neq j_{2} \in [q]$.
	
	Since $M \in \Sym{q}{pd}(\mathbb{R}_{> 0})$,
	and eigenvalues are a continuous function of the entries
	of a matrix, and $\mathcal{L}_{M}(0) = M$,
	we know that there exists an open
	neighborhood $\mathcal{O}$ 
	around $\theta = 0$, s.t., for all $\theta \in \mathcal{O}$,
	$\mathcal{L}_{M}(\theta) \in \Sym{q}{pd}(\mathbb{R}_{> 0})$.
	We can now use \cref{lemma: analyticGadgets} to find
	$\theta^{*} \in \mathcal{O}$, such that
	$N = \mathcal{L}_{M}(\theta^{*}) \in 
	\Sym{q}{pd}(\mathbb{R}_{> 0})$ satisfies:
	$\rho(N) \neq 0$ for all $\rho \in \mathcal{F}$,
	$\zeta_{ij_{1}j_{2}}(N) \neq 0$
	for all $i \in [q], j_{1} \neq j_{2} \in [q]$,
	and $\PlGH(N) \leq \PlGH(M)$.
	Thus,
	$\phi_{\rm{row}}(N) = \prod_{i \in [q]}\prod_{j_{1} < j_{2} \in [q]}
	\zeta_{ij_{1}j_{2}}(N) \neq 0$, i.e., $N$ is intra-row-distinct.
	So, $N$ is the required matrix.
\end{proof}

\diagDistinctImpliesMatrixProperties*
\begin{proof}
	By \cref{lemma: intraRowDistinctLoop},
	there exists a diagonal distinct and intra-row-distinct
	$M' \in \Sym{q}{pd}(\mathbb{R}_{> 0})$
	with $\PlGH(M') \leq \PlGH(M)$.
	Using \cref{lemma: I-closeStretching},
	we obtain a diagonal distinct, intra-row-distinct,
	and I-close
	$N \in \Sym{q}{pd}(\mathbb{R}_{> 0})$
	satisfying $\PlGH(N) \leq \PlGH(M') \leq \PlGH(M)$.
	Since $N_{ii}$ are all distinct and strictly larger
	than all off-diagonal entries,
	$N$ is also strongly diagonal-distinct.
	So, this $N$ is the required matrix.
\end{proof}

Now, we will begin the proof of
\cref{theorem: diagDistinctAllProperties}.
We shall break this theorem into several
smaller lemmas.

\begin{lemma}\label{lemma: diagDistinctAntiDiagonal}
	Let $N \in \Sym{q}{pd}(\mathbb{R}_{> 0})$ be a diagonal dominant, strongly diagonal distinct,
	intra-row-distinct.
	Then, either $\PlGH(N)$ is $\#$P-hard, or
	after a simultaneous permutation of rows/columns,
	\begin{specialenumerate}\setcounter{specialenumeratei}{0}
		\item \label{item: antiDiagMin}
		$ N_{ii^{\rm R}} = \min_{i,j \in [q]} N_{ij}$
		for all $i \in [q]$, where $i^{\rm R} = q+1-i$.
	\end{specialenumerate}
\end{lemma}
\begin{proof}
	Let $m^{*} = \min_{i,j \in [q]} N_{ij}$.
	As $N$ is diagonal dominant, $m^{*}$ must occur
	off diagonal,
	i.e., there exist distinct
	$i^{\dagger} \neq j^{\dagger}$ with
	$N_{i^{\dagger} j^{\dagger}} = 
	N_{j^{\dagger} i^{\dagger}} = m^{*}.$
	If there exists $i^{*}$ such that
	$N_{i^{*}j} > m^{*}$ for all $j$,
	then by \cref{theorem: inductionStep},
	$\PlGH(N) \leq \PlGH(M)$ is $\#$P-hard.
	
	Otherwise, each row of $N$ must
	contain at least one occurrence of $m^{*}$.
	Since $N$ is intra-row-distinct and symmetric,
	this occurrence is unique:
	for each $i \in [q]$, there exists a unique $r(i) \neq i$
	such that $N_{ir(i)} = m^{*}$.
	Finally, in order to  get \ref{item: antiDiagMin},
	we rename the domain elements $[q]$
	to satisfy $r(i) = q+1-i = i^{\rm{R}}$.
	As $r(1) \ne 1$, we rename $\{2, \ldots, q\}$, a
	simultaneous permutation of rows and columns, to have $r(1) = q$.
	If $q > 2$, then $r(2) \neq 1, 2, q$,
	since $N$ is  intra-row-distinct and symmetric.
	Then we rename $\{3, \ldots, q-1\}$ to get $r(2) = q-1$.
	Continuing inductively, we have
	$r(i) = i^{\rm{R}}$ for all $i \in [q]$.
\end{proof}

\begin{lemma}\label{lemma: diagDistinctRestrictions}
	Let $N$ satisfy all the properties of~\cref{lemma: diagDistinctAntiDiagonal}.
	Then, either $\PlGH(N)$ is $\#$P-hard, or
	after a simultaneous permutation of rows/columns,
	$N$ also satisfies:
	\begin{specialenumerate}\setcounter{specialenumeratei}{1}
		\item \label{item: diagD}
		$\exists\, d \in \mathbb{R}$ with 
		$N_{ii}N_{i^{\rm R}i^{\rm R}} = d$ for all $i$,
		\item \label{item: rowCi}
		$\forall\, i \in [q],\ \exists\, c_i \in \mathbb{R}$ such that $N_{ij}N_{ij^{\rm R}} = c_i$ for all $j$,
		\item \label{item: diagOrder}
		$N_{11} > N_{22} > \dots > N_{qq}$.
	\end{specialenumerate}
\end{lemma}
\begin{proof}
	Define $\delta: \Sym{q}{}(\mathbb{R}) \rightarrow \mathbb{R}$ as:
	$$\delta(A) = \min_{i \in [q]}\left(
	\left(\min_{j \neq i^{\rm{R}} \in [q]}A_{ij}\right) 
	- A_{i i^{\rm{R}}} \right).$$
	Since $N$ is intra-row-distinct, 
	and satisfies \ref{item: antiDiagMin},
	$\delta(N) > 0$.
	
	Consider the gadget $\mathcal{L}_{N}(\theta)$
	from \cref{definition: mathcalL}.
	As its entries vary continuously in
	$\theta$ and $\mathcal{L}_{N}(0) = N$,
	there exists an open neighborhood $\mathcal{O}$ of $0$ such that
	$\delta(\mathcal{L}_{N}(\theta)) > 0$, and 
	$\mathcal{L}_{N}(\theta)\in \Sym{q}{pd}(\mathbb{R}_{>0})$ 
	for all $\theta \in \mathcal{O}$.
	
	Suppose  there exists some $\theta \in \mathbb{R}$,
	and $i \neq j$ with 
	$(\mathcal{L}_{N}(\theta))_{ii^{\rm R}} \neq 
	(\mathcal{L}_{N}(\theta))_{jj^{\rm R}}$.
	We note that $\zeta: A \mapsto A_{i i^{\rm R}} - A_{j j^{\rm R}}$
	is a $\Sym{q}{}(\mathbb{R})$-polynomial.
	So, by \cref{lemma: analyticGadgets}, 
	we can choose $\theta^{*} \in \mathcal{O}$
	so that $N' = \mathcal{L}_{N}(\theta^{*})$ satisfies:
	$\rho(N') \neq 0$ for all
	$\rho \in \{\psi_{\rm{diag}}, \phi_{\rm{row}}, \zeta\}$, where
	$\psi_{\rm{diag}}, \phi_{\rm{row}}$ are as in
	\cref{defn:phi-psi-function-for-diagonal-distinct,defn:phi-function-for-intra-row-distinct}.
	Since $\theta^{*} \in \mathcal{O}$,
	our choice of $\mathcal{O}$ ensures that
	$\delta(N') > 0$, so
	the minimum entry $m^{*}$ of $N'$ 
	occurs on some anti-diagonal entry,
	as every entry not on the anti-diagonal is greater than some other entry.
	We now let $i^{*} = \arg\max_{i \in [q]}N'_{i i^{\rm{R}}}$
	(break ties arbitrarily).
	Since $\zeta(N') \neq 0$ we have $N'_{i^{*}(i^{*})^{\rm{R}}} 
	> m^{*}$. 
	Since $\delta(N') > 0$,
	we have
	$N'_{i^{*}j} \geq N'_{i^{*}(i^{*})^{\rm{R}}}$ for all $j \in [q]$.
	Thus, $N'_{i^{*}j} > m^*$ for all $j \in [q]$. By \cref{theorem: inductionStep}, 
	$\PlGH(N')\leq \PlGH(N)$ is $\#$P-hard.
	
	Otherwise, we have $(\mathcal{L}_{N}(\theta))_{i i^{\rm{R}}}
	= (\mathcal{L}_{N}(\theta))_{j j^{\rm{R}}}$
	for all $i,j,\theta$.
	Then for all $i,j,\theta$,
	$$(N_{ii}N_{i^{\rm{R}} i^{\rm{R}}})^{\theta}N_{i i^{\rm{R}}} 
	= (N_{jj}N_{j^{\rm{R}} j^{\rm{R}}})^{\theta}N_{j j^{\rm{R}}},$$
	which implies the existence of a constant $d$ with
	$N_{ii}N_{i^{\rm{R}} i^{\rm{R}}} = d$ for all $i$, proving
	\ref{item: diagD}.
	
	Now consider the gadget $\mathcal{B}_{N}(\theta)$
	from \cref{definition: mathcalB}. 
	Fix an arbitrary pair $j, j' \in [q]$, and
	assume that $(\mathcal{B}_{N}(\theta))_{j j^{\rm{R}}} \neq 
	(\mathcal{B}_{N}(\theta))_{j' (j')^{\rm{R}}}$
	for some $\theta$.
	We note that $\mathcal{B}_{N}(1) = N^{3}
	\in \Sym{q}{pd}(\mathbb{R}_{> 0})$.
	As the entries of $\mathcal{B}_{N}(\theta)$
	vary continuously in $\theta$, we note that
	there exists an open neighborhood
	$\mathcal{O}'$ of $1$, such that
	$\mathcal{B}_{N}(\theta) \in \Sym{q}{pd}(\mathbb{R}_{> 0})$
	for all $\theta \in \mathcal{O}'$.
	So, for all $\theta \in \mathcal{O}'$,
	the matrix $\mathcal{S}_{\mathcal{B}_{N}(\theta)}
	(\nicefrac{1}{3})$
	is well-defined, and moreover, its entries vary
	continuously in $\theta$.
	Since $\mathcal{S}_{\mathcal{B}_{N}(1)}
	(\nicefrac{1}{3}) = N$, we see that there
	exists an open neighborhood $\mathcal{O}''
	\subset \mathcal{O}'$ of $1$, 
	such that for all $\theta \in \mathcal{O}''$,
	$\mathcal{S}_{\mathcal{B}_{N}(\theta)}
	(\nicefrac{1}{3}) \in \Sym{q}{pd}(\mathbb{R}_{> 0})$,
	and
	$\delta(\mathcal{S}_{\mathcal{B}_{N}(\theta)}
	(\nicefrac{1}{3})) > 0$.
	
	Now, since 
	$(\mathcal{B}_{N}(\theta))_{j j^{\rm{R}}} \neq 
	(\mathcal{B}_{N}(\theta))_{j' (j')^{\rm{R}}}$
	for some $\theta$ by our assumption,
	we can use \cref{lemma: analyticGadgets} to
	find some $\theta' \in \mathcal{O}''$
	such that $N' = \mathcal{B}_{N}(\theta')$
	satisfies:
	$(N')_{j j^{\rm{R}}}
	\neq (N')_{j' (j')^{\rm{R}}}$.
	We note that from our choice of
	$\mathcal{O}''$,
	$N' \in \Sym{q}{pd}(\mathbb{R}_{> 0})$,
	$\mathcal{S}_{N'}(\nicefrac{1}{3}) \in 
	\Sym{q}{pd}(\mathbb{R}_{> 0})$,
	and $\delta(\mathcal{S}_{N'}(\nicefrac{1}{3})) > 0$.
	Now, since the entries of
	$\mathcal{S}_{N'}(\theta)$ are continuous
	in $\theta$, we note that there exists some
	open neighborhood $\mathcal{O}'''$ of
	$\nicefrac{1}{3}$, such that
	$\mathcal{S}_{N'}(\theta) \in 
	\Sym{q}{pd}(\mathbb{R}_{> 0})$,
	and $\delta(\mathcal{S}_{N'}(\theta)) > 0$
	for all $\theta \in \mathcal{O}'''$.
	So, we can use \cref{lemma: analyticGadgets} again to
	find some $\theta^{*} \in \mathcal{O}'''$
	such that $N'' = \mathcal{S}_{N'}(\theta^{*})
	\in \Sym{q}{pd}(\mathbb{R}_{> 0})$
	satisfies:
	$\delta(N'') > 0$, and
	$(N'')_{j j^{\rm{R}}} \neq 
	(N'')_{j' (j')^{\rm{R}}}$.
	From here, we can follow along the same argument
	as we did with $\mathcal{L}_{N}(\theta)$
	to show that
	$\PlGH(N)$ is $\#$P-hard.
	
	Otherwise, we assume
	equality $(\mathcal{B}_{N}(\theta))_{j j^{\rm{R}}} = 
	(\mathcal{B}_{N}(\theta))_{j' (j')^{\rm{R}}}$ holds for all $\theta$, giving
	$$\sum_{x,y \in [q]} N_{x j}N_{y j^{\rm{R}}}(N_{xy})^{\theta}
	= \sum_{x,y \in [q]} N_{x j'}N_{y (j')^{\rm{R}}}(N_{xy})^{\theta}.$$
	Since $N$ is strongly diagonal-distinct,
	equality for all $\theta$ forces equality of coefficients
	of each $(N_{ii})^{\theta}$ term. To see this, we can move the term
	$(N_{ii})^{\theta}$ with coefficient $N_{ij}N_{i j^{\rm{R}}} - N_{i j'}N_{i (j')^{\rm{R}}}$ to one side of the equation, and leave everything else
	on the other side. If this coefficient is nonzero, then
	no matter what cancellation may occur on the other side of the equation, 
	it cannot match exactly the order of $(N_{ii})^{\theta}$ as $\theta
	\rightarrow \infty$. Hence
	$N_{ij}N_{i j^{\rm{R}}} = N_{i j'}N_{i (j')^{\rm{R}}}$
	for all $i \in [q]$, and since this is valid for all $j, j' \in [q]$, 
	we have established \ref{item: rowCi}.
	
	We will now use \ref{item: diagD} to establish
	\ref{item: diagOrder} by renaming the
	domain elements $[q]$. The property \ref{item: diagOrder} is trivial
	if $q=1$; we may assume $q \ge 2$.
	Let $j = \arg\max_{i}(N_{ii})$; as $N$ is diagonal-distinct,
	this $j$ is uniquely defined.
	Then, \ref{item: diagD} implies that
	$\arg\min_{i}(N_{ii}) = j^{\rm{R}}$. 
	Since $N$ is diagonal-distinct, $j^{\rm{R}} \ne j$.
	Now, we permute the rows/columns of $N$
	with $\sigma \in S_{q}$, s.t., 
	$\sigma(1) = j$, $\sigma(j) = 1$,
	$\sigma(q) = j^{\rm{R}}$, $\sigma(j^{\rm{R}}) = q$,
	and $\sigma(i) = i$ for all other $i \in [q]$.
	(We note that if $j=1$, this $\sigma$ is the identity map.
	If $j=q$, the four equations reduce to just $\sigma(1) = q$ and 
	$\sigma(q) = 1$. As $j \ne j^{\rm{R}}$, these four equations
	are consistent in all cases.) 
	We note that after this permutation,
	the matrix $N$ still satisfies
	\ref{item: antiDiagMin}, \ref{item: diagD},
	\ref{item: rowCi},
	and also satisfies
	$N_{11} > N_{ii} > N_{qq}$ for $2 \leq i \leq q-1$.
	Continuing inductively, we obtain
	\ref{item: diagOrder}.
\end{proof}

\begin{lemma}\label{lemma: diagDistinctLoopRestrictions}
	Let $N$ satisfy the properties of
	~\cref{lemma: diagDistinctRestrictions}.
	Then, $\PlGH(N)$ is $\#$P-hard, unless $N$
	also satisfies the condition that
	there exists some $\alpha \in \mathbb{R}_{> 0}$, such that:
	\begin{specialenumerate} \setcounter{specialenumeratei}{4}
		\item \label{item: LNalsoBad}
		$\mathcal{L}_{N}(\alpha)$ satisfies all of
		\rm{\ref{item: antiDiagMin} -- \ref{item: diagOrder}},
		\item \label{item: LNRowq}
		$\mathcal{L}_{N}(\alpha)_{iq} =  
		\mathcal{L}_{N}(\alpha)_{qq}$
		for all $i \in [q]$,
		\item \label{item: LNRow1}
		$\mathcal{L}_{N}(\alpha)_{i1} =  
		\mathcal{L}_{N}(\alpha)_{ii}$ for all $i \in [q]$.
	\end{specialenumerate}
\end{lemma}
\begin{proof}
	Consider the gadget $\mathcal{L}_{N}(\theta)$
	from \cref{definition: mathcalL}.
	Let
	$$\delta(\theta) :=\frac{\mathcal{L}_{N}(\theta)_{1q}}
	{\mathcal{L}_{N}(\theta)_{qq}} = 
	\frac{N_{1q}(N_{11}N_{qq})^{\theta}}
	{N_{qq}(N_{qq}N_{qq})^{\theta}} =
	\frac{N_{1q}}{N_{qq}}\left(
	\frac{N_{11}}{N_{qq}}\right)^{\theta}.$$
	
	Since $N$ satisfies \ref{item: diagOrder}
	and we have $q >1$,
	for large  $\theta$ we have
	$\delta(\theta) > 1$.
	Moreover, since $N$ satisfies \ref{item: antiDiagMin},
	and $N$ is intra-row-distinct,
	we have $N_{1q} < N_{qq}$.
	So, we see that
	$\delta(0) < 1$.
	Therefore, the Intermediate Value Theorem implies
	that there must be some $\alpha > 0$, such that
	$\delta(\alpha) = 1$, i.e.,
	$\mathcal{L}_{N}(\alpha)_{1q}
	= \mathcal{L}_{N}(\alpha)_{qq}$.
	
	Let us assume the following claim is true for now,
	we shall prove it later:
	\begin{restatable}{claim}{LAlphaRestrictions}
		\label{claim: LAlphaRestrictions}
		Unless $\PlGH(N)$ is $\#$P-hard,
		$\mathcal{L}_{N}(\alpha)$ satisfies all of
		\rm{\ref{item: antiDiagMin} -- \ref{item: diagOrder}}.
	\end{restatable}
	
	Now we assume \cref{claim: LAlphaRestrictions} is true.
	Note that $N$ satisfies \ref{item: LNalsoBad} which is assumed in the lemma.
	Consequently, since $\mathcal{L}_{N}(\alpha)$
	satisfies \ref{item: rowCi} 
	for the row $q$, we see that for our choice of $\alpha$,
	$$\mathcal{L}_{N}(\alpha)_{iq} \cdot
	\mathcal{L}_{N}(\alpha)_{i^{\rm{R}} q} = 
	\mathcal{L}_{N}(\alpha)_{1q} \cdot 
	\mathcal{L}_{N}(\alpha)_{qq}
	= \left(\mathcal{L}_{N}(\alpha)_{qq} \right)^{2},$$
	for all $i \in [q]$.
	Since $\mathcal{L}_{N}(\alpha)$ satisfies \ref{item: antiDiagMin},
	we see that $\min_{i, j}\mathcal{L}_{N}(\alpha)_{ij}
	= \mathcal{L}_{N}(\alpha)_{1q}$, 
	and by the choice of $\alpha$ we also have $\mathcal{L}_{N}(\alpha)_{1q} =
	\mathcal{L}_{N}(\alpha)_{qq}$.
	Then, $\mathcal{L}_{N}(\alpha)_{iq} \cdot
	\mathcal{L}_{N}(\alpha)_{i^{\rm{R}} q}$ is a product of two positive factors
	both $\ge \mathcal{L}_{N}(\alpha)_{qq}$, and yet  the product equals its square.
	This implies that
	$$\mathcal{L}_{N}(\alpha)_{iq} = 
	\mathcal{L}_{N}(\alpha)_{i^{\rm{R}} q} = 
	\mathcal{L}_{N}(\alpha)_{qq}, \qquad
	\text{for all } i \in [q].$$
	So, we see that $N$ satisfies \ref{item: LNRowq},
	i.e., for the chosen $\alpha$,
	$\mathcal{L}_{N}(\alpha)_{iq} =  
	\mathcal{L}_{N}(\alpha)_{qq}$
	for all $i \in [q]$.

	Moreover, since $\mathcal{L}_{N}(\alpha)$
	satisfies property \ref{item: rowCi} for row $i$,
	we also see that
	$$\mathcal{L}_{N}(\alpha)_{iq} 
	\cdot \mathcal{L}_{N}(\alpha)_{i1} = 
	\mathcal{L}_{N}(\alpha)_{ii} 
	\cdot \mathcal{L}_{N}(\alpha)_{i i^{\rm{R}}},$$
	for all $i \in [q]$.
	Since $\mathcal{L}_{N}(\alpha)$ satisfies
	\ref{item: antiDiagMin} and \ref{item: LNRowq},
	we also know that
	$$\mathcal{L}_{N}(\alpha)_{i i^{\rm{R}}} 
	= \mathcal{L}_{N}(\alpha)_{1q} 
	= \mathcal{L}_{N}(\alpha)_{qq}
	= \mathcal{L}_{N}(\alpha)_{iq},$$
	for all $i \in [q]$.
	So, by canceling this positive factor in the equation above,
	we find that $N$ satisfies \ref{item: LNRow1} as well, 
	i.e., for the chosen $\alpha$,
	$\mathcal{L}_{N}(\alpha)_{i1} =  
	\mathcal{L}_{N}(\alpha)_{ii}$ for all $i \in [q]$.
\end{proof}

We will now prove
\cref{claim: LAlphaRestrictions}.

\LAlphaRestrictions*
\begin{proof}
	Since $\alpha > 0$, we see that
	$$N_{ii} > N_{jj} \iff 
	\mathcal{L}_{N}(\alpha)_{ii} = (N_{ii})^{2\alpha + 1}
	> (N_{jj})^{2\alpha + 1} =  \mathcal{L}_{N}(\alpha)_{jj}.$$
	Since $N$ satisfies \ref{item: diagOrder},
	it trivially follows that $\mathcal{L}_{N}(\alpha)$
	satisfies \ref{item: diagOrder} as well.
	
	Let us now assume that $\mathcal{L}_{N}(\alpha)$
	fails to satisfy one of the following equations:
	\begin{equation}\label{equation: mathcalLAlphaEqns}
		\begin{aligned}
			&\sum_{i \in [q]}\left(\mathcal{L}_{N}(\alpha)_{1q}
			- \mathcal{L}_{N}(\alpha)_{i i^{\rm{R}}}\right)^{2} = 0,\\
			&\sum_{i \in [q]}\left(\mathcal{L}_{N}(\alpha)_{11}
			\mathcal{L}_{N}(\alpha)_{qq} - 
			\mathcal{L}_{N}(\alpha)_{ii}
			\mathcal{L}_{N}(\alpha)_{i^{\rm{R}} i^{\rm{R}}}\right)^{2} = 0,\\
			&\sum_{i \in [q]}\sum_{j \in [q]} \left(
			\mathcal{L}_{N}(\alpha)_{i1}
			\mathcal{L}_{N}(\alpha)_{iq}
			- \mathcal{L}_{N}(\alpha)_{ij}
			\mathcal{L}_{N}(\alpha)_{i j^{\rm{R}}}\right)^{2} = 0.
		\end{aligned}
	\end{equation}
	
	Note that if $\mathcal{L}_{N}(\alpha)$ fails to satisfy
	any of these equations, then
	$\mathcal{L}_{N}(\alpha)$ must fail to satisfy at least one of
	\ref{item: antiDiagMin}, \ref{item: diagD}, or
	\ref{item: rowCi}.    
	Let $\zeta$ be the $\Sym{q}{}(\mathbb{R})$-polynomial
	corresponding to (any one) offending equation, s.t.,
	$\zeta(\mathcal{L}_{N}(\alpha)) \neq 0$.
	We note that since $N$ is
	strongly diagonal distinct, and intra-row-distinct,
	$\rho(\mathcal{L}_{N}(0)) = \rho(N) \neq 0$
	for $\rho \in \{\psi_{\rm{diag}}, \phi_{\rm{row}} \}$ from
	\cref{defn:phi-psi-function-for-diagonal-distinct,defn:phi-function-for-intra-row-distinct}.
	Moreover, since $N \in \Sym{q}{pd}(\mathbb{R}_{> 0})$,
	and $N$ is diagonal dominant,
	we know that there exists some open neighborhood
	$\mathcal{O}$ of $0$, s.t., for all $\theta \in \mathcal{O}$,
	$\mathcal{L}_{N}(\alpha) \in \Sym{q}{pd}(\mathbb{R}_{> 0})$
	is also diagonal dominant.
	
	But then, \cref{lemma: analyticGadgets} immediately
	implies that we can find some
	$\theta^{*} \in \mathcal{O}$, s.t.,
	$\mathcal{L}_{N}(\theta^{*}) \in 
	\Sym{q}{pd}(\mathbb{R}_{> 0})$ is
	diagonal dominant, and intra-row-distinct, but also
	satisfies the condition that $\zeta(\mathcal{L}_{N}(\theta^{*}))
	\neq 0$.
	But then, from \cref{lemma: diagDistinctRestrictions},
	we have that $\PlGH(\mathcal{L}_{N}(\theta^{*}))
	\leq \PlGH(N)$ is $\#$P-hard.
	
	We will now assume that $\mathcal{L}_{N}(\alpha)$
	satisfies \cref{equation: mathcalLAlphaEqns}.
	Then, the only way that $\mathcal{L}_{N}(\alpha)$
	fails to satisfy all of
	\ref{item: antiDiagMin} -- \ref{item: diagOrder} is if
	$\min_{i, j \in [q]}\mathcal{L}_{N}(\alpha)_{ij}
	\neq \mathcal{L}_{N}(\alpha)_{1q}$.
	Since we cannot express this property as 
	$[F(A) \neq 0]$ for some $\Sym{q}{}(\mathbb{R})$-polynomial $F$,
	we have to deal with this case differently.
	We note that since we have no guarantee that
	$\mathcal{L}_{N}(\alpha)$ is strongly diagonal-distinct, 
	or even that
	$\mathcal{L}_{N}(\alpha) \in \Sym{q}{F}(\mathbb{R}_{> 0})$,
	we cannot hope to use \cref{theorem: inductionStep}
	directly to prove the hardness of
	$\PlGH(\mathcal{L}_{N}(\alpha))$.
	
	Instead, we define $\delta: \Sym{q}{}(\mathbb{R}) \rightarrow
	\mathbb{R}$, s.t.,
	$$\delta(A) = \min\left\{\min_{i \in [q]}A_{ii},
	\min_{i \in [q]}A_{i1}\right\} - 
	\min_{i, j \in [q]}A_{ij}.$$
	We note that from our choice of $\alpha$,
	$\mathcal{L}_{N}(\alpha)_{qq} =
	\mathcal{L}_{N}(\alpha)_{1q}$.
	So, we find that from our assumption about
	$\mathcal{L}_{N}(\alpha)$,
	and the fact that $\mathcal{L}_{N}(\alpha)$
	satisfies \ref{item: diagOrder},
	$$\min_{i, j}\mathcal{L}_{N}(\alpha)_{ij} <
	\mathcal{L}_{N}(\alpha)_{1q} 
	= \mathcal{L}_{N}(\alpha)_{qq} \leq 
	\mathcal{L}_{N}(\alpha)_{ii}, \quad
	\text{for all } i \in [q].$$
	Similarly, we note that
	since $N$ satisfies \ref{item: antiDiagMin} and  \ref{item: diagOrder}, 
	and is intra-row-distinct,
	$N_{1q} < N_{i1}$ for all $i \in [q-1]$.
	It follows that,
	$$\mathcal{L}_{N}(\alpha)_{1q}
	= N_{1q}\left(N_{11}N_{qq}\right)^{\alpha}
	< N_{1i}\left(N_{11}N_{ii}\right)^{\alpha}
	= \mathcal{L}_{N}(\alpha)_{1i}, \quad
	\text{for all } i \in [q-1].$$
	So,
	$$\min_{i, j}\mathcal{L}_{N}(\alpha)_{ij}
	< \mathcal{L}_{N}(\alpha)_{1q}
	\leq \mathcal{L}_{N}(\alpha)_{i1}, \quad 
	\text{for all } i \in [q].$$
	Therefore, we see that $\delta(\mathcal{L}_{N}(\alpha)) > 0$.
	
	We will now consider $\mathcal{L}_{N}(\theta)$
	as a function of $\theta$.
	Since the entries of $\mathcal{L}_{N}(\theta)$
	are continuous as a function of $\theta$,
	we see that there must exist an open neighborhood
	$\mathcal{O}'$ around $\alpha$, such that
	$\delta(\mathcal{L}_{N}(\theta)) > 0$ for all
	$\theta \in \mathcal{O}'$.
	Since $N$ is strongly diagonal distinct,
	and intra-row-distinct, we also see that
	$\rho(\mathcal{L}_{N}(0)) = \rho(N) \neq 0$
	for all $\rho \in \{\psi_{\rm{diag}}, \phi_{\rm{row}}\}$.
	Moreover, we note that since $N \in
	\Sym{q}{pd}(\mathbb{R}_{> 0})$,
	$\det(\mathcal{L}_{N}(0)) \neq 0$.
	Therefore, \cref{lemma: analyticGadgets} lets us
	find some $\theta^{*} \in \mathcal{O}'$ such that
	$N' = \mathcal{L}_{N}(\theta^{*}) \in
	\Sym{q}{F}(\mathbb{R}_{> 0})$ satisfies:
	$\rho(N') \neq 0$ for all
	$\rho \in \{\psi_{\rm{diag}}, \phi_{\rm{row}}\}$,
	$\delta(N') > 0$, and $\PlGH(N') \leq \PlGH(N)$.
	
	Let $m^{*} = \min_{i, j}(N')_{ij}$.
	Since $\delta(N') > 0$, we note that for row $1$,
	$(N')_{i1} > m^{*}$ for all $i \in [q]$.
	Moreover, since $(N')_{ii} > m^{*}$ for all $i \in [q]$,
	there must exist some $i^{\dagger} \neq j^{\dagger}$,
	s.t., $(N')_{i^{\dagger}j^{\dagger}} = m^{*}$.
	So, \cref{theorem: inductionStep} implies that
	$\PlGH(N') \leq \PlGH(N)$ is $\#$P-hard, 
	unless $\mathcal{L}_{N}(\alpha)$ satisfies
	\ref{item: antiDiagMin} -- \ref{item: diagOrder}.
\end{proof}

\begin{lemma}\label{lemma: diagDistinctSecondLastRow}
	Let $N$ satisfy the properties of
	~\cref{lemma: diagDistinctLoopRestrictions}.
	Then, $\PlGH(N)$ is $\#$P-hard, unless $N$
	also satisfies the following conditions
	(for the same $\alpha \in \mathbb{R}_{> 0}$ from 
	\cref{lemma: diagDistinctLoopRestrictions}):
	\begin{specialenumerate} \setcounter{specialenumeratei}{7}
		\item \label{item: LNsecondLastMin}
		$|\{i \in [q]: \mathcal{L}_{N}(\alpha)_{(q-1)i} = 
		\mathcal{L}_{N}(\alpha)_{qq}\}| = \nicefrac{q}{2}$.
		\item \label{item: LNsecondLastMax}
		$|\{i \in [q]: \mathcal{L}_{N}(\alpha)_{(q-1)i} = 
		\mathcal{L}_{N}(\alpha)_{(q-1)(q-1)}\}| = \nicefrac{q}{2}$.
	\end{specialenumerate}
\end{lemma}
\begin{proof}    
	We note that since $N \in \Sym{q}{pd}(\mathbb{R}_{> 0})$,
	and $\mathcal{S}_{N}(1) = N$, there exists an open
	ball $\mathcal{O} \subset \mathbb{R}$ around
	$1$, such that
	$\mathcal{S}_{N}(\theta) \in \Sym{q}{pd}(\mathbb{R}_{> 0})$,
	for all $\theta \in \mathcal{O}$.
	We consider $\mathcal{K}_{N}: \mathcal{O} \times 
	\mathbb{R}^{2} \rightarrow
	\Sym{q}{}(\mathbb{R}_{> 0})$ defined as
	$$\mathcal{K}_{N}(\theta_{1}, \theta_{2}, \theta_{3})
	= \left(\mathcal{T}_{\mathcal{L}_{
			(\mathcal{S}_{N}(\theta_{1}))}
		(\theta_{2})}(\theta_{3})\right).$$
	Since $\mathcal{S}_{N}(\theta_{1}) \in 
	\Sym{q}{}(\mathbb{R}_{> 0})$
	for all $\theta_{1} \in \mathcal{O}$,
	we see that $\mathcal{K}_{N}
	(\theta_{1}, \theta_{2}, \theta_{3})_{ij}$
	is a well-defined analytic
	function for all $i, j \in [q]$.
	
	Let us now fix $\alpha$
	from \cref{lemma: diagDistinctLoopRestrictions}, and consider
	$\mathbf{S_{2}}(\mathcal{K}_{N}(1, \alpha, \theta))$ 
	as a function of $\theta$.
	We note that $\mathcal{S}_{N}(1) = N$, and 
	$\mathcal{K}_{N}(1, \alpha, \theta)
	= \mathcal{T}_{\mathcal{L}_{N}(\alpha)}(\theta)$
	for all $\theta \in \mathbb{R}$.
	From \ref{item: LNRowq}, we
	know that $\mathcal{L}_{N}(\alpha)_{qi} = 
	\mathcal{L}_{N}(\alpha)_{qq} = (N_{qq})^{2\alpha + 1}$
	for all $i \in [q]$.
	From \ref{item: LNRow1}, we also know that
	$\mathcal{L}_{N}(\alpha)_{1i} = 
	\mathcal{L}_{N}(\alpha)_{ii} = (N_{ii})^{2\alpha + 1}$
	for all $i \in [q]$.
	So,
	$$\mathbf{S_{2}}(\mathcal{K}_{N}(1, \alpha, \theta))_{qq} = 
	\sum_{i \in [q]}(\mathcal{L}_{N}(\alpha)_{qi})^{2\theta}
	= qN_{qq}^{2\theta(2\alpha+1)},$$
	$$\mathbf{S_{2}}(\mathcal{K}_{N}(1, \alpha, \theta))_{11} = 
	\sum_{i \in [q]}(\mathcal{L}_{N}(\alpha)_{1i})^{2\theta}
	= \sum_{i \in [q]}N_{ii}^{2\theta(2\alpha+1)}.$$
	We also know that $N$ satisfies \ref{item: diagOrder}.
	So,
	$$\frac{\mathbf{S_{2}}(\mathcal{K}_{N}(1, \alpha, \theta))_{qq}}
	{(N_{qq})^{2\theta(2\alpha+1)}}
	= q, 
	\text{ and}$$
	$$\lim\limits_{\theta \rightarrow -\infty}
	\frac{\mathbf{S_{2}}(\mathcal{K}_{N}(1, \alpha, \theta))_{11}}
	{(N_{qq})^{2\theta(2\alpha+1)}}
	=  1 +  \lim\limits_{\theta \rightarrow -\infty}
	\sum\limits_{i < q} (N_{ii}/N_{qq})^{2\theta(2\alpha+1)} = 1.$$
	Let us assume the following claim is true for now,
	we shall prove it later:
	
	\begin{restatable}{claim}{AlphaGSatisfiesDiagD}
		\label{claim: AlphaGSatisfiesDiagD}
		Unless $\PlGH(N)$ is $\#$P-hard,
		$\mathbf{S_{2}}(\mathcal{K}_{N}(1, \alpha, \theta))$ 
		satisfies \ref{item: diagD} for all
		$\theta \in \mathbb{R}$.
	\end{restatable}
	
	Now, \cref{claim: AlphaGSatisfiesDiagD} implies that (recall
	$1^{\rm R} = q$ and $2^{\rm R}=q-1$)
	\begin{equation}\label{eqn: mathcalG1q}
		\lim\limits_{\theta \rightarrow -\infty}
		\frac{\mathbf{S_{2}}
			(\mathcal{K}_{N}(1, \alpha, \theta))_{11} \cdot
			\mathbf{S_{2}}(\mathcal{K}_{N}(1, \alpha, \theta))_{qq}}
		{(N_{qq})^{4\theta(2\alpha+1)}}
		= q = 
		\lim\limits_{\theta \rightarrow -\infty}
		\frac{\mathbf{S_{2}}
			(\mathcal{K}_{N}(1, \alpha, \theta))_{22} \cdot
			\mathbf{S_{2}}
			(\mathcal{K}_{N}(1, \alpha, \theta))_{2^{\rm R}2^{\rm R}}}
		{(N_{qq})^{4\theta(2\alpha+1)}}.
	\end{equation}
	
	Let us now focus on row $2$ of $\mathcal{L}_{N}(\alpha)$.
	Since $N$ satisfies \ref{item: antiDiagMin},
	and is intra-row-distinct,
	we know that $N_{22^{\rm R}} < N_{2i}$ for all
	$i \neq 2^{\rm R} \in [q]$.
	Moreover, since $N$ satisfies \ref{item: diagOrder},
	we also know that $N_{ii} > N_{2^{\rm R}2^{\rm R}}$
	for all $i \in [q-2]$.
	Therefore,
	$$\mathcal{L}_{N}(\alpha)_{22^{\rm R}} = 
	N_{22^{\rm R}}(N_{22}N_{2^{\rm R}2^{\rm R}})^{\alpha} <
	N_{2i}(N_{22}N_{ii})^{\alpha}
	= \mathcal{L}_{N}(\alpha)_{2i}, \quad
	\text{for all } i \in [q-2].$$
	Since $\mathcal{L}_{N}(\alpha)$ satisfies
	\ref{item: antiDiagMin}, and \ref{item: LNRowq}, 
	we also know that
	$\mathcal{L}_{N}(\alpha)_{22^{\rm R}} = \mathcal{L}_{N}(\alpha)_{1q}
	= \mathcal{L}_{N}(\alpha)_{2q} = \mathcal{L}_{N}(\alpha)_{qq} = (N_{qq})^{2\alpha+1}$.
	So, we see that
	$$\lim\limits_{\theta \rightarrow -\infty}
	\frac{\mathbf{S_{2}}(\mathcal{K}_{N}(1, \alpha, \theta))_{22}}
	{(N_{qq})^{2\theta(2\alpha+1)}} = 2.$$
	
	So, \cref{eqn: mathcalG1q} implies that
	$$\lim\limits_{\theta \rightarrow -\infty}
	\frac{\mathbf{S_{2}}
		(\mathcal{K}_{N}(1, \alpha, \theta))_{2^{\rm R}2^{\rm R}}}{(N_{qq})^{2\theta(2\alpha+1)}} = \frac{q}{2}.$$
	But this implies that
	$|\{i \in [q]: \mathcal{L}_{N}(\alpha)_{i(q-1)}
	= \mathcal{L}_{N}(\alpha)_{qq}\}| = \nicefrac{q}{2}$,
	which proves \ref{item: LNsecondLastMin}.
	A consequence of \ref{item: LNsecondLastMin}
	is also that if $q$ is odd, then,
	$\PlGH(N)$ is $\#$P-hard.
	So, $q$ must be even.
	This implies that for all $i \in [q]$,
	$i \neq q + 1 - i = i^{\rm{R}}$.
	
	From \ref{item: LNalsoBad}, we know that
	$\mathcal{L}_{N}(\alpha)$ satisfies
	\ref{item: rowCi} for row $(q-1)$.
	So, for all $i \in [q]$,
	$$\mathcal{L}_{N}(\alpha)_{i(q-1)}
	\mathcal{L}_{N}(\alpha)_{i^{\rm{R}} (q-1)} = 
	\mathcal{L}_{N}(\alpha)_{(q-1)(q-1)}
	\mathcal{L}_{N}(\alpha)_{2(q-1)} = 
	\mathcal{L}_{N}(\alpha)_{(q-1)(q-1)}
	\mathcal{L}_{N}(\alpha)_{qq}.$$
	Since $\mathcal{L}_{N}(\alpha)$ satisfies
	\ref{item: diagOrder}, we know that
	$\mathcal{L}_{N}(\alpha)_{(q-1)(q-1)}
	> \mathcal{L}_{N}(\alpha)_{qq}$.
	So, we see that it is not possible that
	$$\mathcal{L}_{N}(\alpha)_{i (q-1)} = 
	\mathcal{L}_{N}(\alpha)_{i^{\rm{R}} (q-1)} =
	\mathcal{L}_{N}(\alpha)_{qq},$$
	as that would imply
	$\mathcal{L}_{N}(\alpha)_{i(q-1)}
	\mathcal{L}_{N}(\alpha)_{i^{\rm{R}} (q-1)}
	= (\mathcal{L}_{N}(\alpha)_{qq})^{2} <
	\mathcal{L}_{N}(\alpha)_{(q-1)(q-1)}
	\mathcal{L}_{N}(\alpha)_{qq}$.
	Thus,  each pair $\{\mathcal{L}_{N}(\alpha)_{i (q-1)}, 
	\mathcal{L}_{N}(\alpha)_{i^{\rm{R}} (q-1)}\}$ has most one
	element that is equal to $\mathcal{L}_{N}(\alpha)_{qq}$.
	Now, \ref{item: LNsecondLastMin} says the total is $\nicefrac{q}{2}$;
	therefore exactly one from  each pair $\{\mathcal{L}_{N}(\alpha)_{i (q-1)}, 
	\mathcal{L}_{N}(\alpha)_{i^{\rm{R}} (q-1)}\}$ 
	is equal to $\mathcal{L}_{N}(\alpha)_{qq}$.
	%
	and, thus the other element of the pair
	must equal
	$\mathcal{L}_{N}(\alpha)_{(q-1)(q-1)}$.
	This proves \ref{item: LNsecondLastMax}.
\end{proof}

Now, we complete the proof of
\cref{claim: AlphaGSatisfiesDiagD}.
\AlphaGSatisfiesDiagD*
\begin{proof}
	We define $\delta: \Sym{q}{}(\mathbb{R}) \rightarrow 
	\mathbb{R}$ as:
	$$\delta(A) = \min \left\{\left(\min_{i \in [q-1]}
	\left(A_{ii} - A_{(i+1)(i+1)}\right)\right),
	\left(\min_{i \in [q]}A_{ii}
	- \max_{i \in [q], j \neq i \in [q]}A_{ij}\right)
	\right\}.$$
	Since $N$ satisfies \ref{item: diagOrder},
	and is diagonal dominant, $\delta(N) > 0$.
	Since $\mathcal{K}_{N}(1, 0, 1) = N$,
	and the entries of $\mathcal{K}_{N}(\theta_{1},
	\theta_{2}, \theta_{3})$ are continuous
	in $(\theta_{1}, \theta_{2}, \theta_{3})$,
	we may assume that there exists an open
	cube $\mathcal{O}_{1} \times 
	\mathcal{O}_{2} \times \mathcal{O}_{3} \subset
	\mathcal{O} \times \mathbb{R}^{2}$,
	s.t., $\delta(\mathcal{K}_{N}(\theta_{1},
	\theta_{2}, \theta_{3})) > 0$, and
	$\mathcal{K}_{N}(\theta_{1},
	\theta_{2}, \theta_{3}) \in \Sym{q}{pd}(\mathbb{R}_{> 0})$
	for all $(\theta_{1},
	\theta_{2}, \theta_{3}) \in 
	\mathcal{O}_{1} \times 
	\mathcal{O}_{2} \times \mathcal{O}_{3}$.
	We now define $\psi_{\rm{diag}}': A \mapsto
	\psi_{\rm{diag}}(A^{2})$, and
	$\phi_{\rm{row}}': A \mapsto \phi_{\rm{row}}(A^{2})$
	for $\psi_{\rm{diag}}, \phi_{\rm{row}}$ in
	\cref{defn:phi-psi-function-for-diagonal-distinct}, and
	\cref{defn:phi-function-for-intra-row-distinct}.
	Since $\rho(\mathcal{S}_{N}(\nicefrac{1}{2})) \neq 0$
	for $\rho \in \{\psi_{\rm{diag}}', \phi_{\rm{row}}'\}$,
	we may use \cref{lemma: analyticGadgets} to find
	$\theta' \in \mathcal{O}$, s.t.,
	$\rho(\mathcal{S}_{N}(\theta')) \neq 0$
	for both $\rho \in \{\psi_{\rm{diag}}', \phi_{\rm{row}}'\}$.
	In other words,
	$\rho(\mathcal{K}_{N}(\theta', 0, 1)) \neq 0$
	for both $\rho \in \{\psi_{\rm{diag}}', \phi_{\rm{row}}'\}$.
	
	We now assume that $\mathbf{S_{2}}(\mathcal{K}_{N}
	(1, \alpha, \theta))$
	does not satisfy \ref{item: diagD}.
	So,
	$\zeta(\mathbf{S_{2}}(\mathcal{K}_{N}(1, \alpha, \theta))) \neq 0$
	for $\zeta: A \mapsto \sum_{i}(A_{11}A_{qq} - 
	A_{ii}A_{i^{\rm{R}} i ^{\rm{R}}})^{2}$.
	We can now define $\zeta_{2}: A \mapsto \zeta(A^{2})$.
	Then, clearly,
	$\zeta_{2}(\mathcal{K}_{N}(1, \alpha, \theta)) \neq 0$.
	
	We can now use \cref{lemma: analyticGadgets} to
	find some $(\theta_{1}^{*}, \theta_{2}^{*},
	\theta_{3}^{*}) \in
	\mathcal{O}_{1} \times \mathcal{O}_{2} \times
	\mathcal{O}_{3}$, such that
	$N' = \mathcal{K}_{N}(\theta_{1}^{*}, \theta_{2}^{*},
	\theta_{3}^{*}) \in \Sym{q}{pd}(\mathbb{R}_{> 0})$
	satisfies:
	$\rho(N') \neq 0$ for all $\rho \in 
	\{\psi_{\rm{diag}}', \phi_{\rm{row}}', \zeta_{2} \}$, and
	$\delta(N') > 0$.
	
	Therefore, $\mathbf{S_{2}}(N') \in \Sym{q}{pd}(\mathbb{R}_{> 0})$
	satisfies:
	$\rho(\mathbf{S_{2}}(N')) \neq 0$ for all
	$\rho \in \{\psi_{\rm{diag}}, \phi_{\rm{row}}, \zeta \}$.
	We note that since $N' \in \Sym{q}{pd}(\mathbb{R}_{> 0})$,
	$\mathcal{S}_{\mathbf{S_{2}}(N')}(\nicefrac{1}{2}) = N'$.
	Since $\delta(N') > 0$,
	and $N' \in \Sym{q}{pd}(\mathbb{R}_{> 0})$
	we know that there exists some open neighborhood
	$\mathcal{O}'$ of $\nicefrac{1}{2}$, such that
	$\mathcal{S}_{\mathbf{S_{2}}(N')}(\theta) \in 
	\Sym{q}{pd}(\mathbb{R}_{> 0})$, and
	$\delta(\mathcal{S}_{\mathbf{S_{2}}(N')}(\theta)) > 0$
	for all $\theta \in \mathcal{O}'$.
	We use \cref{lemma: analyticGadgets} again,
	to find some $\theta^{*} \in \mathcal{O}'$
	such that $N'' = \mathcal{S}_{\mathbf{S_{2}}(N')}(\theta^{*})
	\in \Sym{q}{pd}(\mathbb{R}_{> 0})$
	satisfies the conditions:
	$\rho(N'') \neq 0$ for
	$\rho \in \{\psi_{\rm{diag}}, \phi_{\rm{row}}, \zeta \}$,
	and $\delta(N'') > 0$.
	From our choice of $\delta$,
	this implies that
	$N''$ is diagonal dominant, and that
	$(N'')_{11} > (N'')_{22} > \dots > (N'')_{qq}$.
	Now, let us assume that there exists some
	permutation $\sigma$ of $[q]$, such that
	$(N'')^{\sigma}$ satisfies \ref{item: diagOrder}.
	But then, $(N'')_{11} > \dots > (N'')_{qq}$ implies that
	$\sigma$ must be the trivial identity permutation.
	But we had observed that $\zeta(N'') \neq 0$
	while defining $N''$.
	So, from \cref{lemma: diagDistinctRestrictions},
	we see that $\PlGH(N'') \leq \PlGH(N') \leq
	\PlGH(N)$ is $\#$P-hard.
\end{proof}

\begin{lemma}\label{lemma: diagDistinctAleph}
	Let $N$ satisfy the properties of
	~\cref{lemma: diagDistinctSecondLastRow}.
	Then, $\PlGH(N)$ is $\#$P-hard, unless
	$N$ also satisfies:
	\begin{specialenumerate} \setcounter{specialenumeratei}{9}
		\item \label{item: LNAleph}
		$\aleph(N) = 0$, where
		$\aleph: \Sym{q}{}(\mathbb{R}_{> 0})
		\rightarrow \mathbb{R}$ is defined as:
	\end{specialenumerate}
	$$\aleph(A) = \Big(\ln(A_{1q}) - \ln(A_{qq}) \Big)
	\Big(\ln(A_{(q-1)(q-1)}) - \ln(A_{qq})\Big) - 
	\Big(\ln(A_{11}) - \ln(A_{qq})\Big)
	\Big(\ln(A_{(q-1)q}) - \ln(A_{qq}) \Big).$$
\end{lemma}
\begin{proof}
	First, we note that $N \in \Sym{q}{pd}(\mathbb{R}_{> 0})$
	satisfies \ref{item: diagOrder}, therefore
	$\ln(N_{ii}) - \ln(N_{qq}) > 0$ for all $i \in [q-1]$.
	Now, since $N$ satisfies \ref{item: LNRowq}, i.e., 
	$\mathcal{L}_{N}(\alpha)_{iq} =
	\mathcal{L}_{N}(\alpha)_{qq}$, (and is positive so the fractions below are well defined)
	we see that
	$$1 = \frac{\mathcal{L}_{N}(\alpha)_{iq}}
	{\mathcal{L}_{N}(\alpha)_{qq}} = 
	\frac{N_{iq}(N_{ii}N_{qq})^{\alpha}}
	{N_{qq}(N_{qq}N_{qq})^{\alpha}} = 
	\frac{N_{iq}}{N_{qq}} \left(\frac{N_{ii}}{N_{qq}}\right)^{\alpha},
	\quad \text{for all } i \in [q].$$
	So,
	$$\alpha = \frac{\left(\ln(N_{qq}) - \ln(N_{iq})\right)}
	{\left(\ln(N_{ii}) - \ln(N_{qq}) \right)}, \quad
	\text{for all } i \in [q-1].$$
	This implies that
	$$\frac{\left(\ln(N_{1q}) - \ln(N_{qq})\right)}
	{\left(\ln(N_{11}) - \ln(N_{qq}) \right)} = -\alpha = 
	\frac{\left(\ln(N_{(q-1)q}) - \ln(N_{qq})\right)}
	{\left(\ln(N_{(q-1)(q-1)}) - \ln(N_{qq}) \right)}.$$
	Therefore, $\aleph(N) = 0$.
\end{proof}

We have  completed the proof of \cref{theorem: diagDistinctAllProperties}.
Now, to finish the proof of \cref{theorem: diagonalDistinctHardness},
we prove
\cref{lemma: diagDistinctHardness}.

\diagDistinctHardness*
\begin{proof}
	We note that since $N \in \Sym{q}{pd}(\mathbb{R}_{> 0})$,
	and $\mathcal{S}_{N}(1) = N$, there exists an open
	ball $\mathcal{O} \subset \mathbb{R}$ around
	$1$, such that
	$\mathcal{S}_{N}(\theta) \in \Sym{q}{pd}(\mathbb{R}_{> 0})$
	for all $\theta \in \mathcal{O}$.
	We consider $\mathcal{K}_{N}: \mathcal{O} \times 
	\mathbb{R}^{2} \rightarrow
	\Sym{q}{}(\mathbb{R}_{> 0})$ defined as
	$$\mathcal{K}_{N}(\theta_{1}, \theta_{2}, \theta_{3})
	= \mathcal{T}_{\mathcal{L}_{
			(\mathcal{S}_{N}(\theta_{1}))}
		(\theta_{2})}(\theta_{3}).$$
	Since $\mathcal{S}_{N}(\theta_{1}) \in 
	\Sym{q}{}(\mathbb{R}_{> 0})$
	for all $\theta_{1} \in \mathcal{O}$,
	we see that $\mathcal{K}_{N}(\theta_{1}, \theta_{2}, \theta_{3})_{ij}$
	is a well-defined analytic
	function for all $i, j \in [q]$.
	Moreover, we note that
	$\mathcal{K}_{N}(1, 0, 1) = N$.
	We define $\psi_{\rm{diag}}': A \mapsto
	\psi_{\rm{diag}}(A^{2})$, and
	$\phi_{\rm{row}}': A \mapsto \phi_{\rm{row}}(A^{2})$
	for $\psi_{\rm{diag}}, \phi_{\rm{row}}$ in
	\cref{defn:phi-psi-function-for-diagonal-distinct}, and
	\cref{defn:phi-function-for-intra-row-distinct}.
	Since $\rho(\mathcal{S}_{N}(\nicefrac{1}{2})) \neq 0$
	for $\rho \in \{\psi_{\rm{diag}}', \phi_{\rm{row}}'\}$,
	we may use \cref{lemma: analyticGadgets} to find
	$\theta' \in \mathcal{O}$, s.t.,
	$\rho(\mathcal{S}_{N}(\theta')) \neq 0$
	for both $\rho \in \{\psi_{\rm{diag}}', \phi_{\rm{row}}'\}$.
	This also means that there exists some
	$\theta' \in \mathcal{O}$, s.t.,
	$\rho(\mathcal{K}_{N}(\theta', 0, 1)) \neq 0$
	for both $\rho \in \{\psi_{\rm{diag}}', \phi_{\rm{row}}'\}$.
	
	Setting this aside, we recall that
	as a consequence of \cref{lemma: diagDistinctAleph},
	$\PlGH(N)$ is $\#$P-hard, unless
	$N$ satisfies \ref{item: antiDiagMin} -- \ref{item: LNAleph}.
	Let $\alpha \in \mathbb{R}_{> 0}$ be the real that 
	satisfies \ref{item: LNalsoBad} -- \ref{item: LNsecondLastMax}
	for the matrix $N$.
	We will now assume the following claim is true for now, we shall
	prove it later.
	\begin{restatable}{claim}{AlphaGThetaNotConstant}
		\label{claim: AlphaGThetaNotConstant}
		There exists some
		$\theta^{*} \in \mathbb{R}$, such that
		$\aleph\Big(\mathbf{S_{2}}\big(\mathcal{K}_{N}(1, \alpha, \theta^{*})
		\big)\Big)  \neq 0$.
	\end{restatable}
	
	We now define $\delta_{1}: \Sym{q}{}(\mathbb{R}_{> 0}) \rightarrow
	\mathbb{R}$ as:
	$$\delta_{1}(A) = \min_{\substack{i \in [q], i_{1}, j_{1} \in [q]\\
			i_{2} \neq j_{2} \in [q]}}\left(\frac{A_{ii}^{2}}
	{A_{i_{1}j_{1}}A_{i_{2}j_{2}}} \right),$$
	and $\delta_{2}: \Sym{q}{}(\mathbb{R}_{> 0}) \rightarrow
	\mathbb{R}$ as:
	$$\delta_{2}(A) = \min_{i \in [q-1], j \in [q]}\left(
	\frac{A_{ii}}{A_{(i+1)j}}\right)^{2}.$$
	
	Since $N$ is I-close,
	we see from \cref{definition: IClose} that
	$N_{ii}^{2} > (\nicefrac{2}{3})^{2} = 
	\nicefrac{4}{9}$ for all $i \in [q]$.
	Moreover,
	$N_{i_{1}j_{1}} < (\nicefrac{4}{3})$, and
	since $i_{2} \neq j_{2}$,
	$N_{i_{2}j_{2}} < (\nicefrac{1}{3})$. So,
	$N_{i_{1}j_{1}}N_{i_{2}j_{2}} < 
	(\nicefrac{4}{3})(\nicefrac{1}{3}) = \nicefrac{4}{9}$
	for all $i_{1}, j_{1} \in [q], i_{2} \neq j_{2} \in [q]$.
	So, $\delta_{1}(N) > 1$.
	Since $N$ satisfies \ref{item: diagOrder}
	and is diagonal dominant,
	we also see that $\delta_{2}(N) > 1$.
	
	Since $\delta_{1}(A), \delta_{2}(A)$ are continuous
	in the entries of $A$, we know that there exists an
	open neighborhood $\mathcal{O}_{1} \times \mathcal{O}_{2}
	\subset \mathcal{O} \times \mathbb{R}$ around
	$(1, 0)$ such that
	for both
	$\rho \in \{\delta_{1}, \delta_{2}\}$,
	the number
	$\rho(\mathcal{L}_{\mathcal{S}_{N}(\theta_{1})}(\theta_{2}))$ is 
	sufficiently close to the quantity $\rho(N)$
	which is greater than 1.
	Specifically, we may let $\mu = \min \left( 
	\left(\delta_{1}(N)\right)^{\nicefrac{1}{2}}, 
	\left(\delta_{2}(N)\right)^{\nicefrac{1}{2}}\right)$.
	Since $\delta_{1}(N) > 1$, and $\delta_{2}(N) > 1$,
	we see that $\mu > 1$, and we also see that
	$\mu < \rho(N)$ for $\rho \in \{\delta_{1}, \delta_{2} \}$.
	So, we may assume
	$\rho(\mathcal{L}_{\mathcal{S}_{N}(\theta_{1})}(\theta_{2}))
	> \mu$,  for all
	$(\theta_{1}, \theta_{2}) \in 
	\mathcal{O}_{1} \times \mathcal{O}_{2}$.
	Moreover, we may assume that
	$\mathcal{L}_{\mathcal{S}_{N}(\theta_{1})}(\theta_{2}) \in
	\Sym{q}{pd}(\mathbb{R}_{> 0})$ for all 
	$(\theta_{1}, \theta_{2}) \in 
	\mathcal{O}_{1} \times \mathcal{O}_{2}$.
	
	We will now let
	$U = \{\mathcal{L}_{\mathcal{S}_{N}
		(\theta_{1})}(\theta_{2}):
	(\theta_{1}, \theta_{2}) \in 
	\mathcal{O}_{1} \times \mathcal{O}_{2}\}$, and
	consider $\mathbf{S_{2}}(\mathcal{T}_{A}(\theta))$
	for $A \in U$.
	From our choice of $U$, we
	see that
	$\delta_{1}(A) > \mu$, and $\delta_{2}(A) > \mu$
	for all $A \in U$.
	Now, we see that for all $A \in U$, and for all
	$i \in [q-1]$,
	$$\frac{\mathbf{S_{2}}(\mathcal{T}_{A}(\theta))_{ii}}
	{\mathbf{S_{2}}(\mathcal{T}_{A}(\theta))_{(i+1)(i+1)}}
	= \frac{\sum_{j \in [q]}(A_{ij})^{2\theta}}
	{\sum_{j \in [q]}(A_{(i+1)j})^{2\theta}}
	= \frac{1 + \sum_{j \neq i \in [q]}
		(\nicefrac{A_{ij}}{A_{ii}})^{2\theta}}
	{\sum_{j \in [q]}
		(\nicefrac{A_{(i+1)j}}{A_{ii}})^{2\theta}}
	\ge
	\frac{1}
	{\sum_{j \in [q]}
		(\nicefrac{A_{(i+1)j}}{A_{ii}})^{2\theta}}
	> \frac{\mu^{\theta}}{q}.$$
	The last inequality here follows from the fact that
	$\left(\nicefrac{A_{ii}}{A_{(i+1)j}}\right)^{2} \geq 
	\delta_{2}(A) > \mu$ for all $i \in [q-1]$, $j \in [q]$.
	
	We also see that
	for all $i, j_{1} \neq j_{2} \in [q]$,
	$$\frac{\mathbf{S_{2}}(\mathcal{T}_{A}(\theta))_{ii}}
	{\mathbf{S_{2}}(\mathcal{T}_{A}(\theta))_{j_{1}j_{2}}}
	= \frac{\sum_{k \in [q]}(A_{ik})^{2\theta}}
	{\sum_{k \in [q]}(A_{j_{1}k}A_{j_{2}k})^{\theta}}
	= \frac{1 + \sum_{k \neq i \in [q]}
		(\nicefrac{A_{ik}}{A_{ii}})^{2\theta}}
	{\sum_{k \in [q]}
		(\nicefrac{A_{j_{1}k}A_{j_{2}k}}
		{A_{ii}^{2}})^{\theta}}
	\ge
	\frac{1}
	{\sum_{k \in [q]}
		(\nicefrac{A_{j_{1}k}A_{j_{2}k}}
		{A_{ii}^{2}})^{\theta}}
	> \frac{\mu^{\theta}}{q}.$$
	The last inequality here follows from the fact that
	$\left(\nicefrac{A_{ii}^{2}}
	{A_{j_{1}k}A_{j_{2}k}}\right) 
	\geq \delta_{1}(A) > \mu$,
	for all $i \in [q]$, $j_{1} \neq j_{2} \in [q]$.
	
	Similarly, we can also see that for all $i \in [q]$,
	$$\frac{\mathbf{S_{2}}(\mathcal{T}_{A}(\theta))_{ii}}
	{\sum_{j \neq i}\mathbf{S_{2}}(\mathcal{T}_{A}
		(\theta))_{ij}}
	= \frac{\sum_{k \in [q]}(A_{ik})^{2\theta}}
	{\sum_{j \neq i \in [q]}\sum_{k \in [q]}(A_{ik}A_{jk})^{\theta}}
	= \frac{1 + \sum_{k \neq i \in [q]}
		(\nicefrac{A_{ik}}{A_{ii}})^{2\theta}}
	{\sum_{j \neq i \in [q]}\sum_{k \in [q]}
		(\nicefrac{A_{ik}A_{jk}}
		{A_{ii}^{2}})^{\theta}}
	> \frac{\mu^{\theta}}{q(q-1)}.$$
	
	So, there exists some $\overline{\theta} \in \mathbb{R}$, such that
	for all $\theta > \overline{\theta}$, and $A \in U$,
	\begin{equation}\label{eqn: N'diagOrder}
		\frac{\mathbf{S_{2}}(\mathcal{T}_{A}(\theta))_{ii}}
		{\mathbf{S_{2}}(\mathcal{T}_{A}(\theta))_{(i+1)(i+1)}} > 1
		\text{ for all } i \in [q-1],
	\end{equation}
	\begin{equation}\label{eqn: N'diagDominance}
		\frac{\mathbf{S_{2}}(\mathcal{T}_{A}(\theta))_{ii}}
		{\mathbf{S_{2}}(\mathcal{T}_{A}(\theta))_{j_{1}j_{2}}} > 1
		\text{ for all } i \in [q], j_{1} \neq j_{2} \in [q],
	\end{equation}
	\begin{equation}\label{eqn: N'pd}
		\frac{\mathbf{S_{2}}(\mathcal{T}_{A}(\theta))_{ii}}
		{\sum_{j \neq i}\mathbf{S_{2}}(\mathcal{T}_{A}
			(\theta))_{ij}} > 1, \text{ for all } i \in [q].
	\end{equation}

	Since $\aleph$ is 
	a $\Sym{q}{}(\mathbb{R}_{> 0})$-analytic function,
	we can construct a new 
	$\Sym{q}{}(\mathbb{R}_{> 0})$-analytic function
	$\aleph_{2}: A \mapsto
	\aleph(A^{2})$.
	Now, \cref{claim: AlphaGThetaNotConstant} implies that
	$\aleph_{2}(\mathcal{K}_{N}(1, \alpha, \theta^{*}))
	\neq 0$ for some $\theta^{*} \in \mathbb{R}$.
	We also recall from the start of this proof (before stating \cref{claim: AlphaGThetaNotConstant}) that we found
	some $\theta' \in \mathcal{O}$, such that
	$\rho(\mathcal{K}_{N}(\theta', 0, 1)) \neq 0$
	for all $\rho \in \{\psi_{\rm{diag}}', \phi_{\rm{row}}'\}$.
	
	So, we can use \cref{lemma: analyticGadgets} to find
	some $(\theta_{1}^{*}, \theta_{2}^{*}, \theta_{3}^{*})
	\in \mathcal{O}_{1} \times \mathcal{O}_{2} \times
	(\overline{\theta}, \infty)$, such that
	$\mathcal{K}_{N}(\theta_{1}^{*}, \theta_{2}^{*}, \theta_{3}^{*})
	\in \Sym{q}{}(\mathbb{R}_{> 0})$ satisfies:
	$\rho(\mathcal{K}_{N}
	(\theta_{1}^{*}, \theta_{2}^{*}, \theta_{3}^{*})) 
	\neq 0$ for
	$\rho \in \{\aleph_{2}, \psi_{\rm{diag}}',
	\phi_{\rm{row}}'\}$, and
	$\PlGH(\mathcal{K}_{N}
	(\theta_{1}^{*}, \theta_{2}^{*}, \theta_{3}^{*})) 
	\leq \PlGH(N)$.
	
	We now let $N' =  \mathbf{S_{2}}
	(\mathcal{K}_{N}(\theta_{1}^{*}, \theta_{2}^{*}, \theta_{3}^{*}))$.
	We see that $\rho(N') \neq 0$ for
	$\rho \in \{\aleph, \psi_{\rm{diag}}, \phi_{\rm{row}}\}$.
	Moreover, from our choice of $\overline{\theta}$
	(\cref{eqn: N'pd}),
	we ensured that
	$(N')_{ii} > \sum_{j \neq i \in [q]}|(N')_{ij}|$
	for all $i \in [q]$.
	So, by the Gershgorin Circle Theorem,
	$N' \in \Sym{q}{pd}(\mathbb{R}_{> 0})$.
	We also see that (\cref{eqn: N'diagOrder,eqn: N'diagDominance})
	$$(N')_{11} > (N')_{22} > \cdots > (N')_{qq} > (N')_{ij}
	\text{ for all } i \neq j \in [q].$$
	So, $N'$ is diagonal dominant, strongly diagonal distinct,
	and intra-row-distinct.
	Now, let us assume that there exists some
	permutation $\sigma \in S_{q}$, s.t., 
	$(N')^{\sigma}$ satisfies \ref{item: diagOrder}.
	But we note that $(N')_{11} > \dots > (N')_{qq}$ without the permutation $\sigma$.
	Therefore, for $(N')^{\sigma}$ to satisfy
	\ref{item: diagOrder},
	this permutation $\sigma$ must be the
	trivial identity permutation.
	But then, $\aleph(N') \neq 0$.
	So, $N' \in \Sym{q}{pd}(\mathbb{R}_{> 0})$ is a
	diagonal dominant, strongly diagonal distinct,
	intra-row-distinct matrix that either fails to satisfy
	\ref{item: diagOrder},
	or fails to satisfy \ref{item: LNAleph}.
\end{proof}

To finish the proof
of \cref{lemma: diagDistinctHardness},
we now prove \cref{claim: AlphaGThetaNotConstant}.

\AlphaGThetaNotConstant*
\begin{proof}
	Let $N' = \mathcal{L}_{N}(\alpha)$.
	For convenience, let us rename
	$N'_{ii} = x_{i}$.
	Since $N$ satisfies \ref{item: LNalsoBad},
	we know that $x_{1} > x_{2} > \dots > x_{q}$.
	Since $N$ satisfies \ref{item: LNRow1},
	and \ref{item: LNRowq}, we see that
	$(N')_{i1} = x_{i}$ for all $i \in [q]$,
	and $(N')_{iq} = x_{q}$ for all $i \in [q]$.
	Since $N$ satisfies \ref{item: LNsecondLastMin},
	and \ref{item: LNsecondLastMax},
	we also see that $|\{i \in [q]: (N')_{(q-1)i} = x_{q}\}|
	= \nicefrac{q}{2}$, and
	$|\{i \in [q]: (N')_{(q-1)i} = x_{q-1}\}|
	= \nicefrac{q}{2}$.
	So,
	$$(\mathcal{T}_{N'}(\theta)^{2})_{11} = 
	\sum_{i \in [q]} (x_{i})^{2\theta}, \quad
	(\mathcal{T}_{N'}(\theta)^{2})_{qq} = 
	q(x_{q})^{2\theta}, \quad 
	(\mathcal{T}_{N'}(\theta)^{2})_{(q-1)(q-1)} = 
	\frac{q}{2}(x_{q})^{2\theta} + \frac{q}{2}(x_{q-1})^{2\theta},$$
	$$(\mathcal{T}_{N'}(\theta)^{2})_{1q} = 
	(x_{q})^{\theta}\sum_{i \in [q]} (x_{i})^{\theta}, \quad
	(\mathcal{T}_{N'}(\theta)^{2})_{(q-1)q} = 
	\frac{q}{2}(x_{q})^{2\theta} + \frac{q}{2}(x_{q}x_{q-1})^{\theta}.$$
	For convenience, we let
	$$S_{1}(\theta) = \sum_{i \in [q]}(x_{i})^{\theta}, \quad
	S_{2}(\theta) = \sum_{i \in [q]} (x_{i})^{2\theta}.$$
	So, we see that
	\begin{align*}
		\aleph(\mathcal{T}_{N'}(\theta)^{2})
		&= \Big(\ln\left((x_{q})^{\theta}S_{1}(\theta)\right)
		- \ln\left(q(x_{q})^{2\theta}\right) \Big)
		\Big(\ln\left(\frac{q}{2}(x_{q})^{2\theta} 
		+ \frac{q}{2}(x_{q-1})^{2\theta}\right) 
		- \ln\left(q(x_{q})^{2\theta}\right)\Big)\\
		& \qquad - 
		\Big(\ln\left(S_{2}(\theta)\right) 
		- \ln\left(q(x_{q})^{2\theta}\right)\Big)
		\Big(\ln\left(\frac{q}{2}(x_{q})^{2\theta} 
		+ \frac{q}{2}(x_{q}x_{q-1})^{\theta}\right) 
		- \ln\left(q(x_{q})^{2\theta}\right) \Big)\\
		&= \left(\ln\left(\sum_{i \in [q]}
		(\nicefrac{x_{i}}{x_{q}})^{\theta}
		\right) -\ln(q)\right)
		\left(\ln\left(1 + (\nicefrac{x_{q-1}}{x_{q}})^{2\theta}\right)
		- \ln(2)\right)\\
		& \qquad - 
		\left(\ln\left(\sum_{i \in [q]}
		(\nicefrac{x_{i}}{x_{q}})^{2\theta}
		\right) -\ln(q)\right)
		\left(\ln\left(1 + (\nicefrac{x_{q-1}}{x_{q}})^{\theta}\right)
		- \ln(2)\right).
	\end{align*}
	
	Now, let us assume that
	$\aleph(\mathbf{S_{2}}(\mathcal{K}_{N}(1, \alpha, \theta))) = 0$
	for  all $\theta \in \mathbb{R}$.
	This implies that 
	$\aleph(\mathcal{T}_{N'}(\theta)^{2}) = 0$
	for all $\theta \in \mathbb{R}$.
	This means that $L_{1}(\theta)D_{2}(\theta) 
	= L_{2}(\theta)D_{1}(\theta)$ for all $\theta \in \mathbb{R}$,
	where $L_{1}(\theta) = \ln(\sum_{i}
	(\nicefrac{x_{i}}{x_{q}})^{\theta}) - \ln(q)$,
	$L_{2}(\theta) = \ln(\sum_{i}
	(\nicefrac{x_{i}}{x_{q}})^{2\theta}) - \ln(q)$,
	$D_{1}(\theta) = \ln(1 + (\nicefrac{x_{q-1}}{x_{q}})^{\theta})
	- \ln(2)$, and
	$D_{2}(\theta) = \ln(1 + (\nicefrac{x_{q-1}}{x_{q}})^{2\theta})
	- \ln(2)$.
	
	We will now consider the asymptotic behavior of
	$L_{1}(\theta), L_{2}(\theta),  D_{1}(\theta), D_{2}(\theta)$
	as $\theta \rightarrow - \infty$.    
	We let $\epsilon(\theta) = (\nicefrac{x_{q-1}}{x_{q}})^{\theta}$.
	Then, the Taylor series expansions of $D_{1}(\theta)$
	and $D_{2}(\theta)$
	tell us that
	$$D_{1}(\theta) = \epsilon(\theta) + o(\epsilon(\theta))
	- \ln(2), \qquad D_{2}(\theta) =
	\epsilon(\theta)^{2} + o(\epsilon(\theta)^{2}) - \ln(2).$$
	Since $x_{1} > \dots > x_{q-1} > x_{q}$,
	we see that $(\nicefrac{x_{i}}{x_{q}})^{\theta} =
	o(\epsilon(\theta))$ for all $i \in [q-2]$.
	So, the Taylor series expansions of $L_{1}(\theta)$ and
	$L_{2}(\theta)$ tell us that
	$$L_{1}(\theta) = \epsilon(\theta) + o(\epsilon(\theta))
	- \ln(q), \qquad
	L_{2}(\theta) = \epsilon(\theta)^{2} + o(\epsilon(\theta)^{2})
	- \ln(q).$$
	So, we find that
	$$L_{1}(\theta)D_{2}(\theta) = \ln(2)\ln(q) - \ln(2)\epsilon(\theta) + o(\epsilon(\theta)), \qquad
	L_{2}(\theta)D_{1}(\theta) = \ln(2)\ln(q) - \ln(q)\epsilon(\theta)
	+ o(\epsilon(\theta)).$$
	But since $q \geq 3$, we know that $\ln(2) \neq \ln(q)$.
	This proves the claim. 
\end{proof}

With that, the proof of \cref{theorem: diagonalDistinctHardness}
is complete.

%% file: A5.Diagonal_Distinct_Dichotomy.tex
\section{Dichotomy for diagonal-distinct matrices}\label{appendix: diagonalDistinctDichotomy}

\subsection{Hardness of full rank matrices}

We shall first prove
\cref{lemma: fullRankHardness}.

\fullRankHardness*
\begin{proof}
	Since $M$ is diagonal distinct, we may
	assume $M_{11} > M_{22} > 
	\dots > M_{qq}$.
	So, we see that
	$$\lim_{\theta \rightarrow \infty}
	\left(\frac{(\mathcal{L}_{M}(\theta)^{2})_{ii}}{(M_{11}M_{ii})^{2 \theta}}\right) = M_{i1}^2, ~~
	\text{for all } i \in [q].$$
	This implies that
	$$\lim_{\theta \rightarrow \infty}
	\frac{(\mathcal{L}_{M}(\theta)^{2})_{jj}}
	{(\mathcal{L}_{M}(\theta)^{2})_{ii}} = 0,
	\quad \text{for all } i < j \in [q].$$
	Since $\det(\mathcal{L}_{M}(0)) = \det(M) \neq 0$,
	\cref{lemma: analyticGadgets}
	lets us pick a large 
	$\theta^{*}$ such that $\mathcal{L}_{M}(\theta^{*})$ has full rank,
	$N = \mathbf{S_{2}}(\mathcal{L}_{M}(\theta^{*}))$ is positive definite
	as a square of a full ranked symmetric matrix, thus
	$N  \in 
	\Sym{q}{pd}(\mathbb{R}_{> 0})$, 
	and is diagonal distinct.
	From \cref{theorem: diagonalDistinctHardness}, 
	$\PlGH(N) \leq \PlGH(M)$ is $\#$P-hard.
\end{proof}

\subsection{Dichotomy for positive valued matrices}

Let us first briefly discuss
the gadget used in
~\cite{govorov2020dichotomy}.
Given a planar multi-graph
$G = (V, E)$ with a planar embedding,
the graph $\mathbf{R_{1, 1}}G$ is constructed as follows:

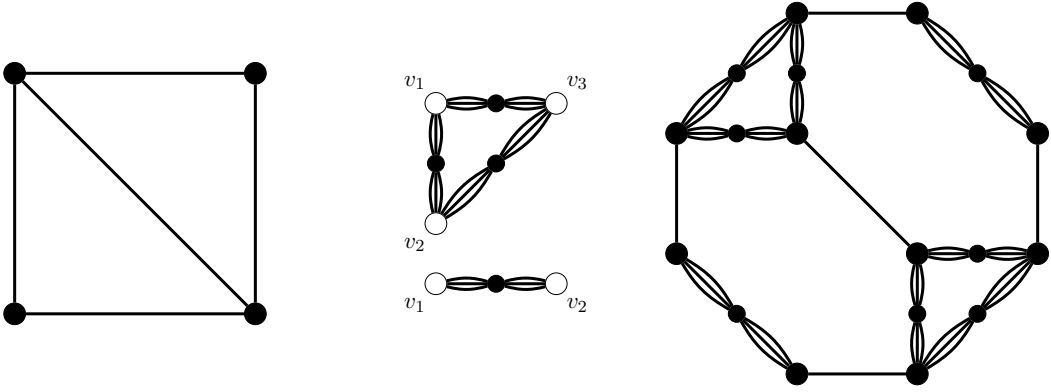
\begin{figure}
	\centering
	\scalebox{0.8}{\input{images/ring}}
	\caption{A graph $G$, the ring gadget $\mathbf{R_{2, 3}}$
		on vertices of degree $3$ and $2$, and $\mathbf{R_{2, 3}}G$.}
	\label{fig:ring}
\end{figure}

We first label the edges incident to each vertex
in a clockwise manner.  
For each vertex $v \in V$,
we arbitrarily pick one incident edge and label it $1$.  
Moving clockwise around $v$,
we label the remaining edges $2, 3, \dots, \deg(v)$.  
This defines a labeling function
$$\ell_v: \{ (u,v) \in E \} \rightarrow [\deg(v)],$$
which we apply to all vertices $v \in V$.

We then construct $\mathbf{R_{1, 1}}G$ by replacing each vertex
$v \in V$ with a simple cycle of length $\deg(v)$.  
The vertices of this cycle are denoted
$(v_1, v_2, \dots, v_{\deg(v)})$,
ordered clockwise according to $\ell_v$.
Intuitively, this cycle forms a “ring” around $v$,
with each $v_i$ corresponding to one incident edge of $v$.

For each edge $(u, v) \in E$, let $\ell_u((u,v)) = i$
and $\ell_v((u,v)) = j$.  
We then connect $u_i$ and $v_j$ with an edge
that replaces the original edge  $(u,v)$. By making the cycles sufficiently small,
planarity is preserved globally.
Therefore, $\mathbf{R_{1, 1}}G$ is a planar multi-graph.

Now, we can construct $\mathbf{R_{m, n}}G$ for $m, n \geq 1$
as follows:
First, we construct $\mathbf{R_{1, 1}}G$.
Then, we replace every edge along the simple cycles
(i.e., excluding the edges already present in $G$),
with a path of length $m$, and then
replace each of these $m$ edges along each of 
these paths with $n$ parallel edges.
Clearly, we see that $\mathbf{R_{m, n}}G$ is 
also a planar multi-graph.

This implies that all proofs in
~\cite{govorov2020dichotomy} that use
the $\mathbf{R_{m, n}}G$ gadget construction
to prove $\GH(N) \leq \GH(M)$ can be
strengthened to prove that
$\PlGH(N) \leq \PlGH(M)$ with
no additional effort.
We will now discuss some of the results from
~\cite{govorov2020dichotomy}.

\begin{notation}\label{notation: RMsimilarity}
	Let $M \in \Sym{q}{}(\mathbb{R}_{> 0})$.
	For $i, j \in [q]$, we let $i \sim j$ iff
	there exists a constant $c$ such that
	$M_{ik} = c \cdot M_{jk}$ for all $k \in [q]$.
	We let $\mathcal{I}(M)$ be the equivalence
	classes on $[q]$ formed by this relation.
\end{notation}

We note that for any
$M \in \Sym{q}{}(\mathbb{R}_{> 0})$,
$\text{rank}(M) \leq |\mathcal{I}(M)| \leq q$.
We can now define $\iota: \mathcal{I}(M)
\rightarrow [q]$ as
$\iota(I) = i$, for an arbitrary $i \in I$.
Note that our choice of $\iota$ may be arbitrary,
but once this $\iota$ is fixed, we can now
define the following:

\begin{definition}\label{definition: mathcalR}
	Let $M \in \Sym{q}{}(\mathbb{R}_{> 0})$,
	with $\mathcal{I}(M)$ and $\iota: \mathcal{I}(M)
	\rightarrow [q]$ defined as above.
	We now define $w_{\theta}: \mathcal{I}(M)
	\rightarrow \mathbb{R}$ as:
	$$w_{\theta}(I) = \sum_{i \in I}
	\left(\frac{M_{ii}}{M_{\iota(I)\iota(I)}}\right)^{\theta}.$$
	Then, we define $\mathcal{R}_{M}(\theta) 
	\in \Sym{|\mathcal{I}(M)|}{}(\mathbb{R}_{> 0})$,
	such that
	$$\mathcal{R}_{M}(\theta)_{IJ}
	= M_{\iota(I)\iota(J)}
	\left(\frac{w_{(2\theta+1)}(I)}{w_{2\theta}(I)}\right)
	\left(\frac{w_{(2\theta+1)}(J)}{w_{2\theta}(J)}\right).$$
\end{definition}

\begin{remark*}
	Note that $\mathcal{R}_{M}(\theta)$ is
	dependent on the choice of $\iota$, but
	for any fixed $\iota$, $\mathcal{R}_{M}(\theta)$
	is well-defined.
\end{remark*}

The following theorem is the technical statement
proved in Theorem 3.1 of
~\cite{govorov2020dichotomy}.

\begin{theorem}[\cite{govorov2020dichotomy}]
	\label{theorem: twinReduction}
	Let $M \in \Sym{q}{}(\mathbb{R}_{> 0})$.
	Then, there exists some $n^{*} \in \mathbb{Z}_{> 0}$,
	such that for all integers $n > n^{*}$,
	$\mathcal{R}_{M}(n) \in \Sym{|\mathcal{I}(M)|}{F}(\mathbb{R}_{> 0})$,
	and $\PlGH(\mathcal{R}_{M}(n)) \leq \PlGH(M)$.
\end{theorem}

\begin{remark*}
	We note once again that this choice of $n^{*}$
	is dependent on $\iota$, but regardless
	of our choice of $\iota$, 
	\cref{theorem: twinReduction} proves that
	some $n^{*} \in \mathbb{Z}_{> 0}$ exists.
\end{remark*}

We can now prove the following dichotomy theorem.

\positiveDichotomy*
\begin{proof}
	First, we assume $M = \mathbf{x}\mathbf{x}^{\tt{T}}$
	for some $\mathbf{x} = (x_{1}, \dots, x_{q})$.
	Then, 
	given any
	graph $G = (V, E)$,
	$$Z_{M}(G) = \sum_{\sigma: V \rightarrow [q]}
	\prod_{(u, v) \in E}(x_{\sigma(u)}x_{\sigma(v)})
	= \prod_{v \in V}\left(\sum_{i \in [q]}x_{i}^{\deg(v)}\right)$$
	can be computed in polynomial time.
	So,
	$\PlGH(M)$ (as well as $\GH(M)$) is polynomial time
	tractable.
	
	Now, let us assume that $M$ has rank $\geq 2$.
	Then, $|\mathcal{I}(M)| \geq 2$.
	Since the choice of $\iota$ may be arbitrary,
	we may choose $\iota$, such that
	$\iota(I) = \arg\max_{i \in I}M_{ii}$.
	We note that since $M$ is diagonal distinct,
	this $\iota$ is uniquely defined.
	Now, for all $I \in \mathcal{I}(M)$,
	$$\lim_{\theta \rightarrow \infty}w_{\theta}(I) = 
	\lim_{\theta \rightarrow \infty}\left( 1 + \sum_{
		i \in I: i \neq \iota(I)}(\nicefrac{M_{ii}}
	{M_{\iota(I)\iota(I)}})^{\theta}\right) = 1.$$
	This implies that
	for all $I \neq J \in \mathcal{I}(M)$,
	$$\lim_{\theta \rightarrow \infty}
	\frac{\mathcal{R}_{M}(\theta)_{II}}
	{\mathcal{R}_{M}(\theta)_{JJ}} = 
	\lim_{\theta \rightarrow \infty} 
	\frac{M_{\iota(I)\iota(I)}}
	{M_{\iota(J)\iota(J)}}
	\left(\frac{w_{(2\theta + 1)}(I)}{w_{2\theta}(I)}\right)^{2}
	\left(\frac{w_{2\theta}(J)}{w_{(2\theta+1)}(J)}\right)^{2}
	= \frac{M_{\iota(I)\iota(I)}}
	{M_{\iota(J)\iota(J)}} \neq 1.$$
	So, we can find some integer $n > n^{*}$
	such that
	$\mathcal{R}_{M}(n) \in \Sym{|\mathcal{I}(M)|}{F}
	(\mathbb{R}_{> 0})$
	is diagonal distinct.
	So, \cref{lemma: fullRankHardness} and
	\cref{theorem: twinReduction} imply that
	$\PlGH(\mathcal{R}_{M}(n)) \leq \PlGH(M)$ is $\#$P-hard.
\end{proof}

\subsection{Dichotomy for non-negative valued matrices}

We will now end this section with the proof
of \cref{theorem: nonNegativeDichotomy}.
We note that the proof of a lemma in~\cite{cai2013graph}
also works in the planar setting, and can be
restated as follows.

\begin{theorem}[\cite{cai2013graph} 
	(Lemma 4.6, p.~940,
	its proof uses 
	Lemma 4.1, p.~937, called the  first pinning lemma)]
	\label{theorem: domainSeparable}
	Let $M \cong M_{1} \oplus \dots \oplus M_{r} \in
	\Sym{q}{}(\mathbb{R})$.
	Then, $\PlGH(M)$ is polynomial time tractable if
	all $\PlGH(M_{i})$ are polynomial time tractable, and
	is $\#$P-hard 
	if $\PlGH(M_{i})$ is $\#$P-hard  for some $i \in [r]$.
\end{theorem}

To prove \cref{theorem: nonNegativeDichotomy},
we shall also need the following lemma.

\begin{lemma}\label{lemma: BDiagonalDistinct}
	Let $M \in \Sym{q}{}(\mathbb{R}_{\geq 0})$ be
	diagonal distinct.
	Then, there exist $\theta_{1}, \theta_{2} \in \mathbb{R}$
	such that $\mathcal{B}_{\mathcal{T}_{M}(\theta_{1})}(\theta_{2})
	\in \Sym{q}{}(\mathbb{R}_{\geq 0})$ is also
	diagonal distinct.
\end{lemma}
\begin{proof}
	We write $\Theta$ for $(\theta_{1}, \theta_{2})
	\in \mathbb{R}^{2}$.
	Let $\mathcal{K}_{M}(\Theta) 
	= \mathcal{B}_{\mathcal{T}_{M}(\theta_{1})}(\theta_{2})$.
	Note that for all $i, j \in [q]$,
	$$\mathcal{K}_{M}(\Theta)_{ij}
	= \sum_{x, y \in [q]}(M_{ix}M_{jy})^{\theta_{1}}
	(M_{xy})^{\theta_{1} \theta_{2}}.$$
	Let us fix some arbitrary $i \neq j \in [q]$.
	Assume that for all $\Theta \in \mathbb{R}^{2}$,
	$\mathcal{K}_{M}(\Theta)_{ii} = 
	\mathcal{K}_{M}(\Theta)_{jj}$.
	For all $k \in [q]$, this means that
	the coefficients of $(M_{kk})^{\theta_{1} \theta_{2}}$
	must be equal. So,
	$$\sum_{\substack{(a, b) \in [q]:\\
			M_{ab} = M_{kk}}}(M_{ia}M_{ib})^{\theta_{1}}
	= \sum_{\substack{(a, b) \in [q]:\\
			M_{ab} = M_{kk}}}(M_{ja}M_{jb})^{\theta_{1}}.$$
	But since $M$ is diagonal distinct, we see that
	$M_{ab} = M_{kk}$ implies that either
	$(a, b) = (k, k)$, or $a \neq b$.
	In other words, we have that
	$$(M_{ik})^{2\theta_{1}} + 
	\sum_{\substack{a\ <\ b\ \in\ [q]:\\
			M_{ab} = M_{kk}}}2(M_{ia}M_{ib})^{\theta_{1}}
	= (M_{jk})^{2\theta_{1}} + 
	\sum_{\substack{a\ <\ b\ \in\ [q]:\\
			M_{ab} = M_{kk}}}2(M_{ja}M_{jb})^{\theta_{1}}.$$
	
	The term $(M_{ik})^{2\theta_{1}}$, after possibly combining with
	other equal terms of the left side of the equation, has an odd coefficient.
	If this term does not coincide with  the term $(M_{jk})^{2\theta_{1}}$
	on the right side of the equation then this term would have an even 
	coefficient, which is a contradiction. Thus,
	since this is true for all $\theta_{1} \in \mathbb{R}$,
	it must be the case that
	$M_{ik} = M_{jk}$.
	But since that is true for all $k \in [q]$,
	we find that $M_{ii} = M_{ji} = M_{ij} = M_{jj}$,
	which contradicts our assumption
	that $M$ is diagonal distinct.
	Therefore, for all $i, j \in [q]$,
	there  exists some
	$\Theta_{ij} = (\theta_{ij, 1}, \theta_{ij, 2})
	\in \mathbb{R}^{2}$,
	such that
	$\mathcal{K}_{M}(\Theta_{ij})_{ii} \neq 
	\mathcal{K}_{M}(\Theta_{ij})_{jj}$.
	
	We may now define $\Sym{q}{}(\mathbb{R})$-polynomials
	$\zeta_{ij}$ for $i \neq j \in [q]$, such that
	$\zeta_{ij}: A \mapsto 
	(A_{ii} - A_{jj})$.
	Now, we can use \cref{lemma: analyticGadgets} to find
	$\Theta^{*} = (\theta_{1}^{*}, \theta_{2}^{*})
	\in \mathbb{R}^{2}$, such that
	$\zeta_{ij}(\mathcal{K}_{M}(\Theta^{*})) \neq 0$
	for all $i \neq j \in [q]$.
	This $\mathcal{K}_{M}(\Theta^{*})$ is the required
	diagonal distinct matrix.
\end{proof}

\nonNegativeDichotomy*
\begin{proof}
	The theorem is obvious if $q = 1$.
	So we assume $q \geq 2$.
	
	We write $M \cong M'$ if for some permutation $\sigma \in S_q$,
	we have $M^{(\sigma)} = M'$, where the matrix $M^{(\sigma)}$
	is obtained by permuting rows and columns of $M$ simultaneously according to
	$\sigma$.
	It is clear from
	\cref{theorem: domainSeparable} that
	if 
	$M \cong M_{1} \oplus \dots \oplus M_{r}$
	for rank $1$ or rank $0$ matrices $M_{i}$,
	then $\PlGH(M)$ is polynomial time tractable.
	Now, let us assume that if
	$M \cong M_{1} \oplus \dots \oplus M_{r}$, then at least
	one of $M_{i}$ has rank at least $2$.
	
	Suppose  $M$ has a row/column
	that is all-0.
	Since $M$ is diagonal distinct,
	it has at most one such row/column.
	We may permute the rows and columns
	of $M$ so that this zero row/column is
	the last row/column.
	In this case,
	$M = M' \oplus \mathbf{0}$ for some
	$M' \in \Sym{q-1}{}(\mathbb{R}_{\geq 0})$
	and $\mathbf{0} \in \Sym{1}{}(\mathbb{R})$ is the
	matrix with the single zero entry.
	Trivially, $\PlGH(M) \equiv \PlGH(M')$.
	The dichotomy statement of the theorem
	for $M$ and the corresponding statement
	for $M'$ are equivalent.
	So we may replace our matrix $M$ with $M'$,
	such that $M$ has no rows/columns with all-0 entries.
	
	We now define a graph (with possible loops) $\Gamma$ on $[q]$
	where an edge  $(i,j) \in [q] \times [q]$ exists
	iff $M_{i,j} \ne 0$. Note that $i=j$ is allowed, and it is a  loop if $M_{i,i} \ne 0$.
	Since $M$ is diagonal distinct,
	all vertices of $\Gamma$ have self-loops,
	except possibly one (and if so, we name it $q$).
	We  can decompose $M$ as a direct
	sum of blocks that correspond to the connected
	components of $\Gamma$.
	So, $M \cong M_{1} \oplus \dots \oplus M_{r}$
	for some $r \leq q$, where each
	$M_{i} \in \Sym{q_{i}}{}(\mathbb{R}_{\geq 0})$,
	for $\sum_{i}q_{i} = q$.
	
	We  claim that $i$ and $j$ belong to
	the same connected component of $\Gamma$ 
	if and only if $(M^q)_{ij} > 0$
	The reverse direction is obvious.
	We prove the forward direction.
	This is also obvious if
	every vertex has a self-loop.
	Suppose vertex $q$ has no self-loop.
	If $i \neq j$ or $i = j \neq q$, 
	then at least one of $i$ or $j$ has a self-loop,
	and so $(M^q)_{ij} > 0$. If $i = j = q$,
	since the $q$-th row is not
	all-0, there exists some $i \neq q$, such that 
	$M_{iq} > 0$. 
	Then, as $q \geq 2$, by using a loop at $i$
	if $q$ is odd, we find that $(M^q)_{qq} > 0$.
	Consequently, we find that
	$M^s \cong M_{1}^{s} \oplus \dots 
	\oplus M_{r}^{s}$, where
	$M_{i}^{s} \in \Sym{q_{i}}{}(\mathbb{R}_{> 0})$
	for all $s \geq q$.
	
	
	Recall that we had assumed that if
	$M \cong M_{1} \oplus \dots \oplus M_{r}$, then at least
	one of $M_{i}$ has rank at least $2$.
	Without loss of generality, we may
	assume that ${\rm rank}(M_{1}) \geq 2$. In particular, $q_1 \ge 2$.
	
	Let  $\Theta = (\theta_{1}, \theta_{2})$.  We will now consider
	$\mathcal{K}_{M}(\Theta) =
	\mathcal{B}_{\mathcal{T}_{M}(\theta_{1})}(\theta_{2})$.
	We note that
	$\mathcal{K}_{M}(\Theta) \cong
	\mathcal{K}_{M_{1}}(\Theta)
	\oplus \dots \oplus 
	\mathcal{K}_{M_{r}}(\Theta)$,
	for all $\Theta \in \mathbb{R}^{2}$.
	Note that the entries of 
	$\mathcal{K}_{M_{1}}(\Theta)$ 
	are continuous in $\theta_{1}$ and $\theta_{2}$,
	and the eigenvalues of a matrix are continuous
	as a function of the entries of the matrix.
	Moreover, since 
	$\mathcal{K}_{M_{1}}(1, 1) = M_{1}^{3}$, 
	we see that
	${\rm rank}(\mathcal{K}_{M_{1}}(1, 1)) = 
	{\rm rank}(M_{1}^{3}) = {\rm rank}(M_{1}) \geq 2$,
	(that the rank equality ${\rm rank}(M_{1}^{3}) = {\rm rank}(M_{1})$ holds 
	is easily seen by the spectral decomposition for real symmetric matrices).
	We see that there exists
	some open neighborhood
	$\mathcal{O} \subset \mathbb{R}^{2}$ around
	$(1, 1)$ such that ${\rm rank}(\mathcal{K}_{M_{1}}(\Theta))
	= {\rm rank}(\mathcal{K}_{M_{1}}(1, 1)) \geq 2$,
	for all $\Theta \in \mathcal{O}$.
	
	Given any $B \in \Sym{q}{}(\mathbb{R}_{\geq 0})$,
	we  define a $\Sym{q}{}(\mathbb{R})$-polynomial
	$$\zeta_{B}: A \mapsto \prod_{\substack{i, j \in [q]:\\
			B_{ij} > 0}}(A_{ij}).$$
	Since $\mathcal{K}_{M}(1, 1) = M^3$, we see that trivially,
	$\zeta_{M^{3}}(\mathcal{K}_{M}(1, 1)) \neq 0$.
	So, using \cref{lemma: BDiagonalDistinct},
	we can find some $\Theta^{*} = 
	(\theta_{1}^{*}, \theta_2^{*}) \in
	\mathcal{O}$, such that
	$N_{1} = \mathcal{K}_{M}(\Theta^{*}) \in
	\Sym{q}{}(\mathbb{R}_{\geq 0})$
	satisfies:
	$N_{1} \cong N_{1, 1} \oplus \dots \oplus N_{1,  r}$
	for $N_{1, i} \in \Sym{q_{i}}{}(\mathbb{R}_{\geq 0})$,
	$\phi_{\rm{diag}}(N_{1}) \neq 0$,
	$\zeta_{M^{3}}(N_{1}) \neq 0$,
	${\rm rank}(N_{1, 1}) \geq 2$,
	and $\PlGH(N_{1}) \leq \PlGH(M)$.
	From our construction of
	the graph $\Gamma$, we see that
	$(M^{3})_{ij} > 0$ if and only if
	$i$ and $j$ are connected in $\Gamma$
	by a path of length $3$.
	So, we see that since $\zeta_{M^{3}}(N_{1}) \neq 0$,
	$(N_{1})_{ij} > 0$ whenever
	$\Gamma$ has a path of length $3$ between $i$ and $j$.
	
	We now repeat this process all over again
	with $N_{1}$ in place of $M$, to find
	$N_{2} \in \Sym{q}{}(\mathbb{R}_{\geq 0})$,
	such that
	$N_{2} \cong N_{2, 1} \oplus \dots \oplus N_{2, r}$,
	$\phi_{\rm{diag}}(N_{2}) \neq 0$,
	$\zeta_{(N_{1})^{3}}(N_{2}) \neq 0$,
	${\rm rank}(N_{2, 1}) \geq 2$, and
	$\PlGH(N_{2}) \leq \PlGH(N_{1}) \leq \PlGH(M)$.
	We recall that $(N_{1})_{ij} > 0$ whenever
	$\Gamma$ has a path of length $3$
	between $i$ and $j$.
	So, $((N_{1})^{3})_{ij} > 0$ whenever
	$\Gamma$ has a path of length $9$ between 
	$i$ and $j$.
	So, $\zeta_{(N_{1})^{3}}(N_{2}) \neq 0$ 
	implies that $\zeta_{M^{9}}(N_{2}) \neq 0$.
	
	In effect, we find that if we repeat this process up to
	$t = \lceil \log_{3}(q) \rceil$  times,
	we construct
	$N_{t} \cong N_{t, 1} \oplus
	\dots \oplus N_{t, r}$, such that
	$\phi_{\rm{diag}}(N_{t}) \neq 0$,
	$\zeta_{M^{(3^t)}}(N_{t}) \neq 0$,
	${\rm rank}(N_{t, 1}) \geq 2$, and
	$\PlGH(N_{t}) \leq \PlGH(M)$.
	
	Since $(M_{1})^{(3^t)} \in \Sym{q_1}{}(\mathbb{R}_{> 0})$
	because of our choice of $t$,
	we see that $N_{t, 1} \in \Sym{q_1}{}(\mathbb{R}_{> 0})$ with
	${\rm rank}(N_{t, 1}) \geq 2$, and
	since $\phi_{\rm{diag}}(N_{t}) \neq 0$,
	$N_{t, 1}$ is diagonal distinct.
	So, from \cref{theorem: positiveDichotomy},
	we see that $\PlGH(N_{t, 1})$ is $\#$P-hard.
	So, from \cref{theorem: domainSeparable},
	we see that $\PlGH(N_{t}) \leq \PlGH(M)$ is
	also $\#$P-hard.
\end{proof}

%% file: images/ring.tex
\begin{tikzpicture}[line join=miter, draw opacity=1]

    \node[circle, draw=black, fill=black] (A) at (-2, -2){};
    \node[circle, draw=black, fill=black] (B) at (-2, 2){};
    \node[circle, draw=black, fill=black] (C) at (2, 2){};
    \node[circle, draw=black, fill=black] (D) at (2, -2){};
    
    \draw[line width=0.5mm, black] (A) -- (B);
    
    \draw[line width=0.5mm, black] (B) -- (C);
    
    \draw[line width=0.5mm, black] (C) -- (D);
    
    \draw[line width=0.5mm, black] (D) -- (A);
    
    \draw[line width=0.5mm, black] (B) -- (D);

\begin{scope}[xshift = 6cm]

    \node[circle, draw=black, fill=white] (A1) at (-1, -1.5){};
    \node[circle, draw=black, fill=white] (A2) at (1, -1.5){};

    \node[draw=none] at (-1.35, -1.85) {$v_{1}$};
    \node[draw=none] at (1.35, -1.85) {$v_{2}$};

    \node[circle, draw=black, fill=black, inner sep=0pt, minimum size=8] (A1A2) at (0, -1.5){};

    \foreach \angle in {15} {
      \draw[line width=0.5mm] (A1) to[bend left=\angle] (A1A2);
      \draw[line width=0.5mm] (A1) to[bend right=\angle] (A1A2);
    }
    \draw[line width=0.5mm] (A1) to (A1A2);
    \foreach \angle in {15} {
      \draw[line width=0.5mm] (A1A2) to[bend left=\angle] (A2);
      \draw[line width=0.5mm] (A1A2) to[bend right=\angle] (A2);
    }
    \draw[line width=0.5mm] (A1A2) to (A2);


    \node[circle, draw=black, fill=white] (D1) at (-1, 1.5){};
    \node[circle, draw=black, fill=white] (D2) at (-1, -0.5){};
    \node[circle, draw=black, fill=white] (D3) at (1, 1.5){};

    \node[draw=none] at (-1.35, 1.85) {$v_{1}$};
    \node[draw=none] at (-1.35, -0.85) {$v_{2}$};
    \node[draw=none] at (1.35, 1.85) {$v_{3}$};

    \node[circle, draw=black, fill=black, inner sep=0pt, minimum size=8] (D1D2) at (-1, 0.5){};
    \node[circle, draw=black, fill=black, inner sep=0pt, minimum size=8] (D2D3) at (0, 0.5){};
    \node[circle, draw=black, fill=black, inner sep=0pt, minimum size=8] (D3D1) at (0, 1.5){};

    \foreach \angle in {15} {
      \draw[line width=0.5mm] (D1) to[bend left=\angle] (D1D2);
      \draw[line width=0.5mm] (D1) to[bend right=\angle] (D1D2);
    }
    \draw[line width=0.5mm] (D1) to (D1D2);
    \foreach \angle in {15} {
      \draw[line width=0.5mm] (D1D2) to[bend left=\angle] (D2);
      \draw[line width=0.5mm] (D1D2) to[bend right=\angle] (D2);
    }
    \draw[line width=0.5mm] (D1D2) to (D2);

    \foreach \angle in {15} {
      \draw[line width=0.5mm] (D2) to[bend left=\angle] (D2D3);
      \draw[line width=0.5mm] (D2) to[bend right=\angle] (D2D3);
    }
    \draw[line width=0.5mm] (D2) to (D2D3);
    \foreach \angle in {15} {
      \draw[line width=0.5mm] (D2D3) to[bend left=\angle] (D3);
      \draw[line width=0.5mm] (D2D3) to[bend right=\angle] (D3);
    }
    \draw[line width=0.5mm] (D2D3) to (D3);

    \foreach \angle in {15} {
      \draw[line width=0.5mm] (D3) to[bend left=\angle] (D3D1);
      \draw[line width=0.5mm] (D3) to[bend right=\angle] (D3D1);
    }
    \draw[line width=0.5mm] (D3) to (D3D1);
    \foreach \angle in {15} {
      \draw[line width=0.5mm] (D3D1) to[bend left=\angle] (D1);
      \draw[line width=0.5mm] (D3D1) to[bend right=\angle] (D1);
    }
    \draw[line width=0.5mm] (D3D1) to (D1);
    
\end{scope}

\begin{scope}[xshift = 12cm]

    \node[circle, draw=black, fill=black] (A1) at (-3, -1){};
    \node[circle, draw=black, fill=black] (A2) at (-1, -3){};

    \node[circle, draw=black, fill=black, inner sep=0pt, minimum size=8] (A1A2) at (-2, -2){};

    \foreach \angle in {15} {
      \draw[line width=0.5mm] (A1) to[bend left=\angle] (A1A2);
      \draw[line width=0.5mm] (A1) to[bend right=\angle] (A1A2);
    }
    \draw[line width=0.5mm] (A1) to (A1A2);
    \foreach \angle in {15} {
      \draw[line width=0.5mm] (A1A2) to[bend left=\angle] (A2);
      \draw[line width=0.5mm] (A1A2) to[bend right=\angle] (A2);
    }
    \draw[line width=0.5mm] (A1A2) to (A2);


    \node[circle, draw=black, fill=black] (B1) at (-1, 1){};
    \node[circle, draw=black, fill=black] (B2) at (-3, 1){};
    \node[circle, draw=black, fill=black] (B3) at (-1, 3){};

    \node[circle, draw=black, fill=black, inner sep=0pt, minimum size=8] (B1B2) at (-2, 1){};
    \node[circle, draw=black, fill=black, inner sep=0pt, minimum size=8] (B2B3) at (-2, 2){};
    \node[circle, draw=black, fill=black, inner sep=0pt, minimum size=8] (B3B1) at (-1, 2){};

    \foreach \angle in {15} {
      \draw[line width=0.5mm] (B1) to[bend left=\angle] (B1B2);
      \draw[line width=0.5mm] (B1) to[bend right=\angle] (B1B2);
    }
    \draw[line width=0.5mm] (B1) to (B1B2);
    \foreach \angle in {15} {
      \draw[line width=0.5mm] (B1B2) to[bend left=\angle] (B2);
      \draw[line width=0.5mm] (B1B2) to[bend right=\angle] (B2);
    }
    \draw[line width=0.5mm] (B1B2) to (B2);

    \foreach \angle in {15} {
      \draw[line width=0.5mm] (B2) to[bend left=\angle] (B2B3);
      \draw[line width=0.5mm] (B2) to[bend right=\angle] (B2B3);
    }
    \draw[line width=0.5mm] (B2) to (B2B3);
    \foreach \angle in {15} {
      \draw[line width=0.5mm] (B2B3) to[bend left=\angle] (B3);
      \draw[line width=0.5mm] (B2B3) to[bend right=\angle] (B3);
    }
    \draw[line width=0.5mm] (B2B3) to (B3);

    \foreach \angle in {15} {
      \draw[line width=0.5mm] (B3) to[bend left=\angle] (B3B1);
      \draw[line width=0.5mm] (B3) to[bend right=\angle] (B3B1);
    }
    \draw[line width=0.5mm] (B3) to (B3B1);
    \foreach \angle in {15} {
      \draw[line width=0.5mm] (B3B1) to[bend left=\angle] (B1);
      \draw[line width=0.5mm] (B3B1) to[bend right=\angle] (B1);
    }
    \draw[line width=0.5mm] (B3B1) to (B1);


    \node[circle, draw=black, fill=black] (C1) at (3, 1){};
    \node[circle, draw=black, fill=black] (C2) at (1, 3){};

    \node[circle, draw=black, fill=black, inner sep=0pt, minimum size=8] (C1C2) at (2, 2){};

    \foreach \angle in {15} {
      \draw[line width=0.5mm] (C1) to[bend left=\angle] (C1C2);
      \draw[line width=0.5mm] (C1) to[bend right=\angle] (C1C2);
    }
    \draw[line width=0.5mm] (C1) to (C1C2);
    \foreach \angle in {15} {
      \draw[line width=0.5mm] (C1C2) to[bend left=\angle] (C2);
      \draw[line width=0.5mm] (C1C2) to[bend right=\angle] (C2);
    }
    \draw[line width=0.5mm] (C1C2) to (C2);


    \node[circle, draw=black, fill=black] (D1) at (1, -1){};
    \node[circle, draw=black, fill=black] (D2) at (1, -3){};
    \node[circle, draw=black, fill=black] (D3) at (3, -1){};

    \node[circle, draw=black, fill=black, inner sep=0pt, minimum size=8] (D1D2) at (1, -2){};
    \node[circle, draw=black, fill=black, inner sep=0pt, minimum size=8] (D2D3) at (2, -2){};
    \node[circle, draw=black, fill=black, inner sep=0pt, minimum size=8] (D3D1) at (2, -1){};
    
    \foreach \angle in {15} {
      \draw[line width=0.5mm] (D1) to[bend left=\angle] (D1D2);
      \draw[line width=0.5mm] (D1) to[bend right=\angle] (D1D2);
    }
    \draw[line width=0.5mm] (D1) to (D1D2);
    \foreach \angle in {15} {
      \draw[line width=0.5mm] (D1D2) to[bend left=\angle] (D2);
      \draw[line width=0.5mm] (D1D2) to[bend right=\angle] (D2);
    }
    \draw[line width=0.5mm] (D1D2) to (D2);

    \foreach \angle in {15} {
      \draw[line width=0.5mm] (D2) to[bend left=\angle] (D2D3);
      \draw[line width=0.5mm] (D2) to[bend right=\angle] (D2D3);
    }
    \draw[line width=0.5mm] (D2) to (D2D3);
    \foreach \angle in {15} {
      \draw[line width=0.5mm] (D2D3) to[bend left=\angle] (D3);
      \draw[line width=0.5mm] (D2D3) to[bend right=\angle] (D3);
    }
    \draw[line width=0.5mm] (D2D3) to (D3);

    \foreach \angle in {15} {
      \draw[line width=0.5mm] (D3) to[bend left=\angle] (D3D1);
      \draw[line width=0.5mm] (D3) to[bend right=\angle] (D3D1);
    }
    \draw[line width=0.5mm] (D3) to (D3D1);
    \foreach \angle in {15} {
      \draw[line width=0.5mm] (D3D1) to[bend left=\angle] (D1);
      \draw[line width=0.5mm] (D3D1) to[bend right=\angle] (D1);
    }
    \draw[line width=0.5mm] (D3D1) to (D1);


    \draw[line width=0.5mm, black] (A1) -- (B2);
    \draw[line width=0.5mm, black] (B3) -- (C2);
    \draw[line width=0.5mm, black] (C1) -- (D3);
    \draw[line width=0.5mm, black] (D2) -- (A2);
    \draw[line width=0.5mm, black] (B1) -- (D1);

\end{scope}

\end{tikzpicture}

%% file: A6.Quantum.tex
\section{Diagonal distinctness and quantum automorphisms}\label{appendix: quantumAppendix}
In this section, we prove the results in
\cref{sec: quantumBody}. Before proving
\cref{theorem: fullGadgetSeparationNonNegative}, we briefly present the theory of $\qut(M)$.

\subsection{The Quantum Automorphism Group}

To apply the theory of \cite{manvcinska2020quantum, cai2023planar}, we need a
more general definition of an edge gadget, in which the labeled vertices
may be the same and the signature may not be symmetric. 
A \emph{2-bilabeled graph} or \emph{binary gadget} $\vk$ is a graph with distinguished
vertices $\ell_1,\ell_2 \in V(\vk)$, possibly with $\ell_1 = \ell_2$.
Define $\vk(M)$ as in \cref{equation: signature}. Note that, if $\ell_1 = \ell_2$, then
$\vk(M)$ is diagonal.
\begin{definition}[$\Gadget(M)$,$\PlGadget(M)$]
	For $M \in \Sym{q}{}(\mathbb{R})$, let $\Gadget(M) := \{\vk(M) : \text{binary gadget } \vk\}$, and define $\PlGadget(M) \subset \Gadget(M)$ to be the signatures of planar binary gadgets.
\end{definition}
Recall that $I$ and $J$ are the identity and all-ones matrices, respectively.
\begin{proposition}\label{proposition: ij}
	For any $M$, we have $I,J \in \PlGadget(M)$. Any $\vk$ with no edges satisfies
	$\vk(M) \in \text{span}_{\mathbb{R}}(I,J)$.
\end{proposition}
\begin{proof}
	Any $\vk$ with no edges is composed of isolated vertices.
	If the two labels of $\vk$ are on distinct vertices,
	then $\vk(M) = q^k J$, and if the labels are
	on the same vertex, then $\vk(M) = q^k I$, where $k$ is the number
	of unlabeled isolated vertices.
\end{proof}

We briefly introduce $\qut(M)$ and its orbits and orbitals. For more background on these topics, see
\cite{manvcinska2020quantum}. A \emph{quantum permutation matrix} $\uc = (u_{ij})_{i,j=1}^q$
is a matrix with entries from a $C^*$-algebra (note that the algebra of bounded linear operators on
some Hilbert space is a prototypical example of a  $C^*$-algebra; for example, in the setting of the nonlocal game defining quantum
isomorphism, the entries
of $\uc$ are the quantum players' measurement operators \cite{atserias2019quantum}) satisfying
\[
\text{(i)}~~ u_{ij}^* = u_{ij} \qquad \text{(ii)}~~ \sum_j u_{ij} = \sum_i u_{ij}
= \mathbf{1} \qquad \text{(iii)}~~ u_{ij} u_{ij'} = \delta_{jj'} u_{ij} ~\text{  and  }~
u_{ij} u_{i'j} = \delta_{ii'} u_{ij},
\]
where $\mathbf{1}$ is identity element of the $C^*$-algebra.
Crucially, the entries of $\uc$ do not necessarily commute. However,
if the entries of $\uc$ come from $\mathbb{C}$, then relations (i), (ii), (iii) force $\uc$ to be an
ordinary (0-1-valued) permutation matrix.

The \emph{quantum automorphism group} $\qut(M)$ of $M$ is a highly abstract object that we
will not precisely define here. It suffices to say that $\qut(M)$ is determined by its
\emph{fundamental representation}, a quantum permutation matrix $\uc$ subject to the
additional constraint $\uc M = M \uc$ \cite{cai2023planar}, an analogue of the fact that an ordinary
permutation matrix $P$ represents a classical automorphism of $M$ if and only if $P M = M P$.
Continuing this analogy, we consider the following constructions.
\begin{definition}[{\cite[Lemma 3.2, Lemma 3.4, Definition 3.5]{lupini2020nonlocal}}] \label{definition: orbits}
	Let $M \in \Sym{q}{}(\mathbb{R})$ and $\uc$ be the fundamental representation of $\qut(M)$.
	The relations $\sim_1$ on $[q]$ defined by $i \sim_1 j \iff u_{ij} \neq 0$ and
	$\sim_2$ on $[q] \times [q]$ defined by $(i,i') \sim_2 (j,j') \iff u_{ij} u_{i'j'} \neq 0$
	are equivalence relations.
	Define the \emph{orbits} and \emph{orbitals} of $\qut(M)$ to be the equivalence classes
	of $\sim_1$ and $\sim_2$, respectively.
\end{definition}
The next statement is a special case of \cite[Theorem 3.10]{lupini2020nonlocal}.
\begin{lemma}
	Let $\uc$ be the fundamental representation of $\qut(M)$ and 
	$K \in \Mat{q}{}(\mathbb{R})$. We have $\uc K = K \uc$ if and only if $K$ is constant
	on the orbitals of $\qut(M)$.
	\label{lemma: orbitals}
\end{lemma}
That $K$ is constant
on the orbitals of $\qut(M)$ means that $(i,i') \sim_2 (j,j')$  implies $K_{ii'} = K_{jj'}$.
The following theorem is a special case of \cite[Theorem 4]{cai2023planar}. It shows that
the span of planar binary gadget signatures 
exactly captures those matrices `invariant under
$\qut(M)$'.
\begin{theorem}
	\label{theorem: duality}
	Let $\uc$ be the fundamental representation of $\qut(M)$ and
	$K \in \Mat{q}{}(\mathbb{C})$. We have $\uc K = K \uc$ if and only if
	$K \in \text{span}_{\mathbb{C}}(\PlGadget(M))$.
\end{theorem}
The fundamental representation $\uc$ of the quantum symmetric group $S_q^+$ is
subject to no relations other than the relations (i),(ii),(iii) making $\uc$ a
quantum permutation matrix.
Equivalently, $S_q^+ = \qut(\mathbf{0})$
(where $\mathbf{0} \in \Sym{q}{}(\mathbb{R})$ denotes
the all-0 matrix), since $\uc \mathbf{0} = \mathbf{0} \uc$
is always satisfied. We now prove the following stronger version of \cref{proposition: quantumSymmetric}.
\begin{proposition}
	Let $M \in \Sym{q}{}(\mathbb{R}_{\geq 0})$. Then the following are
	equivalent.
	\begin{enumerate}[label={\rm(\arabic*)}]
		\item $\qut(M) = S_q^+$,
		\item $\PlEdge(M) \subset {\rm span}_{\mathbb{R}}(I,J)$,
		\item $M \in {\rm span}_{\mathbb{R}}(I,J)$ .
	\end{enumerate}
\end{proposition}
\begin{proof}
	We show (1) $\implies$ (2) $\implies$ (3) $\implies$ (1). 
	Suppose $\qut(M) = S_q^+$. Then $\qut(M)$ and 
	$S_q^+ = \qut(\mathbf{0})$ have the same
	fundamental representation, so, by \cref{theorem: duality}, 
	$\text{span}_{\mathbb{C}}(\PlGadget(M)) = \text{span}_{\mathbb{C}}(\PlGadget(\mathbf{0}))$. 
	If a binary gadget $\vk$ satisfies $\vk(\mathbf{0}) 
	\neq \mathbf{0}$, then
	$\vk$ must have no edges, so, by \cref{proposition: ij}, $\vk(\mathbf{0}) \in \text{span}_{\mathbb{R}}(I,J)$. Therefore
	$\PlEdge(M) \subset \PlGadget(M) \subset \text{span}_{\mathbb{C}}(\PlGadget(M))
	= \text{span}_{\mathbb{C}}(I,J)$, and, since $M$ is real-valued,
	$\PlEdge(M) \subset \text{span}_{\mathbb{R}}(I,J)$.
	
	Assume $\PlEdge(M) \subset \text{span}_{\mathbb{R}}(I,J)$.
	The gadget composed of a single edge with two labeled endpoints has signature
	$M \in \PlEdge(M)$, so $M \in \text{span}_{\mathbb{R}}(I,J)$. 
	
	Finally, assume $M \in \text{span}_{\mathbb{R}}(I,J)$. It is known that $I$ and $J$, and hence all matrices in $\text{span}_{\mathbb{R}}(I,J)$, commute with any quantum permutation matrix $\uc$ (see e.g. \cite[Theorem 3.10]{lupini2020nonlocal}; this also follows from 
	\cref{proposition: ij} or direct calculation). Thus the relation $\uc M = M \uc$ does not impose any new constraints on $\uc$, so $\qut(M) = S_q^+$.
\end{proof}

It is known that $i,j \in [q]$ are in the same orbit of $\qut(M)$ if and only if they are 
indistinguishable by unary planar gadgets \cite[Lemma 10]{cai2023planar}.
For diagonal pairs $(i,i)$ and $(j,j)$, the notion of orbit and orbital are the
same: $(i,i) \sim_2 (j,j)$ if and only if $i \sim_1 j$, as $u_{ij}u_{ij} = u_{ij}$. Therefore similar reasoning with orbitals in place of orbits yields
an analogous result for binary planar gadgets.
\begin{lemma}\label{lemma: diagonal_orbit}
	Let $M \in \Sym{q}{}(\mathbb{R})$. Then
	$\mathbf{K}(M)_{ii} = \mathbf{K}(M)_{jj}$ 
	for every $\mathbf{K}(M) \in \PlGadget(M)$ 
	if and only if $i$ and $j$ are in the same
	orbit of $\qut(M)$.
\end{lemma}
\begin{proof}
	$(\Leftarrow)$: Suppose $u_{ij} \neq 0$, and let 
	$\mathbf{K}(M) \in \PlGadget(M)$.
	\cref{theorem: duality} gives $\uc \mathbf{K}(M) =
	\mathbf{K}(M) \uc$. Hence
	\begin{equation}
		\sum_\ell u_{i\ell} \mathbf{K}(M)_{\ell j} = 
		(\uc \mathbf{K}(M))_{ij} = 
		(\mathbf{K}(M) \uc)_{ij}
		= \sum_\ell \mathbf{K}(M)_{i\ell} u_{\ell j}.
		\label{equation: uk_ku}
	\end{equation}
	Multiplying both sides of \cref{equation: uk_ku} on the left by $u_{ij} \neq 0$ and applying
	property (iii) of a quantum permutation matrix gives
	\[
	u_{ij} \mathbf{K}(M)_{jj} =  
	\sum_\ell u_{ij} u_{i\ell} \mathbf{K}(M)_{\ell j}
	= \sum_\ell \mathbf{K}(M)_{i\ell} u_{ij} u_{\ell j} 
	= \mathbf{K}(M)_{ii} u_{ij},
	\]
	and it follows that $\mathbf{K}(M)_{ii} = \mathbf{K}(M)_{jj}$.
	
	$(\Rightarrow)$: We prove the contrapositive. Suppose $i$ and $j$ are in different orbits
	of $\qut(M)$. Then $(i,i)$ and $(j,j)$ are in different orbitals of
	$\qut(M)$, as $u_{ij} u_{ij} = u_{ij} = 0$. Define $L \in \Mat{q}{}(\mathbb{R})$ by
	$$L_{xy} = \begin{cases} 1 & (x,y) \sim_2 (i,i) \\ 0 & \text{otherwise}\end{cases}.$$
	By definition, $L$ is constant on the orbitals of $\qut(M)$, so
	\cref{lemma: orbitals} and \cref{theorem: duality} give 
	$L \in \text{span}_{\mathbb{C}}(\PlGadget(M))$: there are 
	$c_1,\ldots,c_r \in \mathbb{C}$ and
	$\mathbf{K}_{1}(M),\ldots,\mathbf{K}_{r}(M) \in \PlGadget(M)$ such that
	$L = \sum_{\ell = 1}^r c_\ell \cdot \mathbf{K_{\ell}}(M)$. Now
	\[
	\sum_{\ell = 1}^r c_\ell \cdot 
	\mathbf{K_{\ell}}(M)_{ii} = L_{ii} = 1 \neq 0 = L_{jj}
	= \sum_{\ell = 1}^r c_\ell \cdot 
	\mathbf{K_{\ell}}(M)_{jj},
	\]
	so there is some $\mathbf{K_{\ell}}(M)$ such that
	$\mathbf{K_{\ell}}(M)_{ii} \neq \mathbf{K_{\ell}}(M)_{jj}$.
\end{proof}

If $M \in \Sym{q}{}(\mathbb{R}_{\geq 0})$, then it suffices to consider symmetric planar edge gadget signatures.
\qutOrbits*
\begin{proof}
	By \cref{lemma: diagonal_orbit}, it suffices to show that, 
	if $\mathbf{K'}(M)_{ii} \neq \mathbf{K'}(M)_{jj}$ for
	some $\mathbf{K'}(M) \in \PlGadget(M)$, then there is a
	$\mathbf{K}(M) \in \PlEdge(M)$ such that 
	$\mathbf{K}(M)_{ii} \neq \mathbf{K}(M)_{jj}$. 
	First observe that
	the entries of $\mathbf{K'}(M)$ are nonnegative, 
	as they are sums of products of entries of $M$. 
	$\mathbf{K'}(M)$ can fail to be in $\PlEdge(M)$ in two ways:
	its underlying binary
	gadget $\vk'$ may have both labels on the same vertex, 
	or $\mathbf{K'}(M)$ may not be symmetric.
	In the first case, we have $\mathbf{K'}(M) = 
	\diag(\mathbf{x})$ for 
	some $\mathbf{x} \in (\mathbb{R}_{\geq 0})^q$. 
	Construct an edge gadget $\vk$ as the disjoint union of two copies of $\vk'$,
	but with the originally doubly-labeled vertex of $\vk'$ now labeled $\ell_{1}$ in the
	first copy and $\ell_{2}$ in the second copy. 
	Now $\mathbf{K}(M) := \mathbf{x} \mathbf{x}^{\tt{T}} \in \PlEdge(M)$ and, by
	nonnegativity, $\mathbf{K}(M)_{ii} = 
	(\mathbf{K'}(M)_{ii})^2 \neq 
	(\mathbf{K'}(M)_{jj})^2 = 
	\mathbf{K}(M)_{jj}$.
	
	In the second case, $\mathbf{K'}(M)$ is asymmetric, 
	so the two labels of $\vk'$ are on distinct vertices.
	Construct $\vk$ by again creating two copies $\vk'$, then merging $\ell_1$ of the first copy with $\ell_2$ of the second copy and
	vice-versa. Now $\mathbf{K}(M) \in 
	\Sym{q}{}(\mathbb{R}_{\geq 0})$
	such that $\mathbf{K}(M)_{ij} = \mathbf{K'}(M)_{ij} \cdot
	\mathbf{K'}(M)_{ji}$ for all $i, j \in [q]$.
	So,
	$\mathbf{K}(M) \in \PlEdge(M)$.
	Again, $\mathbf{K}(M)_{ii} = 
	(\mathbf{K'}(M)_{ii})^2
	\neq (\mathbf{K'}(M)_{jj})^2 =
	\mathbf{K}(M)_{jj}$.
\end{proof}

For \cref{proposition: autOrbits}, the classical version of \cref{theorem: qutOrbits}, we apply
the following classical version of \cref{theorem: duality}, a special case of \cite[Theorem 5]{young2022equality} (see also \cite[Lemma 2.5]{lovasz2006rank})\footnote{Lovász \cite{lovasz2006rank} considers
	only those 2-labeled graphs with labels on distinct vertices. As a result, 
	\cite[Lemma 2.5]{lovasz2006rank} only applies to \emph{twin-free} $M$.}.
Here, we view $\aut(M)$ as the group of permumtation matrices representing the permumtations in $\aut(M)$.
\begin{theorem}[\cite{young2022equality}]\label{theorem: classical_duality}
	Let $K \in \Mat{q}{}(\mathbb{R})$. We have $PK = KP$ for every
	$P \in \aut(M)$ if and only if $K \in \text{span}_{\mathbb{C}}(\Gadget(M))$.
\end{theorem}
Now, following the proof of \cite[Lemma 5]{young2022equality}, we obtain a classical version of \cref{lemma: diagonal_orbit}.
\begin{lemma}\label{lemma: classical_diagonal_orbit}
	Let $M \in \Sym{q}{}(\mathbb{R})$. Then
	$\mathbf{K}(M)_{ii} = \mathbf{K}(M)_{jj}$ 
	for every $\mathbf{K}(M) \in \Gadget(M)$ 
	if and only if $i$ and $j$ are in the same
	orbit of $\aut(M)$.
\end{lemma}
\begin{proof}
	$(\Leftarrow)$ is immediate. Conversely, suppose $i$ and $j$ are in
	different orbits of $\aut(M)$, and define the diagonal matrix 
	$L \in \Mat{q}{}(\mathbb{R})$ by:
	$$L_{xx} = \begin{cases} 1 & x \sim i \\ 0 & \text{otherwise}\end{cases}.$$ 
	By definition, $L$ is constant on the orbits of $\aut(M)$,
	so $PL = LP$ for every $P \in \aut(M)$. Thus \cref{theorem: classical_duality} gives
	$L \in \text{span}_{\mathbb{C}}(\Gadget(M))$ so, as in the proof of
	\cref{theorem: qutOrbits}, one of the terms 
	$\mathbf{K}(M) \in \Gadget(M)$ of $L$ must satisfy
	$\mathbf{K}(M)_{ii} \neq \mathbf{K}(M)_{jj}$.
\end{proof}
Now we obtain the following result, asserted in \cref{sec: quantumBody}, by \cref{lemma: classical_diagonal_orbit} via the
same argument by which \cref{theorem: qutOrbits} follows from \cref{lemma: diagonal_orbit}.
\autOrbits*

\subsection{Diagonal distinctness from pairwise distinguishability}
In this subsection, we shall prove
\cref{theorem: fullGadgetSeparationNonNegative}.
Before proving the dichotomy for all
$M \in \Sym{q}{}(\mathbb{R}_{\geq 0})$,
let us prove the dichotomy for
$M \in \Sym{q}{}(\mathbb{R}_{> 0})$.

\begin{lemma}\label{lemma: makeDiagonalDistinct}
	Let $M \in \Sym{q}{}(\mathbb{R}_{\geq 0})$.
	Assume that for all $x < y \in [q]$,
	there is a $\mathbf{K_{xy}}(M) \in \PlEdge(M)$ such that $\mathbf{K_{xy}}(M)_{xx} \neq \mathbf{K_{xy}}(M)_{yy}$ and $\mathbf{K_{xy}}(M)_{ii} > 0$
	for every $i \in [q]$.
	Then there is a diagonal distinct $\mathbf{K}(M) 
	\in \PlEdge(M)$.
\end{lemma}
\begin{proof}
	Given $\mathbf{z} = (z_{xy})_{x < y \in [q]}
	\in (\mathbb{Z}_{> 0})^{\nicefrac{q(q-1)}{2}}$, we construct
	the edge gadget $\mathbf{K_{z}}$ from a disjoint union of $z_{xy}$ copies of each
	$\mathbf{K_{xy}}$ by merging all of the vertices 
	labeled $\ell_{1}$, then merging all of the vertices
	labeled $\ell_{2}$. 
	Equivalently, construct $\mathbf{K_{z}}$ from the thickening gadget
	$\mathbf{T_{n}}$ for $n = \sum_{x < y} z_{xy}$ 
	by replacing each edge with the 
	appropriate $\mathbf{K_{xy}}$. 
	Clearly, $\mathbf{K_{z}}$ is a planar gadget, and
	$\mathbf{K_{z}}(M) \in \PlEdge(M)$ for all
	$\mathbf{z}$.    
	It follows that for all $i, j \in [q]$,
	$$\mathbf{K_{z}}(M)_{ij} = \prod_{x < y \in [q]}
	\Big((\mathbf{K_{xy}}(M))_{ij}\Big)^{z_{ij}}.$$
	
	Every $\mathbf{K_{xy}}(M)$ has 
	positive diagonal entries, so for every
	$\mathbf{z} \in 
	(\mathbb{Z}_{> 0})^{\nicefrac{q(q-1)}{2}}$,
	$\mathbf{K_{z}}(M)_{ii} > 0$ for all $i \in [q]$.
	
	We will now show that there exists some
	$\mathbf{z}^{*}$ such that 
	$\mathbf{K_{z^{*}}}(M)$ is diagonal distinct.
	Note that pairwise distinctness
	of $\mathbf{K_{z}}(M)_{ii}$ for $i \in [q]$ is equivalent to pairwise distinctness
	of the values $\log (\mathbf{K_{z}}(M)_{ii}) 
	= \sum_{x < y} z_{xy} \cdot
	\log(\mathbf{K_{xy}}(M)_{ii})$.
	This is equivalent to a choice of $\mathbf{z}^{*}$ 
	on which the
	polynomial
	\begin{align*}
		\zeta(\mathbf{z})
		&= \prod_{i < j \in [q]} \left(
		\sum_{x < y \in [q]} z_{xy} \cdot 
		\log(\mathbf{K_{xy}}(M)_{ii})
		- \sum_{x < y \in [q]} z_{xy} \cdot 
		\log(\mathbf{K_{xy}}(M)_{jj})
		\right)\\
		&= \prod_{i < j \in [q]}
		\sum_{x < y \in [q]} z_{xy} \cdot 
		\big(\log(\mathbf{K_{xy}}(M)_{ii}) - 
		\log(\mathbf{K_{xy}}(M)_{jj})\big)
	\end{align*}
	is nonzero.
	For any $i < j\in [q]$,
	by our assumption about the gadget
	$\mathbf{K}_{ij}(M)$,
	we see that the coefficient of $z_{ij}$ in
	$$\sum_{x < y \in [q]}z_{xy} \cdot \big(
	\log(\mathbf{K_{xy}}(M)_{ii}) - 
	\log(\mathbf{K_{xy}}(M)_{jj})\big)$$
	is: 
	$\log(\mathbf{K}_{ij}(M)_{ii}) - 
	\log(\mathbf{K}_{ij}(M)_{jj}) \neq 0$.
	Thus each factor of $\zeta(\mathbf{z})$ is nonzero,
	and so $\zeta(\mathbf{z})$ is not the zero polynomial.
	Since $(\mathbb{Z}_{> 0})$ is infinite,
	and $\zeta$ is non-zero, it cannot vanish on all
	of $(\mathbb{Z}_{> 0})^{\nicefrac{q(q-1)}{2}}$.
	Hence there exists some
	$\mathbf{z}^{*} \in 
	(\mathbb{Z}_{> 0})^{\nicefrac{q(q-1)}{2}}$,
	such that $\mathbf{K_{z^{*}}}(M) \in \PlEdge(M)$ 
	is diagonal distinct.
\end{proof}

We can now also prove the following:

\begin{theorem}\label{theorem: fullGadgetSeparation}
	Let $M \in \Sym{q}{}(\mathbb{R}_{> 0})$.
	Assume that for all $x < y \in [q]$, there exists
	some $\mathbf{K_{xy}}(M) \in \PlEdge(M)$ such that
	$\mathbf{K_{xy}}(M)_{xx} \neq \mathbf{K_{xy}}(M)_{yy}$.
	Then, $\PlGH(M)$ is polynomial-time tractable
	if $M$ is rank $1$,
	and in all other cases, $\PlGH(M)$ is $\#$P-hard.   
\end{theorem}
\begin{proof}
	When ${\rm rank}(M) = 1$, we already know that
	$\PlGH(M)$ is 
	tractable from the proof of
	\cref{theorem: positiveDichotomy}
	in \autoref{appendix: diagonalDistinctDichotomy}.
	So, we may assume that
	${\rm rank}(M) \geq 2$.
	
	Note that since $M \in \Sym{q}{}(\mathbb{R}_{> 0})$,
	each 
	$\mathbf{K_{xy}}(M)_{ii} > 0$ for all $i \in [q]$. Therefore, by
	\cref{lemma: makeDiagonalDistinct},
	there exists some $\mathbf{K}(M)
	\in \PlEdge(M)$, such that
	$\mathbf{K}(M)$ is diagonal distinct.
	Given $n \in \mathbb{Z}_{\geq 0}$, we will
	now construct a gadget $\mathbf{K_{n}}$
	by starting with the gadget $\mathbf{T_{(n+1)}}$,
	and then replacing all but one of the parallel edges
	with a copy of $\mathbf{K}$.
	Clearly,  $\mathbf{K_{n}}$ is planar, and
	for all $i, j \in [q]$, we see that
	$$\mathbf{K_{n}}(M)_{ij} = M_{ij} \cdot 
	\big(\mathbf{K}(M)_{ij} \big)^{n}.$$
	
	Given any planar graph
	$G = (V, E)$,
	by computing $Z_{\mathbf{K_{n}}(M)}(G)
	= Z_{M}(\mathbf{K_{n}}G)$ for all $n \in [|E|]$,
	we obtain a Vandermonde system of linear equations of polynomial size.
	We can combine all equal terms $\mathbf{K}(M)_{ij} $ to get
	a full ranked Vandermonde system.
	Solving this system of linear equations shows that
	$\PlGH(\mathcal{K}_{M}(\theta)) \leq \PlGH(M)$
	for all $\theta \in \mathbb{R}$, where
	$\mathcal{K}_{M}(\theta) \in \Sym{q}{}(\mathbb{R}_{> 0})$
	is defined as:
	$\mathcal{K}_{M}(\theta)_{ij} = M_{ij} \cdot 
	(\mathbf{K}(M)_{ij})^{\theta}$.
	
	We note that $\mathcal{K}_{M}(0) = M$, and from our
	assumption about $M$, we note that
	${\rm rank}(M) \geq 2$.
	Since the entries of $\mathcal{K}_{M}(\theta)$ are
	continuous as a function of $\theta$, and
	eigenvalues are continuous as a function of the
	entries of the matrix, we see that there
	exists an open neighborhood
	$\mathcal{O} \subset \mathbb{R}$ around $0$,
	such that ${\rm rank}(\mathcal{K}_{M}(\theta)) \geq 2$
	for all $\theta \in \mathcal{O}$.
	Now, consider some $i \neq j \in [q]$.
	From \cref{lemma: makeDiagonalDistinct},
	we may assume without loss of generality that
	$\mathbf{K}(M)_{ii} < \mathbf{K}(M)_{jj}$.
	Then,
	$$\lim\limits_{\theta \rightarrow \infty}
	\left(\frac{\mathcal{K}_{M}(\theta)_{ii}}
	{\mathcal{K}_{M}(\theta)_{jj}} \right)
	= \left(\frac{M_{ii}}{M_{jj}} \right) \cdot
	\lim\limits_{\theta \rightarrow \infty} \left(
	\frac{\mathbf{K}(M)_{ii}}{\mathbf{K}(M)_{jj}}\right)^{\theta}
	= 0.$$
	Therefore, we see that for large enough $\theta$,
	$\phi_{\rm{diag}}(\mathcal{K}_{M}(\theta)) \neq 0$
	(where $\phi_{\rm{diag}}$ is as defined in
	\cref{defn:phi-psi-function-for-diagonal-distinct}).
	But then, \cref{lemma: analyticGadgets} lets us find
	some $\theta^{*} \in \mathcal{O}$ such that
	$N = \mathcal{K}_{M}(\theta^{*}) \in 
	\Sym{q}{}(\mathbb{R}_{> 0})$
	satisfies: $\phi_{\rm{diag}}(N) \neq 0$,
	and $\PlGH(N) \leq \PlGH(M)$.
	From our choice of $\theta^{*} \in \mathcal{O}$,
	we also see that ${\rm rank}(N) \geq 2$.
	So, \cref{theorem: positiveDichotomy}
	implies that $\PlGH(N) \leq \PlGH(M)$ is
	$\#$P-hard.
\end{proof}

\begin{remark*}
	The main challenge with extending
	the dichotomy to all
	$M \in \Sym{q}{}(\mathbb{R}_{\geq 0})$
	is that
	\cref{lemma: makeDiagonalDistinct} requires each $\mathbf{K_{xy}}(M)$ to have
	strictly positive diagonal entries, which may not hold if $M$ has 0 entries.
	We need the $\mathbf{K}$ given by \cref{lemma: makeDiagonalDistinct} to have
	strictly positive diagonal entries because in the key step where
	we construct $\mathbf{K_{z}}(M)$,
	if any of the gadgets $\mathbf{K_{xy}}(M)$
	has a zero diagonal entry: $\mathbf{K_{xy}}(M)_{ii} = 0$
	for some $i \in [q]$, that forces
	$\mathbf{K_{z}}(M)_{ii} = 0$.
	So, we need to be careful to avoid that situation.
\end{remark*}

\begin{notation}
	For $M \in \Sym{q}{}(\mathbb{R}_{\geq 0})$,
	we let $\Gamma_{M}
	= ([q], E(\Gamma_{M}))$ denote the
	\emph{underlying graph} of $M$,
	namely $(i, j) \in E(\Gamma_{M})$ if and only if
	$M_{ij} > 0$.
\end{notation}

\begin{lemma}\label{lemma: makeDiagonalDistinctBipartite}
	Let $M \in \Sym{q}{}(\mathbb{R}_{\geq 0})$
	(for $q \geq 2$),
	such that the underlying graph $\Gamma_{M}$ 
	is \emph{connected} and bipartite.
	Assume that for all $x < y \in [q]$, there
	is a $\mathbf{K_{xy}}(M) \in \PlEdge(M)$
	such that $\mathbf{K_{xy}}(M)_{xx} \neq 
	\mathbf{K_{xy}}(M)_{yy}$.
	Then, for all $x < y \in [q]$,
	$\mathbf{K_{xy}} \in \PlEdge(M)$
	also satisfies the condition:
	$\mathbf{K_{xy}}(M)_{ii} \neq 0$ for all $i \in [q]$.
\end{lemma}
\begin{proof}
	It is well-known that,
	if there is at least one homomorphism from graph $H$ to bipartite graph
	$\Gamma$, then $H$ must be bipartite.
	Similarly, since $\mathbf{K_{xy}}(M)_{xx} \neq
	\mathbf{K_{xy}}(M)_{yy}$, 
	one of $\mathbf{K_{xy}}(M)_{xx}$ or 
	$\mathbf{K_{xy}}(M)_{yy}$ is nonzero, 
	so $\mathbf{K_{xy}}$ must be bipartite.
	Now, assume that the two labeled vertices 
	of $\mathbf{K_{xy}}$ are
	$\ell_{1}$ and $\ell_{2}$.
	First, we assume that
	they belong to the same connected component of
	$\mathbf{K_{xy}}$, but
	belong to  different partitions.
	Since one of $\mathbf{K_{xy}}(M)_{xx}$ or 
	$\mathbf{K_{xy}}(M)_{yy}$ is nonzero,
	this implies that the odd-length
	path from $\ell_{1}$ to $\ell_{2}$ in $\mathbf{K_{xy}}$
	had to be mapped to an even length path
	between $x$ and $x$, or $y$ and $y$ in
	$\Gamma_{M}$.
	This contradiction therefore implies that
	if $\ell_{1}$ and $\ell_{2}$ belong to the 
	same connected component of 
	$\mathbf{K_{xy}}$, then
	$\ell_{1}$ and $\ell_{2}$ must be in the same
	partition of $\mathbf{K_{xy}}$.
	
	Since $\Gamma_{M}$ is connected,
	we note that every $i \in [q]$ is adjacent to some other
	$j \in [q]$ in $\Gamma_{M}$.
	We note that if $\ell_{1}$ and $\ell_{2}$
	belong to different connected components
	of $\mathbf{K_{xy}}$, we can trivially
	define a bipartition of $\mathbf{K_{xy}}$
	such that $\ell_{1}$ and $\ell_{2}$ belong
	to the same partition.
	We have also seen that if $\ell_1$ and $\ell_2$ belong to the 
	same connected component, then they must
	belong to the same partition.
	So, in any case, there exists a bipartition
	of $\mathbf{K_{xy}}$ such that
	$\ell_{1}$, and $\ell_{2}$ belong
	to the same partition.
	Now, we can define a map
	$\mathbf{K_{xy}}$ to $\Gamma_{M}$
	that sends the partition containing
	$\ell_{1}$ and $\ell_{2}$
	to $i$ and the other partition to $j$.
	We see that this map is a homomorphism, so
	$\mathbf{K_{xy}}(M)_{ii} 
	\geq (M_{ij})^{|E(\mathbf{K_{xy}})|} > 0$, as $M_{ij} > 0$. 
\end{proof}
For the non-bipartite case, we apply the ``connector'' idea of Lovász and
Szegedy \cite[Theorem 1.4]{lovasz2009contractors} and Dvo{\v{r}}{\'a}k
\cite[Lemma 11]{dvorak2010recognizing}: by replacing every edge of
$\mathbf{K_{xy}}$ by a linear combination of sufficiently long paths, we
may assume the underlying graph admits homomorphisms to $M$.
\begin{lemma}\label{lemma: makeDiagonalDistinctNonBipartite}
	Let $M \in \Sym{q}{}(\mathbb{R}_{\geq 0})$
	(for $q \geq 2$),
	such that the underlying graph $\Gamma_{M}$ 
	is \emph{connected} but not bipartite.
	Assume that for all $x < y \in [q]$, there
	is a $\mathbf{K_{xy}}(M) \in \PlEdge(M)$
	such that $\mathbf{K_{xy}}(M)_{xx} \neq 
	\mathbf{K_{xy}}(M)_{yy}$.
	Then, there exist $\mathbf{J_{xy}} \in \PlEdge(M)$
	for all $x < y \in [q]$, such that
	$\mathbf{J_{xy}}(M)_{xx} \neq 
	\mathbf{J_{xy}}(M)_{yy}$, and
	$\mathbf{J_{xy}}(M)_{ii} \neq 0$ for all $i \in [q]$.
\end{lemma}
\begin{proof}
	By nonnegativity, $(M^k)_{ii} \neq 0$ if and only if 
	there is a length-$k$ walk in $\Gamma_{M}$ 
	from vertex $i$ back to itself. 
	Since $\Gamma_{M}$ is connected and not bipartite,
	there exists a smallest $m \geq 2$ such that, 
	for every $k \geq m$, every vertex $i$ has
	a length-$k$ walk to itself, so
	$(M^k)_{ii} \neq 0$ for every $i \in [q]$.
	We first show that
	\begin{equation}\label{eq:nspan}
		M \in \text{span}_{\mathbb{R}}(M^{m}, \ldots M^{m+q}).
	\end{equation} 
	Write $M = H \diag(\lambda_1,\ldots,\lambda_q) H^{\tt{T}}$ 
	for orthogonal $H$, so
	$M^k = H \diag(\lambda_1^k,\ldots,\lambda_q^k) H^{\tt{T}}$, 
	and hence \cref{eq:nspan} is
	equivalent to $\diag(\lambda_1,\ldots,\lambda_q) \in \text{span}_{\mathbb{R}}(
	\diag(\lambda_1^m,\ldots,\lambda_q^m), \ldots \diag(\lambda_1^{m+q},\ldots,\lambda_q^{m+q}))$.
	Without loss of generality, 
	let $\lambda_1,\ldots,\lambda_s$, with $s \leq q$, 
	be the unique non-zero eigenvalues of $M$. 
	Then
	\cref{eq:nspan} is equivalent to 
	\begin{equation}\label{eq:nspan2}
		\diag(\lambda_1,\ldots,\lambda_s) \in \text{span}_{\mathbb{R}}(
		\diag(\lambda_1^m,\ldots,\lambda_s^m), \ldots \diag(\lambda_1^{m+q},\ldots,\lambda_s^{m+q})).
	\end{equation}
	Upon multiplying its $i$th column by the nonzero scalar $\lambda_i^{-m}$ for
	every $i \in [s]$, the matrix
	\[
	L = 
	\begin{bmatrix}
		\lambda_1^m & \lambda_2^m & \ldots & \lambda_s^m \\
		\lambda_1^{m+1} & \lambda_2^{m+1} & \ldots & \lambda_s^{m+1} \\
		\vdots & \vdots & \ldots & \vdots \\
		\lambda_1^{m+q} & \lambda_2^{m+q} & \ldots & \lambda_s^{m+q} \\
	\end{bmatrix}
	\]
	becomes a Vandermonde matrix with $q$ rows and $s \leq q$ columns. Since
	$\lambda_1,\ldots,\lambda_s$ are distinct, $L$ has rank $s$, 
	so the rows of $L$ span
	$\mathbb{R}^s$. This proves \cref{eq:nspan2}, hence also \cref{eq:nspan}.
	
	Now write $M = \sum_{i=0}^q c_i M^{m+i}$ for $c_0,\ldots,c_q \in \mathbb{R}_{\neq 0}$. 
	Fix $x < y \in [q]$.
	For every $i,j \in [q]$, substituting into \cref{equation: signature} gives
	\newcommand\numberthis{\addtocounter{equation}{1}\tag{\theequation}}
	\begin{align*}
		\mathbf{K_{xy}}(M)_{ij} 
		&= \sum_{\substack{\tau: V(\mathbf{K_{xy}}) \to [q] \\ \tau(\ell_1) = i,
				\tau(\ell_2) = j}} \prod_{(u,v) \in E(\mathbf{K_{xy}})} 
		\sum_{i=0}^q c_i (M^{m+i})_{\tau(u)\tau(v)} \\
		&= \sum_{\psi: E(\mathbf{K_{xy}}) \to \{m,\ldots,m+q\}} 
		\left(\prod_{e \in E(\mathbf{K_{xy}})} c_{\psi(e)-m}\right)
		\sum_{\substack{\tau: V(\mathbf{K_{xy}}) \to [q] \\ \tau(\ell_1) = i,
				\tau(\ell_2) = j}} \prod_{(u,v) \in E(\mathbf{K_{xy}})} 
		(M^{\psi((u,v))})_{\tau(u)\tau(v)}\\
		&= \sum_{\psi: E(\mathbf{K_{xy}}) \to \{m,\ldots,m+q\}} 
		(c_\psi)
		\mathbf{K_{xy}^\psi}(M)_{ij}, \numberthis\label{eq:kpsi_sum}
	\end{align*}
	where $c_\psi := \prod_{e \in E(\mathbf{K_{xy}})} c_{\psi(e)-m} \neq 0$ and
	$\mathbf{K_{xy}^\psi} \in \PlGadget(M)$ is constructed by replacing each
	$(u,v) \in E(\mathbf{K_{xy}})$ with the stretching gadget $\mathbf{S}_{\psi((u,v))}$, so
	that, as in \cref{eq:signature_matrix},
	\begin{equation}
		\label{eq:quantum_gadget}
		\mathbf{K_{xy}^\psi}(M)_{ij} =
		\sum_{\substack{\tau: V(\mathbf{K_{xy}}) \to [q] \\ \tau(\ell_1) = i,
				\tau(\ell_2) = j}} \prod_{(u,v) \in E(\mathbf{K_{xy}})} 
		(M^{\psi((u,v))})_{\tau(u)\tau(v)}.\\
	\end{equation}
	For every $\psi: E(\mathbf{K_{xy}}) \to \{m,\ldots,m+q\}$ and $(u,v) \in E(\mathbf{K_{xy}})$,
	the diagonal entries of $M^{\psi((u,v))}$ are nonzero by definition of $m$.
	It then follows from nonnegativity that $\mathbf{K_{xy}^\psi}(M)_{ii} > 0$ for every
	$i \in [q]$,
	as the assignment $\tau$ mapping every vertex of $\mathbf{K_{xy}}$ to $i$ yields a nonzero 
	term in \cref{eq:quantum_gadget}. Furthermore, it follows from \cref{eq:kpsi_sum} and the
	assumption $\mathbf{K_{xy}}(M)_{xx} \neq \mathbf{K_{xy}}(M)_{yy}$ that there is a 
	$\psi: E(\mathbf{K_{xy}}) \to \{m,\ldots,m+q\}$ for which 
	$\mathbf{K_{xy}^\psi}(M)_{xx} \neq \mathbf{K_{xy}^\psi}(M)_{yy}$.
	This $\mathbf{K_{xy}^\psi}$ has the desired properties. Finally,
	as in the proof of \autoref{theorem: qutOrbits}, we may replace
	$\mathbf{K_{xy}^\psi} \in \PlGadget(M)$ with a
	$\mathbf{J_{xy}} \in \PlEdge(M)$ with the desired properties 
	(this involves squaring the diagonal
	entries of $\mathbf{J_{xy}}(M)$, which does not change whether these entries
	are nonzero).
\end{proof}

Using \cref{lemma: makeDiagonalDistinctBipartite,lemma: makeDiagonalDistinctNonBipartite}
to meet the conditions of \cref{lemma: makeDiagonalDistinct} for $M \in \Sym{q}{}(\mathbb{R}_{\geq 0})$
with a connected $\Gamma_{M}$,
we can now construct
a $\mathbf{K}(M) \in \PlEdge(M)$ such that
$\mathbf{K}(M)$ is diagonal distinct.

\begin{corollary}\label{corollary: makeDiagonalDistinctConnected}
	Let $M \in \Sym{q}{}(\mathbb{R}_{\geq 0})$,
	such that the underlying graph $\Gamma_{M}$ 
	is connected.
	Assume that for all $x < y \in [q]$, there
	is a $\mathbf{K_{xy}}(M) \in \PlEdge(M)$
	such that $\mathbf{K_{xy}}(M)_{xx} \neq 
	\mathbf{K_{xy}}(M)_{yy}$.
	Then there is a diagonal distinct 
	$\mathbf{K}(M) \in \PlEdge(M)$.
\end{corollary}

\cref{corollary: distinct} now follows as an immediate
corollary of \cref{corollary: makeDiagonalDistinctConnected} 
and \cref{theorem: qutOrbits}.
\distinct*

It is natural to ask whether Lemmas \ref{lemma: makeDiagonalDistinctBipartite} and \ref{lemma: makeDiagonalDistinctNonBipartite} hold for all $M \in \Sym{q}{}(\mathbb{R}_{\geq 0})$, even those for which
$\Gamma_M$ is disconnected. The proofs still go through when
all components of $\Gamma_M$ are bipartite or all components of
$\Gamma_M$ are not bipartite. However, we show using techniques of
Dvo{\v{r}}{\'a}k \cite{dvorak2010recognizing} that when $\Gamma_M$
contains a mix of bipartite and non-bipartite components, the existence
of a $\mathbf{K_{xy}}$ separating $x$ and $y$ does not guarantee the
existence of the desired $\mathbf{J_{xy}}$ with nonzero diagonal
entries (although the constructed $M$ have nontrivial
automorphisms, hence do not satisfy the lemmas' hypothesis that all
pairs $x < y$ can be separated).

Given a graph $H$ with adjacency matrix $M$, let $H \times K_2$ be the bipartite double cover
of $H$, with adjacency matrix $M \otimes \left[\begin{smallmatrix}
	0 & 1 \\ 1 & 0
\end{smallmatrix}\right]$. 
Equivalently, $H \times K_2$ has vertex set $V(H) \times \{0,1\}$ 
and $((h,b),(h',b')) \in E(H \times K_2)$ if and only if $(h,h') \in E(H)$ and $b \neq b'$
($H \times K_2$ is bipartite between $\{(h,0) \mid h \in V(H)\}$ and
$\{(h,1) \mid h \in V(H)\}$). 
Let $H_1 \cup H_2$ be the disjoint union of graphs $H_1$ and $H_2$.
\begin{proposition}
	For any connected non-bipartite graph $H$, let $H_1 := H \cup H$ and $H_2 := H \times K_2$,
	and let
	$M \in \Sym{V(H_1) \sqcup V(H_2)}{}(\mathbb{R}_{\geq 0})$ be the 
	adjacency matrix of 
	$H_1 \cup H_2$. Then for every $x \in V(H_1)$, there is a
	$y \in V(H_2)$ such that 
	\begin{itemize}
		\item there exists a $\vk(M) \in \PlEdge(M)$ such that
		$\vk(M)_{xx} \neq \vk(M)_{yy}$, but
		\item every $\mathbf{J}(M) \in \Gadget(M)$ such that $\mathbf{J}(M)_{xx} \neq \mathbf{J}(M)_{yy}$
		also satisfies $\mathbf{J}(M)_{yy} = 0$.
	\end{itemize}
\end{proposition}
\begin{proof}
	For every vertex $x$ of $H_1 = H \cup H$, there is an odd $r$ such that there 
	is a length-$r$ closed walk from $x$ to itself in $H$.
	So, if $\mathbf{K}$ is the edge gadget consisting of two copies of the cycle $C_r$, each with one labeled
	vertex, then $\mathbf{K}(M)_{xx} \neq 0$. Additionally,
	$\mathbf{K}(M)_{yy} = 0$ for every $y \in V(H_2)$, as $H_2 = H \times K_2$ is bipartite.
	This proves the first claim. 
	
	For the second claim, we may assume that $\mathbf{J} \in \Gadget(M)$ contains no 
	connected components without
	a labeled vertex, as such components simply multiply $\mathbf{J}(M)_{xx}$ and $\mathbf{J}(M)_{yy}$
	by the same constant. Since $H_2$ is not bipartite, any non-bipartite $\mathbf{J}$ will
	satisfy $\mathbf{J}(M)_{yy} = 0$ for every $y \in V(H_2)$. So assume $\mathbf{J}$ is bipartite.
	
	Dvo{\v{r}}{\'a}k \cite[Theorem 13]{dvorak2010recognizing} shows that there is an
	isomorphism $f: H_1 \times K_2 \to H_2 \times K_2$ preserving the second coordinate
	-- that is, for every
	$(h,b) \in V(H_1 \times K_2)$, we have $f((h,b)) = ((h',b))$ for some $h' \in V(H_2)$ --
	and, for any bipartite $\mathbf{J}$, uses $f$ to construct a family of bijections between 
	the sets $\text{Hom}(\mathbf{J},H_1)$ and
	$\text{Hom}(\mathbf{J},H_2)$ of homomorphisms from $\mathbf{J}$ to $H_1$ and $H_2$,
	respectively (ignoring for now the vertex labels of $\mathbf{J}$). 
	Each $c: V(\mathbf{J}) \to \{0,1\}$ defining a valid bipartition of $\mathbf{J}$
	yields a bijection as follows. For $\beta \in \{1,2\}$, define 
	$\pi_1^{H_\beta \times K_2}: V(H_\beta \times K_2) \to V(H_\beta)$ by
	$\pi_1(h,b) = h$. Then map $g_1 \in \text{Hom}(\mathbf{J},H_1)$ to $g_2 \in \text{Hom}(\mathbf{J},H_2)$
	defined by
	\begin{equation}\label{eq:g1tog2}
		g_2(\ell) = \pi_1^{H_2 \times K_2}(f(g_1(\ell),c(\ell)))
	\end{equation}
	for $\ell \in V(\mathbf{J})$. The inverse of this map is given by
	\begin{equation}\label{eq:g2tog1}
		g_1(\ell) = \pi_1^{H_1 \times K_2}(f^{-1}(g_2(\ell),c(\ell)))
	\end{equation}
	for $\ell \in V(\mathbf{J})$. Let $\ell_1,\ell_2 \in V(\mathbf{J})$ be the labeled vertices. If
	there is an odd-length path in $\mathbf{J}$ between $\ell_1$ and $\ell_2$, then again
	$\mathbf{J}(M)_{yy} = 0$ for every $y \in V(K_2)$ and we are done. 
	Assume otherwise, so we can choose $c$ such that $c(\ell_1) = c(\ell_2) = 0$. Let
	$y := \pi_1^{H_2 \times K_2}(f(x,0)) \in V(H_2)$. 
	If $g_1(\ell_1) = g_1(\ell_2) = x$, then, by \cref{eq:g1tog2}, $g_2(\ell_1) = g_2(\ell_2)
	= y$. Conversely, if $g_2(\ell_1) = g_2(\ell_2) = y$, then, by \cref{eq:g2tog1},
	\[
	g_1(\ell_1) = g_1(\ell_2) = \pi_1^{H_1 \times K_2}(f^{-1}(y,0))
	= \pi_1^{H_1 \times K_2}(x,0) = x.
	\]
	Therefore the bijection between $\text{Hom}(\mathbf{J},H_1)$ and $\text{Hom}(\mathbf{J},H_2)$
	restricts to a bijection between those homomorphisms in $\text{Hom}(\mathbf{J},H_1)$ sending
	$\ell_1$ and $\ell_2$ to $x$ and those homomorphisms in $\text{Hom}(\mathbf{J},H_2)$ sending
	$\ell_1$ and $\ell_2$ to $y$. Since $\mathbf{J}$ contains no connected components without a
	labeled vertex, any homomorphism from $\mathbf{J}$ to $H_1 \cup H_2$ sending $\ell_1$ and
	$\ell_2$ to $x$ (resp. $y$) is a homomorphism from $\mathbf{J}$ to $H_1$ (resp. $H_2$).
	Hence $\mathbf{J}(M)_{xx} = \mathbf{J}(M)_{yy}$.
\end{proof}

\cref{theorem: fullGadgetSeparationNonNegative}, however, only
requires \cref{corollary: makeDiagonalDistinctConnected} as stated,
for $\Gamma_M$ connected.

\fullGadgetSeparationNonNegative*
\begin{proof}
	We note that $\Gamma_{M}$ can be decomposed
	into its various connected components, such that
	(after a renaming of $[q]$),
	$M = M_{1} \oplus \dots \oplus M_{r}$,
	such that each $M_{i} \in \Sym{q_{i}}{}(\mathbb{R}_{\geq 0})$
	corresponds to a connected component of
	$\Gamma_{M}$.
	
	We have already seen in the proof
	of \cref{theorem: positiveDichotomy} that
	when $M_{i}$ are rank $0$ or rank $1$, then
	$\PlGH(M_{i})$ is polynomial-time tractable.
	Now, let us assume that $M_{i}$ is
	bipartite, but rank $2$.
	In this case, it is already known
	(\cite{bulatov2005complexity} see Theorem 6)
	that $\GH(M_{i})$ (and therefore, also
	$\PlGH(M_{i})$)) is polynomial-time tractable.
	Now, \cref{theorem: domainSeparable}
	implies that if every $M_{i}$ is rank $0$
	or rank $1$ or bipartite rank $2$,
	then $\PlGH(M)$ is polynomial-time tractable.
	
	Let us now assume otherwise.
	Without loss of generality, we may assume that
	either ${\rm rank}(M_{1}) \geq 2$ and $M_{1}$
	is not bipartite, or that
	${\rm rank}(M_{1}) \geq 3$ and $M_1$ is bipartite.
	In either case, we can use
	\cref{corollary: makeDiagonalDistinctConnected}
	to find some gadget $\mathbf{K}$ such that
	$\mathbf{K}(M_{1})$ is diagonal distinct.
	When $M_{1}$ is not bipartite,
	using the same techniques as in
	\cref{theorem: fullGadgetSeparation},
	we can construct $N = \mathcal{K}_{M}(\theta^{*})$
	that satisfies:
	$N = N_{1} \oplus \dots \oplus N_{r}$ for
	$N_{i} \in \Sym{q_{i}}{}(\mathbb{R}_{\geq 0})$,
	$\phi_{\rm{diag}}(N_{1}) \neq 0$,
	$\PlGH(N) \leq \PlGH(M)$, and
	${\rm rank}(N_{1}) = {\rm rank}(M_{1}) \geq 2$.
	From the \cref{theorem: nonNegativeDichotomy}
	dichotomy, it then follows that
	$\PlGH(N_{1})$ is $\#$P-hard, and 
	\cref{theorem: domainSeparable} then implies that
	$\PlGH(N) \leq \PlGH(M)$ is also
	$\#$P-hard.
	
	When $M_{1}$ is bipartite, we note that 
	the construction $\mathcal{K}_{M}$ that we
	used in \cref{theorem: fullGadgetSeparation}
	does not quite work, because the diagonal
	values of $M$ are all zeros.
	We note that since $\Gamma_{M_{1}}$ is
	connected, there exists some
	$m \geq 1$, such that $(M_{1})^{2m}
	= M_{1, 1} \oplus M_{1, 2}$, such that
	$M_{1, 1} \in \Sym{r_{1}}{}(\mathbb{R}_{> 0})$,
	and $M_{1, 2} \in \Sym{r_{2}}{}(\mathbb{R}_{> 0})$
	(where $r_{1} + r_{2} = q_{1}$).
	Moreover, since ${\rm rank}(M_{1}) \geq 3$,
	we see that at least one of $M_{1, 1}$ and
	$M_{1, 2}$ must have rank at least $2$.
	Without loss of generality,
	we may assume that ${\rm rank}(M_{1, 1}) \geq 2$.
	
	Now, we construct the gadget
	$\mathbf{K_{n}}$ by starting with $\mathbf{T_{(n+1)}}$,
	replacing all but one of the parallel edges
	with a copy of $\mathbf{K}$, and
	replacing the final edge with a copy of
	$\mathbf{S_{2m}}$.
	So,
	$$\mathbf{K_{n}}(M_{1})_{ij} = ((M_{1})^{2m})_{ij} \cdot
	(\mathbf{K}(M_{1})_{ij})^{n}.$$
	We can now interpolate with
	this gadget, just as in the proof of
	\cref{theorem: fullGadgetSeparation} to
	construct $\mathcal{K}_{M}(\theta)$
	such that $\mathcal{K}_{M}(\theta)_{ij}
	= (M^{2m})_{ij} \cdot (\mathbf{K}(M)_{ij})^{\theta}$,
	and $\PlGH(\mathcal{K}_{M}(\theta)) \leq 
	\PlGH(M)$ for all $\theta \in \mathbb{R}$.
	
	We can now find some
	$\theta^{*}$ close to $0$, such that
	$N = \mathcal{K}_{M}(\theta^{*})$ satisfies:
	$N = N_{1} \oplus \dots \oplus N_{r}$,
	$\phi_{\rm{diag}}(N_{1}) \neq 0$,
	$N_{1} = N_{1, 1} \oplus N_{1, 2}$,
	${\rm rank}(N_{1, 1}) = {\rm rank}(M_{1, 1})
	\geq 2$.
	So, once again,
	\cref{theorem: nonNegativeDichotomy,theorem: domainSeparable}
	imply that $\PlGH(N) \leq \PlGH(M)$
	is $\#$P-hard.
\end{proof}

\subsection{Undecidability}
Graphs $H$ and $H'$ with adjacency matrices $M$ and $M'$ are \emph{quantum isomorphic} if
there is a quantum permutation matrix $\uc$ such that $\uc M = M' \uc$. It is well-known
that connected $H$ and $H'$ are classically isomorphic if and only if some vertex of $H$
and some vertex of $H'$ are in the same orbit of $\aut(H \cup H') = \aut(M \oplus M')$, 
where $\cup$ denotes disjoint union. The quantum analogue is also true:
Connected graphs $H$ and $H'$ with adjacency matrices $M$ and $M'$ are quantum isomorphic if and only if some $v \in V(H)$
and $v' \in V(H')$ are in the same orbit of $\qut(M \oplus M')$
\cite[Theorem 4.5]{lupini2020nonlocal}. Furthermore, either a graph or its complement is connected, a connected graph cannot be quantum isomorphic to a disconnected graph \cite{lupini2020nonlocal}, and two graphs are quantum isomorphic if and only if their complements are \cite{atserias2019quantum}. The matrix $M \oplus M'$ (and the adjacency matrix of $\overline{H} \cup \overline{H'}$) are in 
$\Sym{|V(H)|+|V(H')|}{}(\mathbb{R}_{\geq 0})$, so we have the following.
\begin{theorem}[\cite{lupini2020nonlocal}]
	Quantum isomorphism of graphs reduces to the following problem:
	given $q \geq 1$, $M \in \Sym{q}{}(\mathbb{R}_{\geq 0})$ and $i,j \in [q]$, determine whether $i$ and $j$ are in the same orbit of $\qut(M)$.
	\label{theorem: orbit_isomorphism}
\end{theorem}
We now obtain a slightly stronger form of \cref{theorem: undec}, obtaining undecidability from even positive definite $M$ with positive entries.
\begin{theorem}\label{theorem: undecidable_appendix}
	The following problem is undecidable:
	given $M \in \Sym{q}{pd}(\mathbb{R}_{> 0})$ and $i,j \in [q]$, determine whether there 
	is a $\mathbf{K}(M) \in \PlEdge(M)$ 
	such that $\mathbf{K}(M)_{ii} \neq \mathbf{K}(M)_{jj}$.
\end{theorem}
\begin{proof}
	We reduce the problem in \cref{theorem: orbit_isomorphism} to this
	problem, and the result follows from the undecidability of quantum
	isomorphism \cite{atserias2019quantum,slofstra2019set}. Let $M \in \Sym{q}{}(\mathbb{R}_{\geq 0})$ and $i,j \in [q]$.
	By \autoref{proposition: ij}, $I,J \in \PlGadget(M)$. Let $\lambda$ be the smallest
	eigenvalue of $M+J$, and let $M' := M+J+(|\lambda|+1)I$, so 
	$M' \in \Sym{q}{pd}(\mathbb{R}_{> 0})$. By \cref{theorem: duality}, the
	fundamental representations $\uc$ of $\qut(M)$ and $\uc'$ of
	$\qut(M')$ both commute with $J$ and $(|\lambda|+1)I$.
	In general, if $\mathcal{V}$ is the fundamental representation of
	$\qut(A)$, then the space $\{K \mid \mathcal{V} K = K \mathcal{V}\}$,
	called the $(1,1)$-\emph{intertwiner space} of $\qut(A)$, is closed
	under linear combinations and generated by $A$ and matrices that,
	like $J$ and $I$, don't depend on $A$ \cite[Lemma 5]{cai2023planar}. Thus
	both $\uc$ and $\uc'$ commute with both $M$ and $M'$, and therefore,
	for all $K$, we have
	$\uc K = K \uc$ if and only if $\uc' K = K \uc'$. Now it follows
	from \cref{lemma: orbitals} that $\qut(M)$ and $\qut(M')$ have the
	same orbitals, so $i$ and $j$ are in the same orbit of $\qut(M)$ if
	and only if $i$ and $j$ are in the same orbit of $\qut(M')$.
	By \cref{theorem: qutOrbits}, this in turn is equivalent to $\mathbf{K}(M)_{ii} = \mathbf{K}(M)_{jj}$ 
	for every symmetric $\mathbf{K}(M) \in \PlEdge(M')$.
\end{proof}

We similarly obtain a stronger version of \cref{remark: separate_nonplanar}.
\begin{proposition}\label{proposition: separate_nonplanar}
	For infinitely many $q \geq 1$, there exist $M \in \Sym{q}{pd}(\mathbb{R}_{> 0})$ and $i,j \in [q]$ such that $K_{ii} = K_{jj}$ for every $\mathbf{K}(M) \in\PlEdge(M)$, 
	but there is a $\mathbf{K'}(M) \in \Edge(M)$ 
	such that $\mathbf{K'}(M)_{ii} \neq \mathbf{K'}(M)_{jj}$.
\end{proposition}
\begin{proof}
	Chan and Martin \cite[Theorem 4.3]{chan2024quantum} show that any two
	Hadamard graphs $H$ and $H'$ with the same number of vertices are quantum isomorphic. They also note that, if there exists a Hadamard
	matrix of order $n$, then there exist non-isomorphic Hadamard graphs on $8n$
	vertices. Since there exist Hadamard matrices of infinitely many distinct
	orders (e.g. Walsh matrices), there exist infinitely many pairs of graphs
	that are quantum isomorphic but not isomorphic. Let $H,H'$ be
	such a pair, with adjacency matrices $M,M'$. By the discussion before
	\cref{theorem: orbit_isomorphism}, we may
	assume $H$ and $H'$ are connected by taking complements if needed, so there exist $i \in V(H)$ and $j \in V(H')$ in the same orbit of $\qut(M \oplus M')$ but
	different orbits of $\aut(M \oplus M')$. As in the proof of
	\cref{theorem: undecidable_appendix}, let $\lambda$ be the smallest eigenvalue of $(M \oplus M') + J$. Then $M'' := (M \oplus M') + J + (|\lambda|+1)I \in \Sym{q}{pd}(\mathbb{R}_{> 0})$ and $\qut(M'')$ and
	$\qut(M \oplus M')$ have the same orbits. 
	By the complete symmetry among the $[q]$ vertices of the graphs
	with adjacency matrices $I$ or $J$, $\aut(M'')$ and $\aut(M \oplus M')$ also have the same orbits.
	Hence, by \cref{proposition: autOrbits} and
	\cref{theorem: qutOrbits}, $M''$ satisfies the desired property.
\end{proof}

%% file: refs.bib
@article{guo2020complexity,
	title={The complexity of planar Boolean \#CSP with complex weights},
	author={Guo, Heng and Williams, Tyson},
	journal={Journal of Computer and System Sciences},
	volume={107},
	pages={1--27},
	year={2020},
	publisher={Elsevier}
}

@inproceedings{dyer2000complexity,
	title={The complexity of counting graph homomorphisms (extended abstract)},
	author={Dyer, Martin and Greenhill, Catherine},
	editor={David B. Shmoys},
	booktitle={Proceedings of the Eleventh Annual {ACM-SIAM} Symposium on Discrete Algorithms, January 9-11, 2000, San Francisco, CA, {USA}},
	pages={246--255},
	publisher={{ACM/SIAM}},
	year={2000}
}

@article{bulatov2005complexity,
	title={The complexity of partition functions},
	author={Bulatov, Andrei and Grohe, Martin},
	journal={Theoretical Computer Science},
	volume={348},
	number={2-3},
	pages={148--186},
	year={2005},
	publisher={Elsevier}
}

@article{goldberg2010complexity,
	title={A complexity dichotomy for partition functions with mixed signs},
	author={Goldberg, Leslie Ann and Grohe, Martin and Jerrum, Mark and Thurley, Marc},
	journal={SIAM Journal on Computing},
	volume={39},
	number={7},
	pages={3336--3402},
	year={2010},
	publisher={SIAM}
}

@article{cai2013graph,
	title={Graph homomorphisms with complex values: A dichotomy theorem},
	author={Cai, Jin-Yi and Chen, Xi and Lu, Pinyan},
	journal={SIAM Journal on Computing},
	volume={42},
	number={3},
	pages={924--1029},
	year={2013},
	publisher={SIAM}
}

@article{vertigan2005computational,
	title={The computational complexity of Tutte invariants for planar graphs},
	author={Vertigan, Dirk},
	journal={SIAM Journal on Computing},
	volume={35},
	number={3},
	pages={690--712},
	year={2005},
	publisher={SIAM}
}

@article{govorov2020dichotomy,
	title={A dichotomy for bounded degree graph homomorphisms with nonnegative weights},
	author={Govorov, Artem and Cai, Jin-Yi and Dyer, Martin},
	journal={arXiv preprint arXiv:2002.02021},
	year={2020}
}

@book{jacobson1985basic,
	title={Basic Algebra},
	author={Jacobson, N.},
	number={v. 2},
	isbn={9780716719335},
	lccn={84025836},
	series={Basic Algebra},
	url={https://books.google.com/books?id=oNmDSAAACAAJ},
	year={1985},
	publisher={W.H. Freeman}
}

@book{lovasz2012large,
	title={Large networks and graph limits},
	author={Lov{\'a}sz, L{\'a}szl{\'o}},
	volume={60},
	year={2012},
	publisher={American Mathematical Soc.}
}

@article{bulatov2013complexity,
	title={The complexity of the counting constraint satisfaction problem},
	author={Bulatov, Andrei A.},
	journal={Journal of the ACM (JACM)},
	volume={60},
	number={5},
	pages={1--41},
	year={2013},
	publisher={ACM New York, NY, USA}
}

@inproceedings{dyer2010complexity,
	title={On the complexity of \# CSP},
	author={Dyer, Martin E. and Richerby, David M.},
	booktitle={Proceedings of the forty-second ACM symposium on Theory of computing},
	pages={725--734},
	year={2010}
}

@InProceedings{dyer2011decidability,
	author =	{Dyer, Martin E. and Richerby, David M.},
	title =	{{The \# CSP Dichotomy is Decidable}},
	booktitle =	{28th International Symposium on Theoretical Aspects of Computer Science (STACS 2011) },
	pages =	{261--272},
	year =	{2011},
	volume =	{9},
	editor =	{Thomas Schwentick and Christoph D{\"u}rr},
	publisher =	{Schloss Dagstuhl--Leibniz-Zentrum fuer Informatik},
	address =	{Dagstuhl, Germany},
	doi =		{10.4230/LIPIcs.STACS.2011.261}
}

@article{dyer2013complexity,
	volume={42},
	number={3},
	author={Dyer, Martin E. and Richerby, David M.},
	title = {An effective dichotomy for the counting constraint satisfaction problem},
	publisher = {Society for Industrial and Applied Mathematics},
	year = {2013},
	journal = {SIAM Journal on Computing},
	pages = {1245 -- 1274},
}

@article{bulatov2012complexity,
	title={The complexity of weighted and unweighted \# CSP},
	author={Bulatov, Andrei and Dyer, Martin and Goldberg, Leslie Ann and Jalsenius, Markus and Jerrum, Mark and Richerby, David},
	journal={Journal of Computer and System Sciences},
	volume={78},
	number={2},
	pages={681--688},
	year={2012},
	publisher={Elsevier}
}

@inproceedings{cai2017complexity,
	title={Complexity of counting CSP with complex weights},
	author={Cai, Jin-Yi and Chen, Xi},
	booktitle={Proceedings of the forty-fourth annual ACM symposium on Theory of computing (STOC) (2012), pp. 909--920},
	series={Journal of the ACM},
	volume={64(3)},
	pages={1--39},
	year={2017}
}

@article{cai2016nonnegative,
	title={Nonnegative weighted \# CSP: an effective complexity dichotomy},
	author={Cai, Jin-Yi and Chen, Xi and Lu, Pinyan},
	journal={SIAM Journal on Computing},
	volume={45},
	number={6},
	pages={2177--2198},
	year={2016},
	publisher={SIAM}
}

@article{valiant2008holographic,
	title={Holographic algorithms},
	author={Valiant, Leslie G},
	journal={SIAM Journal on Computing},
	volume={37},
	number={5},
	pages={1565--1594},
	year={2008},
	publisher={SIAM}
}

@article{valiant1979complexity,
	title={The complexity of computing the permanent},
	author={Valiant, Leslie G.},
	journal={Theoretical computer science},
	volume={8},
	number={2},
	pages={189--201},
	year={1979},
	publisher={Elsevier}
}

@article{kasteleyn1961statistics,
	title={The statistics of dimers on a lattice: I. The number of dimer arrangements on a quadratic lattice},
	author={Kasteleyn, Pieter W.},
	journal={Physica},
	volume={27},
	number={12},
	pages={1209--1225},
	year={1961},
	publisher={Elsevier}
}

@article{temperley1961dimer,
	title={Dimer problem in statistical mechanics-an exact result},
	author={Temperley, Harold NV. and Fisher, Michael E.},
	journal={Philosophical Magazine},
	volume={6},
	number={68},
	pages={1061--1063},
	year={1961},
	publisher={Taylor \& Francis}
}

@article{kasteleyn1963dimer,
	title={Dimer statistics and phase transitions},
	author={Kasteleyn, Pieter W.},
	journal={Journal of Mathematical Physics},
	volume={4},
	number={2},
	pages={287--293},
	year={1963},
	publisher={American Institute of Physics}
}

@article{kasteleyn1967graph,
	title={Graph theory and crystal physics},
	author={Kasteleyn, Pieter},
	journal={Graph theory and theoretical physics},
	pages={43--110},
	year={1967},
	publisher={Academic Press}
}

@inproceedings{cai2009holant,
	title={Holant problems and counting CSP},
	author={Cai, Jin-Yi and Lu, Pinyan and Xia, Mingji},
	booktitle={Proceedings of the forty-first annual ACM symposium on Theory of computing},
	pages={715--724},
	year={2009}
}

@article{cai2016complete,
	title={A complete dichotomy rises from the capture of vanishing signatures},
	author={Cai, Jin-Yi and Guo, Heng and Williams, Tyson},
	journal={SIAM Journal on Computing},
	volume={45},
	number={5},
	pages={1671--1728},
	year={2016},
	publisher={SIAM}
}

@article{cai2019holographic,
	title={Holographic algorithm with Matchgates is universal for planar \# CSP over boolean domain},
	author={Cai, Jin-Yi and Fu, Zhiguo},
	journal={SIAM Journal on Computing},
	volume={51},
	number={2},
	pages={STOC17--50},
	year={2019},
	publisher={SIAM}
}

@article{atserias2019quantum,
	title={Quantum and non-signalling graph isomorphisms},
	author={Atserias, Albert and Man{\v{c}}inska, Laura and Roberson, David E. and {\v{S}}{\'a}mal, Robert and Severini, Simone and Varvitsiotis, Antonios},
	journal={Journal of Combinatorial Theory, Series B},
	volume={136},
	pages={289--328},
	year={2019},
	publisher={Elsevier}
}

@article{lupini2020nonlocal,
	title={Nonlocal games and quantum permutation groups},
	author={Lupini, Martino and Man{\v{c}}inska, Laura and Roberson, David E.},
	journal={Journal of Functional Analysis},
	volume={279},
	number={5},
	pages={108592},
	year={2020},
	publisher={Elsevier}
}

@inproceedings{manvcinska2020quantum,
	title={Quantum isomorphism is equivalent to equality of homomorphism counts from planar graphs},
	author={Man{\v{c}}inska, Laura and Roberson, David E.},
	booktitle={2020 IEEE 61st Annual Symposium on Foundations of Computer Science (FOCS)},
	pages={661--672},
	year={2020},
	organization={IEEE}
}

@article{lovasz1967operations,
	title={Operations with structures},
	author={Lov{\'a}sz, L{\'a}szl{\'o}},
	journal={Acta Math. Acad. Sci. Hungar},
	volume={18},
	number={3-4},
	pages={321--328},
	year={1967}
}

@book{horn2012matrix,
	title={Matrix analysis},
	author={Horn, Roger A. and Johnson, Charles R.},
	year={2012},
	publisher={Cambridge university press}
}

@inproceedings{backens2017new,
	author={Miriam Backens},
	title={A New Holant Dichotomy Inspired by Quantum Computation},
	booktitle={44th International Colloquium on Automata, Languages, and Programming, {ICALP} 2017, July 10-14, 2017, Warsaw, Poland},
	series={LIPIcs},
	volume={80},
	pages={16:1--16:14},
	publisher={Schloss Dagstuhl - Leibniz-Zentrum f{\"{u}}r Informatik},
	year= {2017}
}

@inproceedings{backens2018complete,
	author={Miriam Backens},
	title={A Complete Dichotomy for Complex-Valued Holant\^{}c},
	booktitle={45th International Colloquium on Automata, Languages, and Programming, {ICALP} 2018, July 9-13, 2018, Prague, Czech Republic},
	series={LIPIcs},
	volume={107},
	pages={12:1--12:14},
	publisher={Schloss Dagstuhl - Leibniz-Zentrum f{\"{u}}r Informatik},
	year={2018}
}

@article{yang2022local,
	title={Local holographic transformations: tractability and hardness},
	author={Yang, Peng and Fu, Zhiguo},
	journal={Frontiers of Computer Science},
	volume={17},
	number={2},
	pages={1--11},
	year={2022},
	publisher={Springer}
}

@article{fu2019blockwise,
	title={On blockwise symmetric matchgate signatures and higher domain \# CSP},
	author={Fu, Zhiguo and Yang, Fengqin and Yin, Minghao},
	journal={Information and Computation},
	volume={264},
	pages={1--11},
	year={2019},
	publisher={Elsevier}
}

@article{fu2014holographic,
	title={Holographic algorithms on bases of rank 2},
	author={Fu, Zhiguo and Yang, Fengqin},
	journal={Information Processing Letters},
	volume={114},
	number={11},
	pages={585--590},
	year={2014},
	publisher={Elsevier}
}

@inproceedings{cai2023complexity,
	title={The complexity of counting planar graph homomorphisms of domain size 3},
	author={Cai, Jin-Yi and Maran, Ashwin},
	booktitle={Proceedings of the 55th Annual ACM Symposium on Theory of Computing},
	pages={1285--1297},
	year={2023}
}

@book{hell2004graphs,
	title={Graphs and homomorphisms},
	author={Hell, Pavol and Nesetril, Jaroslav},
	volume={28},
	year={2004},
	publisher={OUP Oxford}
}

@book{cai2017complexitybook,
	title={Complexity dichotomies for counting problems: Volume 1, Boolean domain},
	author={Cai, Jin-Yi and Chen, Xi},
	year={2017},
	publisher={Cambridge University Press}
}

@article{mityagin2015zero,
	title={The zero set of a real analytic function},
	author={Mityagin, Boris},
	journal={arXiv preprint arXiv:1512.07276},
	year={2015}
}

@article{cai2023planar,
	title={Planar \#CSP Equality Corresponds to Quantum Isomorphism -- A Holant Viewpoint},
	author={Cai, Jin-Yi and Young, Ben},
	journal={ACM Transactions on Computation Theory},
	year={2024},
	month={9},
	publisher={ACM New York, NY},
	volume={16},
	number={3},
	publisher = {Association for Computing Machinery},
	address = {New York, NY, USA},
	articleno={18}
}

@misc{young2022equality,
	title={Equality on all \#CSP Instances Yields Constraint Function Isomorphism via Interpolation and Intertwiners}, 
	author={Ben Young},
	year={2022},
	eprint={2211.13688},
	archivePrefix={arXiv},
	primaryClass={cs.DM}
}

@article{lovasz2006rank, 
	title = {The rank of connection matrices and the dimension of graph algebras}, 
	journal = {European Journal of Combinatorics},
	volume = {27},
	number = {6},
	pages = {962-970},
	year = {2006},
	author = {László Lovász}
}

@article{chan2024quantum,
	title = {Quantum isomorphism of graphs from association schemes},
	journal = {Journal of Combinatorial Theory, Series B},
	volume = {164},
	pages = {340-363},
	year = {2024},
	author = {Ada Chan and William J. Martin},
}

@article{slofstra2019set, 
	title={The Set of Quantum Correlations is not Closed}, 
	volume={7}, 
	journal={Forum of Mathematics, Pi}, 
	author={Slofstra, William}, 
	year={2019}
}

@article{cai2024polynomial,
	title={Polynomial and analytic methods for classifying complexity of planar graph homomorphisms},
	author={Cai, Jin-Yi and Maran, Ashwin},
	journal={arXiv preprint arXiv:2412.17122},
	year={2024}
}

@article{lovasz2009contractors,
	title={Contractors and connectors of graph algebras},
	author={Lov{\'a}sz, L{\'a}szl{\'o} and Szegedy, Bal{\'a}zs},
	journal={Journal of Graph Theory},
	volume={60},
	number={1},
	pages={11--30},
	year={2009},
	publisher={Wiley Online Library}
}

@article{dvorak2010recognizing,
	title={On recognizing graphs by numbers of homomorphisms},
	author={Dvo{\v{r}}{\'a}k, Zden{\v{e}}k},
	journal={Journal of Graph Theory},
	volume={64},
	number={4},
	pages={330--342},
	year={2010},
	publisher={Wiley Online Library}
}
